%% file: LNBSLSA.tex
\documentclass[journal,twoside,web]{ieeecolor}

\usepackage{etoolbox}
\newtoggle{arxiv}
\toggletrue{arxiv} % arxiv version, includes appendices
\usepackage{generic}
\usepackage{cite}
\usepackage{amsmath,amssymb,amsfonts}
\usepackage{algorithmic}
\usepackage{graphicx}
\usepackage{textcomp}
\def\BibTeX{{\rm B\kern-.05em{\sc i\kern-.025em b}\kern-.08em
    T\kern-.1667em\lower.7ex\hbox{E}\kern-.125emX}}
\usepackage{amsmath,amssymb}
\newcommand{\E}{\mathbb{E}} 
\renewcommand{\P}{\mathbb{P}} 
\newcommand{\R}{\mathbb{R}} 
\newcommand{\Z}{\mathbb{Z}} 
\newcommand{\N}{\mathbb{N}} 
\newcommand{\trans}{\mathsf{T}}
\newcommand{\Var}{\textrm{Var}}
\newcommand{\Cov}{\textrm{Cov}}
\newcommand{\D}{\mathcal{D}}
\newcommand{\T}{\mathcal{T}}
\renewcommand{\i}{\mathbf{i}}
\newtheorem{theorem}{Theorem}
\newtheorem{definition}{Definition}
\newtheorem{lemma}{Lemma}
\newtheorem{corollary}{Corollary}
\newtheorem{proposition}{Proposition}
\DeclareMathOperator*{\argmin}{arg\,min}

\usepackage[ruled,noend]{algorithm2e}

\SetAlFnt{\small}
\SetAlCapFnt{\small}
\SetAlCapNameFnt{\small}
\SetAlCapHSkip{0pt}
\IncMargin{-\parindent}
\SetKw{KwBreak}{break}
\SetKw{KwAnd}{and}
\SetKw{KwReturn}{return}

\usepackage{bbm,subcaption}
\usepackage{mathtools}
\mathtoolsset{showonlyrefs}
\usepackage{mathrsfs}
\usepackage{stackrel}
\let\labelindent\relax
\usepackage{enumitem}

\iftoggle{arxiv}{
\newcommand{\new}[1]{{\color{black}#1}}
}{
\newcommand{\new}[1]{{\color{blue}#1}}
}

\begin{document}
\title{Local non-Bayesian social learning with stubborn agents}
\author{Daniel Vial, Vijay Subramanian, \IEEEmembership{Senior Member, IEEE}
\thanks{We are grateful for financial support from the NSF (grants  EPCN:1603861 and CIF:AF:2008130) and LGE Inc.\ via Mcity. D.\ Vial is with the University of Texas, Austin, TX (email: dvial@utexas.edu). V.\ Subramanian is with the University of Michigan, Ann Arbor, MI (email: vsubram@umich.edu).}}

\maketitle

\begin{abstract}
We study a social learning model in which agents iteratively update their beliefs about the true state of the world using private signals and the beliefs of other agents in a non-Bayesian manner. Some agents are stubborn, meaning they attempt to convince others of an erroneous true state (modeling fake news). We show that while agents learn the true state on short timescales, they ``forget'' it and believe the erroneous state to be true on longer timescales. Using these results, we devise strategies for seeding stubborn agents so as to disrupt learning, which outperform intuitive heuristics and give novel insights regarding vulnerabilities in social learning.
\end{abstract}

\input{introduction}

\input{model}

\input{results}

\input{adversarial}

\input{related}

\bibliographystyle{IEEEtran}
\bibliography{references}  

\iftoggle{arxiv}{

\begin{appendices}

\onecolumn
\allowdisplaybreaks

\input{specialCase}

\input{mainProofOutline}

\input{adversarialDetails}

\input{mainProofDetails}

\input{proofBeliefs}

\end{appendices}

}

\end{document}

%% file: introduction.tex
\section{Introduction} \label{secIntro}

With the rise of social networks, people increasingly receive news through non-traditional sources. One recent study shows that two-thirds of American adults have gotten news through social media \cite{shearer2017news}. Such news sources are fundamentally different than traditional ones like print media and television, in the sense that social media users read and discuss news on the same platform. As a consequence, users turning to these platforms for news receive information not only from major publications but from others users as well; in the words of \cite{allcott2017social}, a user ``with no track record or reputation can in some cases reach as many readers as Fox News, CNN, or the New York Times.'' This phenomenon famously reared its head during the 2016 United States presidential election when fake news stories were shared tens of millions of times \cite{allcott2017social}, and it has remained a critical issue in 2020 \cite{vanity}.

In this paper, we study a mathematical model describing this situation. The model includes a set \new{of} agents attempting to learn the true state of the world (e.g.\ which of two candidates is better suited for office). \new{Each agent iteratively updates its belief (i.e.\ its distribution over possible states) in a manner similar to the non-Bayesian social learning model of \cite{jadbabaie2012non} using information from three sources.} First, each agent receives noisy observations of the true state, modeling e.g.\ news stories. Second, each agent observes the beliefs of a subset of other agents, modeling e.g.\ discussions with other social media users. Third, each agent may observe the beliefs of \textit{stubborn agents} or \textit{bots} who aim to persuade others of an erroneous true state, modeling e.g.\ users spreading fake news.\footnote{The term \textit{stubborn agents} has been used in the literature to describe such agents; the term \textit{bots} is used in reference to automated social media accounts spreading fake news while masquerading as real users \cite{shao2018spread}.} \new{This process continues iteratively until a finite learning horizon.}

Under this model, two competing forces emerge as the learning horizon grows. On the one hand, agents receive more observations of the true state, which help them learn. On the other hand, the beliefs of the bots gradually propagate through the system, suggesting that agents become increasingly exposed to bots and thus less likely to learn. Hence, while the horizon clearly affects the learning outcome, the nature of this effect -- namely, whether learning becomes more or less likely as the horizon grows -- is less clear.

This effect of the learning horizon has often been ignored in works with models similar to ours. For example, our model is nearly identical to that in the empirical work \cite{azzimonti2018social}, in which the authors show that polarized opinions can arise when there are two types of bots with diametrically opposed viewpoints. However, the experiments in \cite{azzimonti2018social} simply fix a large learning horizon and do not consider the effect of varying it. Models similar to ours have also been treated analytically \new{in e.g.\ \cite{jadbabaie2012non,golub2010naive,acemoglu2010spread,lalitha2014social}, but these works consider infinite horizons and/or cooperative settings (i.e.\ no stubborn agents). See Section \ref{secRelated} for details on these (and other) works.}

In the first part of the paper (see Section \ref{secResults}), we argue that \new{\textit{the learning horizon plays a prominent role when stubborn agents are present and should not be ignored}}. In particular, we show that the learning outcome depends on the relationship between the horizon $T_n$ and a quantity $p_n$ that describes the ``density'' of bots in the system, where both quantities may vary with the number of agents $n$. Mathematically, letting $\theta \in (0,1)$ denote the true state and $\theta_{T_n}(i^*)$ \new{the mean of the belief (hereafter, the \textit{estimate})} for a uniformly random agent $i^*$ at the horizon $T_n$, we show (see Theorem \ref{thmMain})\footnote{The theorem also addresses the case $\lim_{n \rightarrow \infty} T_n(1-p_n) \in (0,\infty)$.}
\begin{equation} \label{eqResultPreview}
\theta_{T_n}(i^*) \xrightarrow[n \rightarrow \infty]{\P} \begin{cases} \theta , & T_n(1-p_n) \xrightarrow[n \rightarrow \infty]{} 0 \\  %\theta (1-e^{-c \eta})/(c \eta) , & T_n(1-p_n) \xrightarrow[n \rightarrow \infty]{} c \in (0,\infty) \\ 
0 , & T_n(1-p_n) \xrightarrow[n \rightarrow \infty]{} \infty \end{cases} .
\end{equation}
Here $p_n$ is smaller when more bots are present and 0 is the erroneous true state promoted by the bots. Hence, in words, \eqref{eqResultPreview} says the following: if there are sufficiently few bots, in the sense that $T_n(1-p_n) \rightarrow 0$, $i^*$ learns the true state; if there are sufficiently many bots, in the sense that $T_n(1-p_n) \rightarrow \infty$, $i^*$ adopts the extreme estimate 0 promoted by the bots. \new{Additionally, we show the belief of $i^*$ converges to a Dirac measure in a certain sense (see Corollary \ref{corBeliefs}).}

We note the result in \eqref{eqResultPreview} assumes a particular random graph model for the social network connecting agents and bots (a modification of the so-called \textit{directed configuration model}). For such models, \textit{phase transitions} -- wherein small changes to model parameters lead to starkly different behaviors -- are often observed. In this case, assuming $T_n = (1-p_n)^{-k}$ for some $k > 0$, and also assuming $p_n \rightarrow 1$, the learning outcome suddenly drops from $\theta$ to $0$ as $k$ changes from e.g. $0.99$ to $1.01$. Put differently, agents initially (at time $(1-p_n)^{-0.99}$) learn the true state, then suddenly (at time $(1-p_n)^{-1.01}$) ``forget'' the true state and adopt the extreme estimate $0$. Hence, we show the horizon can lead to drastically different outcomes. We also note proving \eqref{eqResultPreview} involves analyzing hitting probabilities for random walks on random graphs with absorbing states (bots in our setting), which may be of independent interest.

In the second part of the paper (see Section \ref{secAdversarial}), we study a setting in which an adversary chooses how many bots to connect to each agent, subject to a budget constraint. The adversary would like to minimize $\theta_{T_n}(i^*)$ (i.e.\ to convince agents of the erroneous state $0$), but this quantity depends on the graph topology, which is not publicly available for social networks like Twitter. Hence, motivated by \eqref{eqResultPreview}, we formulate the adversary's problem as minimizing $p_n$, which only depends on the \textit{degrees} in the graph -- e.g.\ number of followers on Twitter, which is publicly available. We clarify that $\theta_{T_n}(i^*)$ is monotone in $p_n$ only as $n \rightarrow \infty$ for the random graph of Section \ref{secResults} (see Theorem \ref{thmMain}). Thus, we use $p_n$ as a tractable (albeit nonrigorous) surrogate for the true objective function $\theta_{T_n}(i^*)$, and we show empirically that these quantities are closely correlated for real social networks (see Figure \ref{fig_objVsFinBel}). Alternatively, given a target $\theta_{T_n}(i^*)$, we can minimize the horizon $T_n$ when this target estimate is reached. However, we view $T_n$ as fixed and thus do not pursue this dual problem.

Minimizing $p_n$ amounts to solving an integer program, which can be done in polynomial time owing to the structure of $p_n$. However, the computational complexity is $\Omega(n^2)$, which is infeasible for social networks like Twitter. Thus, we propose a randomized approximation algorithm that runs in time $n \log n$ and that produces a constant-fraction approximation of the optimal solution with high probability (see Theorem \ref{thmConstApprox}). Moreover, whereas the logic of the optimal solution is somewhat opaque, the form of our approximate solution offers the interpretation that \textit{successful adversaries carefully balance agents' influence and susceptibility to influence}. For a social network like Twitter, this means targeting users with many followers (i.e.\ influential users) who follow very few users themselves, so that fake news will occupy a greater portion of the targeted users' feeds. While somewhat intuitive, the precise form of the randomized scheme is far from obvious. Furthermore, empirical results show that our scheme disrupts learning to a larger extent than schemes that more obviously balance influence and susceptibility. Thus, we believe our analysis provides new insights into vulnerabilities of news sharing platforms and non-Bayesian social learning models.

The paper is organized as follows. In Section \ref{secLearnModel}, we define our learning model. Sections \ref{secResults} and \ref{secAdversarial} follow the outline above. We discuss related work in Section \ref{secRelated}. %
\iftoggle{arxiv}{%
}{%
Proof details are deferred to the full version of the paper, \cite{vial2019local_arxiv}. Preliminary versions of the paper appeared in abstracts \cite{vial2019local_ec,vial2019local_allerton}. %
}%

\textit{Notational conventions:} The following notation is used frequently. For $k \in \N$, we let $[k] = \{1,\ldots,k\}$, and for $k, k' \in \N$ we let $[k] + k' = k' + [k] = \{1+k',\ldots,k+k'\}$. All vectors are treated as row vectors. We let $e_i$ denote the vector with 1 in the $i$-th position and 0 elsewhere. We denote the set of nonnegative integers by $\N_0 = \N \cup \{0\}$. We use $1(A)$ for the indicator function, i.e.\ $1(A) = 1$ if $A$ is true and 0 otherwise. All random variables are defined on a common probability space $(\Omega,\mathcal{F},\P)$, with $\E[\cdot] = \int_{\Omega} \cdot\ d \P$ denoting expectation, $\xrightarrow{\P}$ denoting convergence in probability, and $a.s.$ meaning $\P-$almost surely.

%% file: model.tex
\section{Learning model} \label{secLearnModel}

We begin by defining the model of social learning studied throughout the paper. The basic ingredients are (1) a true state of the world, (2) a social network connecting two sets of nodes, some who aim to learn the true state and some who wish to persuade others of an erroneous true state, and (3) a learning horizon. We discuss each in turn.

The true state of the world is a constant $\theta \in (0,1)$. For example, in an election between candidates representing two political parties (say, Party 1 and Party 2), \new{$\theta \approx 0$ and $\theta \approx 1$ means the Party 1 and 2 candidates are superior, respectively}. We emphasize that $\theta$ is a deterministic constant and depends neither on time, nor on the number of nodes in the system.

A directed graph $G = (A \cup B,E)$ connects disjoint sets of nodes $A$ and $B$. We refer to elements of $A$ as \textit{regular agents}, or simply \textit{agents}, and elements of $B$ as \textit{stubborn agents} or \textit{bots}. While agents attempt to learn the true state $\theta$, bots aim to disrupt this learning and convince agents that the true state is instead 0. In the election example, agents represent voters who study the two candidates to learn which is superior, while bots are loyal to Party 1 and aim to convince agents that the corresponding candidate is superior (despite possible evidence to the contrary). Edges in the graph represent connections in a social network over which nodes share beleifs in a manner that will be described shortly. \new{An edge $j \rightarrow i$ means that $i$ observes $j$'s belief. Let $N_{in}(i) = \{ j \in A \cup B : j \rightarrow i \in E \}$ and $d_{in}(i) = |N_{in}(i)|$; we assume $N_{in}(i) \neq \emptyset$.}

Agents and bots share beliefs until a learning horizon $T \in \N$. We will allow the horizon to depend on the number of agents $n \triangleq |A|$ and will thus denote it by $T_n$ at times. In the election example, $T$ represents the duration of the election, i.e.\ the number of time units that agents can learn about the candidates and that bots can attempt to convince agents of the superiority of the Party 1 candidate.

\new{Given these basic ingredients, we can define the learning process. At time $t = 0$, agent $i \in A$ has a $\text{Beta}(\alpha_0(i),\beta_0(i))$ belief, where $\alpha_0(i) \in ( 0 , \bar{\alpha} ]$ and $\beta_0(i) \in ( 0 , \bar{\beta} ]$ for some $\bar{\alpha},\bar{\beta} \in (0,\infty)$ that do not depend on $n$. For each $t \in [T]$, $i$ receives the signal $s_t(i) \sim \text{Bernoulli}(\theta)$. In the absence of a network, the Bayesian approach dictates that $i$ update its parameters to $\alpha_t(i) = \alpha_{t-1}(i) + s_t(i)$ and $\beta_t(i) = \beta_{t-1}(i) + (1-s_t(i))$ and its belief to $\mu_t(i) = \text{Beta}(\alpha_t(i),\beta_t(i))$, namely, for any (measurable) $\mathcal{A} \subset [0,1]$,
\begin{equation}
\mu_t(i) ( \mathcal{A} ) \propto \int_{ x \in \mathcal{A} } x^{\alpha_t(i)-1} (1-x)^{\beta_t(i)-1} dx .
\end{equation}
In our running example, $\alpha_t(i)$ and $\beta_t(i)$ represent the number of news stories favorable to respective parties that $i$ has read during the election, plus some prior parameters $\alpha_0(i)$ and $\beta_0(i)$ that account for $i$'s biases from before the election. As $t$ grows, the belief $\mu_t(i)$ converges to a Dirac measure on its mean $\theta_t(i) = \alpha_t(i) / (\alpha_t(i) + \beta_t(i))$; intuitively, $i$ becomes increasingly confident that the true state is the fraction of stories favorable to a certain party.

In the presence of a network, we proceed in the same manner, except the parameters are updated as follows:
\begin{gather} \label{eqParamUpdateInitial} 
\alpha_t(i) = (1-\eta) ( \alpha_{t-1}(i) + s_t(i) ) +  \sum_{j \in N_{in}(i)} \frac{\eta \alpha_{t-1}(j) }{d_{in}(i)} , \\
\beta_t(i) = (1-\eta) ( \beta_{t-1}(i) + 1-s_t(i) ) + \sum_{j \in N_{in}(i)} \frac{\eta \beta_{t-1}(j) }{d_{in}(i)}  ,
\end{gather}
where $\eta \in (0,1)$. Intuitively, $i$ reads the news and calculates its favorability of the parties as before, then discusses with its neighbors to update its favoribility. Mathematically, $i$ performs a Bayesian parameter update and then averages parameters. \cite{azzimonti2018social} uses the same update, whereas agents in \cite{jadbabaie2012non} do Bayesian \textit{belief} updates and then average \textit{beliefs}. Our update also resembles the deGroot model \cite{degroot1974reaching}, where there are no signals and estimates are averaged across neighbors. See Section \ref{secRelated}.

Finally, we specify bot behavior. For $i \in B$, we set $N_{in}(i) = \{ i \}$, $\alpha_0(i) = 0$, $\beta_0(i) = \bar{\beta}$, and $s_t(i) = 0\ \forall\ t \in [T]$, then iteratively define $\{ \alpha_t(i), \beta_t(i) \}_{t=1}^T$ via \eqref{eqParamUpdateInitial}. More explicitly, a simple inductive proof shows
\begin{equation} \label{eqParamBots}
\alpha_t(i) = 0, \quad \beta_t(i) = \bar{\beta} + (1-\eta)t \quad \forall\ t \in [T] .
\end{equation}
In our running example, $\alpha_0(i) = 0$, $\beta_0(i) = \bar{\beta}$, and $s_t(i) = 0$ means $i$'s prior parameters and signals are maximally biased toward Party 1. Furthermore, we can interpret $N_{in}(i) = \{ i \}$ as bots being ``echo chambers'' who only listen to themselves. Finally, note that since all bots $i \in B$ have the same behavior, we assume (without loss of generality) that the outgoing neighbor set of $i \in B$ is $N_{out}(i) = \{ i, i' \}$ for some $i' \in A$, i.e.\ in addition to its self-loop, each bot has a single outgoing neighbor from the agent set.}

%% file: results.tex
\section{Learning outcome} \label{secResults}

\new{To begin our analysis of the learning outcome, we show when all agents are (pathwise) connected to bots, their beliefs converge to those of the bots. Formally, for $p \geq 1$, let
\begin{equation}
  W_p(\mu,\nu) =  \inf_{(X,Y): X \sim \mu , Y \sim \nu} ( \E |X-Y|^p )^{1/p}
\end{equation}
denote the $p$-Wasserstein distance for probability measures $\mu$ and $\nu$, where $X \sim \mu , Y \sim \nu$ means $X$ and $Y$ have respective marginals $\mu$ and $\nu$. For $x \in [0,1]$, let $\delta_x$ denote the Dirac measure $\delta_x(\mathcal{A}) = 1(x \in \mathcal{A})$ for measurable $\mathcal{A} \subset [0,1]$. We then have the following (see % 
\iftoggle{arxiv}{%
Appendix \ref{appProofBeliefs} %
}%
{%
\cite[Appendix V]{vial2019local_arxiv} %
}%
for a proof).
\begin{proposition} \label{propBeliefToZero}
Suppose that for any $i \in A$, there exists $l \in \N$ and $( i_\tau )_{\tau=0}^{l} \in (A \cup B)^{l+1}$ such that $i_0 = i$, $i_{\tau-1} \rightarrow i_\tau\in E\ \forall\ \tau \in [l]$, and $i_l \in B$. Then for any $i \in A$ and $p \geq 1$,  $\lim_{t \rightarrow \infty} \theta_t(i) = \lim_{t \rightarrow \infty} W_p ( \mu_t(i) , \delta_0 ) = 0\ a.s.$
\end{proposition}

Hence, for a large enough horizon, estimates and beliefs become arbitrarily close to zero. A natural follow-up question is how such a horizon scales -- and in which graph parameters -- for a sequence of graphs $\{ G_n \}_{n=1}^\infty$. In this section, we address this question for a particular random graph model, succinctly described as the \textit{directed configuration model} (DCM) plus bots. The DCM constructs a graph with prespecified degrees, which, conditioned on being simple (i.e.\ having no self-loops or multi-edges), is uniformly distributed among (simple) graphs of those degrees \cite[Proposition 7.15]{van2016random}. This is an appealing property for deGroot-like learning models such as ours, because in the deGroot model for undirected graphs, asymptotic estimates depend only on the degrees and the initial beliefs. Thus, loosely speaking, our analysis is ``average-case'' over relevant graphs. Furthermore, we will show the graph parameters that dictate learning for the DCM are tractable, which we exploit in Section \ref{secAdversarial} for general graphs.

Having motivated our study of the DCM, we define it in Section \ref{secGraphModel}, present our main result for the DCM in Section \ref{secMainResult}, and discuss our assumptions in Section \ref{secMainDiscuss}.
}

\subsection{Graph model} \label{secGraphModel}

To begin, we realize a sequence $\{ d_{out}(i) , d_{in}^A(i) , d_{in}^B(i) \}_{i \in A}$ called the \textit{degree sequence} from some distribution; here we let $A = [n]$. In the construction described next, $i \in A$ will have $d_{out}(i)$ outgoing neighbors ($i$ will be observed by $d_{out}(i)$ other agents), $d_{in}^A(i)$ incoming neighbors from the $A$ ($i$ will observe $d_{in}^A(i)$ agents), and $d_{in}^B(i)$ incoming neighbors from $B$ ($i$ will observe $d_{in}^B(i)$ bots). Here the total in-degree of $i$ is $d_{in}(i) = d_{in}^A(i) + d_{in}^B(i)$ (as used in \eqref{eqParamBots}). We assume
\begin{gather} \label{eqDegSeqBasicAss}
d_{out}(i) , d_{in}^A(i) \in \N, \quad d_{in}^B(i) \in \N_0\ \quad \forall\ i \in A , \\ 
 \sum_{i \in A} d_{out}(i) = \sum_{i \in A} d_{in}^A(i) .
\end{gather}
In words, the first condition says $i$ is observed by and observes at least one agent, and may observe one or more bots. The second condition says sum out-degree must equal sum in-degree in the agent sub-graph; this will be necessary to construct a graph with the given degrees. Finally, it will be convenient to define the degree vector of $i \in A$ as
\begin{equation} \label{eqDegVector}
d(i) = ( d_{out}(i) , d_{in}^A(i) , d_{in}^B(i) ) .
\end{equation}

After realizing the degree sequence, we begin the graph construction.\footnote{This construction is presented more formally in %
\iftoggle{arxiv}{%
Appendix \ref{appBranchApproxProofOutline}%
}%
{%
\cite[Appendix II-A]{vial2019local_arxiv}%
}
.} %
First, we attach $d_{out}(i)$ outgoing half-edges, $d_{in}^A(i)$ incoming half-edges labeled $A$, and $d_{in}^B(i)$ incoming half-edges labeled $B$, to each $i \in A$; we will refer to these half-edges as \textit{outstubs}, \textit{$A$-instubs}, and \textit{$B$-instubs}, respectively. \new{Let $O_A$ denote the set of all agents' outstubs.} We then pair each outstub in $O_A$ with an $A$-instub to form edges between agents in a breadth-first-search fashion that proceeds as follows:
\begin{itemize}[leftmargin=10pt,itemsep=0pt,topsep=0pt]
\item Sample $i^*$ from $A$ uniformly. For each the $d_{in}^A(i^*)$ $A$-instubs attached to $i^*$, sample an outstub uniformly from $O_A$ (resampling if the sampled outstub has already been paired), and connect the instub and outstub to form an edge from some agent to $i^*$.
\item Let $A_1 = \{ i \in A \setminus \{ i^* \} : \textrm{an outstub of $i$ was paired with}$ $\textrm{an $A$-instub of $i^*$} \}$. For each $i \in A_1$, pair the $d_{in}^A(i)$ $A$-instubs attached to $i$ in the same manner the $A$-instubs of $i^*$ were paired in the previous step.
\item Continue iteratively until all $A$-instubs have been paired. In particular, during the $l$-th iteration, we pair all $A$-instubs attached to $A_l$, the agents at geodesic distance $l$ from $i^*$.
\end{itemize}
\new{The procedure above yields the standard DCM, plus} unpaired $B$-instubs attached to some agents. \new{To pair these instubs}, we define $B = n + \big[ \sum_{i \in A} d_{in}^B(i) \big]$ to be the set of bots (hence, the node set is $A \cup B = \big[ n + \sum_{i \in A} d_{in}^B(i) \big]$). %\textcolor{red}{(Is $B$ correct? I think so -- by notational conventions at end of intro, $n + [ \sum_{i \in A} d_{in}^B(i) ] = \{n+1,\ldots,n+ \sum_{i \in A} d_{in}^B(i) \}$. Maybe it was confusing because the intro used $n$ in the other place, i.e. $[n] + k$; I updated this.)}. 
To each $i \in B$ we add a single self-loop and a single unpaired outstub (as described at the end of Section \ref{secLearnModel}). This yields $\sum_{i \in A} d_{in}^B(i)$ unpaired outstubs attached to bots. Finally, we pair these outstubs arbitrarily with the $\sum_{i \in A} d_{in}^B(i)$ unpaired $B$-instubs from above to form edges from bots to agents (the pairing can be arbitrary since all bots behave the same).

We note that the pairing of $A$-instubs with outstubs from $O_A$ did not prohibit \new{multi-edges, so} the set of edges $E$ formed will in general be a multi-set. For this reason, we replace the summation in the $\alpha_t(i)$ update \eqref{eqParamUpdateInitial} with
\begin{gather} 
\sum_{j \in A \cup B} \frac{\eta | \{ j' \rightarrow i' \in E : j' = j, i' = i \} | \alpha_{t-1}(j) }{d_{in}(i)} , 
\end{gather}
and analogously for the $\beta_t(i)$ update, i.e.\ we weigh the parameters of $i$'s neighbors proportional to the number of edges pointing to $i$. We also note that if $d_{in}^B(i) = 0\ \forall\ i \in A$, the construction above reduces to the standard DCM.

Our results will require assumptions on the degree sequence $\{d(i) \}_{i \in A}$, where (we recall) $d(i)$ is the degree vector of $i$ (see \eqref{eqDegVector}). First, we define $f_n^* , f_n : \N \times \N \times \N_0 \rightarrow [0,1]$ by
\begin{align} \label{eqEmpDist}
f_n^*(i,j,k) &  =  \sum_{a = 1}^n \frac{1 ( d(a) = (i,j,k) )}{n} , \\
f_n(i,j,k) &  = \sum_{a = 1}^n \frac{ d_{out}(a) 1 ( d(a) = (i,j,k) )  }{ \sum_{a' = 1}^n d_{out}(a') }   .
\end{align}
In words, $f_n^*$ and $f_n$ are the degree distributions of agents sampled uniformly and sampled proportional to out-degree, respectively. Note that, since the first agent $i^*$ added to the graph is sampled uniformly from $A$, the degrees of $i^*$ are distributed as $f_n^*$. Furthermore, recall that, to pair $A$-instubs, we sample outstubs uniformly from $O_A$, resampling if the sampled outstub is already paired. It follows that, each time we add a new agent to the graph (besides $i^*$), its degrees are distributed as $f_n$. We also note that, because the degree sequence is random, these distributions are random as well. From these random distributions, we define the random variables
\begin{align} \label{eqWalkRVsMain}
\tilde{p}_n^* &  = \sum_{j \in \N, k \in \N_0} \frac{j}{j+k} \sum_{i \in \N} f_n^*(i,j,k) , \\
 \tilde{p}_n &  = \sum_{j \in \N, k \in \N_0} \frac{j}{j+k} \sum_{i \in \N} f_n(i,j,k)  , \\
 \tilde{q}_n &  = \sum_{j \in \N, k \in \N_0} \frac{j}{j+k} \frac{1}{j+k} \sum_{i \in \N} f_n(i,j,k) . 
\end{align}
Following the discussion above, $\tilde{p}_n^*$ is the expected value (conditioned on the degree sequence) of the ratio of $A$-instubs to total instubs for $i^*$; $\tilde{p}_n$ is the expected value of this same ratio, but for new agents added to the graph. The interpretation of $\tilde{q}_n$ is similar. At the end of Section \ref{secMainResult}, we discuss in more detail why these random variables arise in our analysis.

We now state four assumptions, which we discuss in detail in Section~\ref{secMainDiscuss}. Two of these require the degree sequence to be well-behaved (with high probability) -- specifically, \ref{assGraphDegSeq} requires certain moments of the degree sequence to be finite, while \ref{assBranchDegSeq} requires $\{ \tilde{p}_n \}_{n \in \N}$ to be close to a deterministic sequence $\{ p_n \}_{n \in \N}$. The other assumptions, \ref{assGraphHorizon} and \ref{assBranchHorizon}, impose maximum and minimum rates of growth for the learning horizon $T_n$. In particular, $T_n$ must be finite for each finite $n$ but grow to infinity with $n$.
\begin{enumerate}[label=A\arabic*]
\item  \label{assGraphDegSeq} $\lim_{n \rightarrow \infty} \P(\Omega_{n,1}) = 1$, where, for some $\nu_1, \nu_2, \nu_3 , \gamma > 0$ independent of $n$ such that $\nu_3 > \nu_1$,\footnote{The assumption $\nu_3 > \nu_1$ only eliminates the trivial case of a line graph; see Section \ref{secMainDiscuss} for details.}
\begin{align}
\Omega_{n,1} & = \Big\{ \Big| \frac{ \sum_{i=1}^n  d_{out}(i) }{n} - \nu_1 \Big| < n^{-\gamma} \Big\}  \\
& \quad\quad \cap \Big\{ \Big| \frac{ \sum_{i=1}^n  d_{out}(i)^2 }{n} - \nu_2 \Big| < n^{-\gamma} \Big\} \\
& \quad\quad \cap \Big\{ \Big| \frac{ \sum_{i=1}^n  d_{out}(i) d_{in}^A(i) }{n} - \nu_3 \Big| < n^{-\gamma} \Big\} .
\end{align}
\item  \label{assGraphHorizon} $\exists\ N \in \N$ and $\zeta \in (0,1/2) $ independent of $n$ s.t.\ $T_n \leq  \zeta \log ( n ) /  \log ( \nu_3 / \nu_1  )\ \forall\ n \geq N$.
\item \label{assBranchDegSeq} $\lim_{n \rightarrow \infty} \P(\Omega_{n,2}) = 1$, where, for some $p_n \in [0,1]$ s.t.\ $\lim_{n \rightarrow \infty} p_n = p \in [0,1]$, some $0 \leq \delta_n = o(1/T_n)$, and some $\xi \in (0,1)$ independent of $n$,
\begin{equation}
\Omega_{n,2} = \left\{ | p_n - \tilde{p}_n | < \delta_n ,  \tilde{p}_n^* \geq \tilde{p}_n , \tilde{q}_n < 1 - \xi \right\} .
\end{equation} 
\item \label{assBranchHorizon} $\lim_{n \rightarrow \infty} T_n = \infty$. 
\end{enumerate}

\subsection{Main result} \label{secMainResult}

We can now present Theorem \ref{thmMain}. The theorem states that the estimate at time $T_n$ of a uniformly random agent converges in probability as $n \rightarrow \infty$. As discussed in the introduction, the limit depends on the relative asymptotics of the time horizon $T_n$ and the quantity $p_n$ defined in \ref{assBranchDegSeq}. For example, this limit is $\theta$ when $T_n(1-p_n) \rightarrow 0$; note that $T_n(1-p_n) \rightarrow 0$ requires $p_n$ to quickly approach 1 (since $T_n \rightarrow \infty$ by \ref{assBranchHorizon}), which by \ref{assBranchDegSeq} and \eqref{eqWalkRVsMain} suggests the number of bots is small. Hence, $i^*$ learns the true state when there are sufficiently few bots. (The other cases can be interpreted similarly.)
\begin{theorem} \label{thmMain}
Assume that $G$ is the DCM and that \ref{assGraphDegSeq}, \ref{assGraphHorizon}, \ref{assBranchDegSeq}, and \ref{assBranchHorizon} hold. Then for $i^* \sim A$ uniformly,
\begin{equation} 
\theta_{T_n}(i^*) \xrightarrow[n \rightarrow \infty]{\P} \begin{cases} \theta , & T_n(1-p_n) \rightarrow 0 \\  \frac{\theta (1-e^{-c \eta}) }{c \eta} , & T_n(1-p_n) \rightarrow c \in (0,\infty) \\ 0 , & T_n(1-p_n) \rightarrow \infty \end{cases} .
\end{equation}
\end{theorem}

Before discussing the proof, we make several observations:
\begin{itemize}[leftmargin=10pt,itemsep=0pt,topsep=0pt]
\item Suppose $p_n$ is fixed and consider varying $T_n$. To be concrete, let $p_n = 1 - ( \log n )^{-1/2}$ and define $T_{n,1} = ( \log n )^{1/4}$ and $T_{n,2} = ( \log n )^{3/4}$ (note $T_{n,1}, T_{n,2}$ satisfy \ref{assGraphHorizon}, \ref{assBranchHorizon}). Then $T_{n,1}(1-p_n) \rightarrow 0$ and $T_{n,2}(1-p_n) \rightarrow \infty$, so by Theorem \ref{thmMain}, the estimate of $i^*$ converges to $\theta$ at time $T_{n,1}$ and to 0 at time $T_{n,2}$. In words, $i^*$ initially (at time $(\log n)^{1/4}$) learns the state of the world, then later (at time $(\log n)^{3/4}$) forgets it and adopts the bot estimates.
\item Alternatively, suppose $T_n$ is fixed and consider varying $p_n$. For example, let $p_n = 1 - c / T_n$ for some $c \in (0,\infty)$. Here smaller $c$ implies fewer bots, and Theorem \ref{thmMain} says the limiting estimate of $i^*$ is a decreasing convex function of $c$. One interpretation is that, if an adversary deploys bots in hopes of driving agent estimates to 0, the marginal benefit of deploying additional bots is smaller when $c$ is larger, i.e.\ the adversary experiences ``diminishing returns''. It is also worth noting that, since $(1-e^{-c \eta})/(c \eta) \rightarrow 1$ as $c \rightarrow 0$ and $(1-e^{-c \eta})/(c \eta) \rightarrow 0$ as $c \rightarrow \infty$, the limiting estimate of $i^*$ is continuous as a function of $c$.
\item If $T_n(1-p_n) \rightarrow c \in (0,\infty)$, consider the limiting estimate of $i^*$ as a function of $\eta$. By Theorem \ref{thmMain}, this estimate tends to $\theta$ as $\eta \rightarrow 0$ and tends to $(1-e^{-c})/c$ as $\eta \rightarrow 1$. This is expected from \eqref{eqParamUpdateInitial}: when $\eta = 0$, agents ignore the network (and thus avoid exposure to biased bot beliefs) and form estimates based only on unbiased signals; when $\eta = 1$, the opposite is true.
\item If $p_n \rightarrow p < 1$, we must have $T_n(1-p_n) \rightarrow \infty$ (since $T_n \rightarrow \infty$ by \ref{assBranchHorizon}), and the estimate of $i^*$ tends to 0 by Theorem \ref{thmMain}. Loosely speaking, this says that a necessary condition for learning is that the bots vanish asymptotically (in the sense that $p_n \rightarrow 1$).
\item In fact, in the case $p_n \not \rightarrow 1$, a stronger result holds: the set of agents $i$ for which $\theta_{T_n}(i) \not \rightarrow 0$ vanishes relative to $n$. See %
\iftoggle{arxiv}{%
Appendix \ref{secSecondResult} %
}{%
\cite[Appendix I]{vial2019local_arxiv} %
}%
for details.
\end{itemize}

The proof of Theorem \ref{thmMain} is lengthy and deferred to %
\iftoggle{arxiv}{%
Appendices \ref{appMainProofOutline} and \ref{appMainProofDetails}, where Appendix~\ref{appMainProofOutline} lays out the structure of the proof. %
}{%
\cite[Appendices II and IV]{vial2019local_arxiv},  where \cite[Appendix II]{vial2019local_arxiv} lays out the structure of the proof. %
}%
However, we next present a short argument to illustrate the fundamental reason why the three cases of the limiting estimate arise in Theorem \ref{thmMain}.

At a high level, these three cases arise as follows. First, when $T_n(1-p_n) \rightarrow 0$, the ``density'' of bots within the $T_n$-step incoming neighborhood of $i^*$ is small. As a consequence, $i^*$ is not exposed to the biased beliefs of bots by time $T_n$ and is able to learn the true state ($\theta_{T_n}(i^*) \rightarrow \theta$). On the other hand, when $T_n(1-p_n) \rightarrow \infty$, this ``density'' is large; $i^*$ is exposed to bot beliefs and thus adopts them. Finally, when $T_n(1-p_n) \rightarrow c \in (0,\infty)$, the ``density'' is moderate; $i^*$ does not fully learn, nor does $i^*$ fully adopt bot beliefs.

This explanation is not at all surprising; what is more subtle is what precisely \textit{density of bots within the $T_n$-step incoming neighborhood of $i^*$} means. It turns out that the relevant quantity is the probability that a random walker exploring this neighborhood reaches the set of bots. To illustrate this, consider a random walk $\{X_l\}_{l \in \N}$ that begins at $X_0 = i^*$ and, for $l \geq 0$, chooses $X_{l+1}$ uniformly from all incoming neighbors of $X_{l}$ (agents and bots); note here that the walk follows edges in the direction opposite to their polarity in the graph. For this walk, it is easy to see that, conditioned on the event $X_{l} \in A$, the event $X_{l+1} \in A$ occurs with probability
\begin{equation} \label{eqIllustrateProofProb}
 \frac{d_{in}^A(X_{l}) }{( d_{in}^A(X_{l})+d_{in}^B(X_{l}) } .
\end{equation}
Crucially, we sample this walk and construct the graph simultaneously, by choosing which instub of $X_{l-1}$ to follow \textit{before} actually pairing these instubs. Assuming they are later paired with agent outstubs chosen uniformly at random, and hence connected to agents chosen proportional to out-degree, we can average \eqref{eqIllustrateProofProb} over the out-degree distribution to obtain that $X_{l+1} \in A$ occurs with probability
\begin{equation} \label{eqIllustrateProofProbAvgDeg}
\sum_{a \in A} \frac{d_{in}^A(a)}{d_{in}^A(a)+d_{in}^B(a)} \frac{d_{out}(a)}{\sum_{a' \in A} d_{out}(a') } = \tilde{p}_n .
\end{equation}
Now since bots have a self-loop and no other incoming edges, they are absorbing states on this walk. It follows that $X_{T_n} \in A$ if and only if $X_l \in A\ \forall\ l \in [T_n]$; by the argument above, this latter event occurs with probability $\tilde{p}_n^{T_n}$. Since $\tilde{p}_n \approx p_n$ by \ref{assBranchDegSeq}, we thus obtain that $X_{T_n} \in A$ with probability
\begin{equation}
\tilde{p}_n^{T_n} \approx p_n^{T_n} = \left( 1 - \frac{T_n (1-p_n) }{T_n} \right)^{T_n} \approx e^{-\lim_{n \rightarrow \infty} T_n (1-p_n) } .
\end{equation}
From this final expression, Theorem \ref{thmMain} emerges: when $T_n(1-p_n) \rightarrow 0$, the random walker remains in the agent set with probability $\approx 1$; this corresponds to $i^*$ avoiding exposure to bot beliefs and learning the true state. Similarly, $T_n(1-p_n) \rightarrow \infty$ means the walker is absorbed into the bot set with probability $\approx 1$, corresponding to $i^*$ adopting bot estimates. Finally, $T_n(1-p_n) \rightarrow c \in (0,\infty)$ means the walker stays in the agent set with probability $\approx e^{-c} \in (0,1)$, corresponding to $i^*$ not fully learning nor fully adopting bot estimates.

We note that the actual proof of Theorem \ref{thmMain} does not precisely follow the foregoing argument. Instead, we locally approximate the graph construction with a certain branching process; we then study random walks on the tree resulting from this branching process.\footnote{This is necessary because the argument leading to \eqref{eqIllustrateProofProbAvgDeg} assumes instubs are paired with outstubs chosen uniformly at random, which is not true if resampling of outstubs occurs in the construction from Section \ref{secGraphModel}.} However, the foregoing argument illustrates the basic reason why the three distinct cases of Theorem \ref{thmMain} arise. We also observe that the argument leading to \eqref{eqIllustrateProofProbAvgDeg} shows why $\tilde{p}_n$ enters into our analysis. The other random variables defined in \eqref{eqWalkRVsMain} enter similarly. Specifically, $\tilde{p}_n^*$ arises in almost the same manner, but pertains only to the first step of the walk; this distinction arises since the walk starts at $i^*$, the degrees of which relate to $\tilde{p}_n^*$. On the other hand, $\tilde{q}_n$ arises when we analyze the variance of agent estimates. This is because analyzing the variance involves studying \textit{two} random walks; by an argument similar to \eqref{eqIllustrateProofProbAvgDeg}, the probability of both walks visiting the same agent is
\begin{equation}
 \sum_{a \in A} \frac{d_{in}^A(a)}{ (d_{in}^A(a)+d_{in}^B(a))^2 }  \frac{d_{out}(a)}{\sum_{a' \in A} d_{out}(a') } = \tilde{q}_n .
\end{equation}

\new{Finally, we note that the proof of Theorem \ref{thmMain} reveals that the variance of each agent's belief vanishes, so beliefs converge to Dirac measures. Combined with the theorem, this yields the following corollary. See %
\iftoggle{arxiv}{%
Appendix \ref{appProofBeliefs} %
}%
{%
\cite[Appendix V]{vial2019local_arxiv} %
}%
for a proof.

\begin{corollary} \label{corBeliefs}
Assume $G$ is the DCM and \ref{assGraphDegSeq}, \ref{assGraphHorizon}, \ref{assBranchDegSeq}, and \ref{assBranchHorizon} hold. Let $L(p_n) = L( \{ p_n \}_{n=1}^\infty , T_n )$ denote the limit (in probability) of $\theta_{T_n}(i^*)$ from Theorem \ref{thmMain}. Then for any $p\geq 1$ and for $i^* \sim A$ uniformly, $W_p ( \mu_{T_n}(i^*) , \delta_{ L(p_n) } ) \xrightarrow[n \rightarrow \infty]{\P}  0$.
\end{corollary}
}

\subsection{Comments on assumptions} \label{secMainDiscuss}

We now return to comment on the assumptions needed to prove our results. First, \ref{assGraphDegSeq} states that certain empirical moments of the degree distribution -- namely, for $i^* \sim A$ uniformly, the first two moments of $d_{out}(i^*)$ and the correlation between $d_{out}(i^*)$ and $d_{in}^A(i^*)$ -- converge to finite limits. Roughly speaking, this says our graph lies in a sparse regime, where typical node degrees do not grow with the number of nodes.\footnote{This is analogous to e.g.\ an Erd\H{o}s-R\'enyi model with edge probability $\lambda / n$ for constant $\lambda > 0$, where degrees converge to $\textrm{Poisson}(\lambda)$ random variables.} We also note $\nu_3 > \nu_1$ in \ref{assGraphDegSeq} is minor and simply eliminates an uninteresting case. To see this, first note that when $\Omega_{n,1}$ holds, we have (roughly)
\begin{equation} \label{eqApproxNu3Nu1}
 \frac{\nu_3}{\nu_1} \approx \sum_{i=1}^n \frac{d_{out}(i)}{ \sum_{i'=1}^n d_{out}(i') } d_{in}^A(i) \geq 1,
\end{equation}
where we have used the assumed inequality $d_{in}^A(i) \geq 1\ \forall\ i \in [n]$. Hence, $\nu_3 < \nu_1$ cannot occur, so assuming $\nu_3 > \nu_1$ only prohibits $\nu_3 = \nu_1$. This remaining case is uninteresting because $\nu_3 / \nu_1$ is the limiting number of offspring for each node in the branching process we analyze; thus, if $\nu_3 = \nu_1$, the tree resulting from this process is simply a line graph.

Next, \ref{assGraphHorizon} states $T_n = O(\log n)$. Together with \ref{assGraphDegSeq}, these assumptions are standard given our analysis approach, which, as discussed previously, locally approximates the graph construction with a branching process. We also note that, with the interpretation of $\nu_3 / \nu_1$ above, it follows that the number of agents within the $T_n$-step neighborhood of $i^*$ is roughly
\begin{equation}
( \nu_3 / \nu_1 )^{T_n} = O \left( (\nu_3/\nu_1)^{ \zeta \log_{\nu_3/\nu_1} ( n ) } \right) = O \left( n^{\zeta} \right) = o ( n ) .
\end{equation}
In words, the size of the aforementioned neighborhood vanishes relative to $n$. This is why our title refers to the learning as ``local'': only a vanishing fraction of other agents (those within this neighborhood) affect the estimate of $i^*$.

The remaining statements are needed to establish estimate convergence on the tree resulting from the branching process. \ref{assBranchHorizon} states $T_n \rightarrow \infty$ with $n$, which is an obvious requirement for convergence. \ref{assBranchDegSeq} essentially says that three events occur with high probability. First, $\tilde{p}_n$ should be close to a convergent, deterministic sequence $p_n$; this is necessary since the asymptotics of $p_n$ play a prominent role in Theorem \ref{thmMain}. Second, $\tilde{p}_n^* \geq \tilde{p}_n$ essentially says that bots prefer to attach to agents with higher out-degrees, i.e.\ more influential agents; this is the behavior one would intuitively expect from bots aiming to disrupt learning. Third, $\tilde{q}_n < 1 - \xi \in (0,1)$ is satisfied if, for example, all agents have total in-degree at least two.

\new{Finally, while we focused on the DCM in this section, our analytical approach is more general. At a high level, the key properties of the DCM we used are that most nodes' $O(\log n)$-step neighborhoods are treelike and ``statistically similar,'' which allows for a branching process coupling. Such couplings exist more generally, though this $O(\log n)$ scaling will be smaller for denser graphs, which makes $T_n$ smaller as well.
}

%% file: adversarial.tex
\section{Adversarial setting} \label{secAdversarial}

We next formalize the adversarial problem introduced in Section \ref{secIntro}. We begin with some notation. Let $m_n = \sum_{i=1}^n d_{out}(i)$, and (with slight abuse of notation to the previous section), define the function $\tilde{p}_n : \N_0^n \rightarrow [0,1]$ by
\begin{equation}  \label{eqTildePnFunction}
 \tilde{p}_n(d) = \sum_{i=1}^n \frac{ d_{in}^A(i) }{ d_{in}^A(i) + d(i) } \frac{ d_{out}(i) }{ m_n }\ \forall\ d \in \N_0^n ,
\end{equation} 
which is simply $\tilde{p}_n$, as defined in \eqref{eqWalkRVsMain}, viewed as a function of the bot in-degrees $d(i) \triangleq d_{in}^B(i)$\footnote{We suppress the sub- and super-scripts to avoid cumbersome notation.}. Given a budget $b_n \in \N$, the adversary's problem is then as follows:
\begin{equation} \label{eqOriginalOpt}
  \min_{d \in \N_0^n} \tilde{p}_n(d)\ s.t.\ \sum_{i=1}^n d(i) \leq b_n .
\end{equation}
Thus, the adversary's objective function only depends on the agent degrees $\{ d_{out}(i) , d_{in}^A(i) \}_{i \in [n]}$ (e.g.\ numbers of followers and followees on Twitter), and not the topology of the agent sub-graph. Consequently, the topology will play no role in this section, \new{i.e.\ we do not require the DCM assumption}. We reiterate that, by Theorem \ref{thmMain}, solving \eqref{eqOriginalOpt} is equivalent to minimizing estimates asymptotically for the DCM.\footnote{More precisely, this only holds if the solution of \eqref{eqOriginalOpt} converges in the sense of \ref{assBranchDegSeq}. We are unsure if this holds, but we view it as a minor technical point and leave it as an open problem.} For general graph topologies, we treat \eqref{eqOriginalOpt} as a nonrigorous but tractable surrogate for estimate minimization, and we will soon show empirically that this is a reasonable choice.

\subsection{Exact solution} \label{secExactSolution}

\begin{algorithm}[t]
\DontPrintSemicolon
\caption{Exact solution of \eqref{eqOriginalOpt}} \label{algExact}

Let $d \in dom( \hat{p}_n )$, compute $\hat{p}_n(d)$

\While{not terminated}{

Compute $\hat{p}_n(d-e_i+e_j)\ \forall\ i,j \in [n]$ s.t.\ $i \neq j$

Let $(i^*,j^*) \in \arg \min_{ (i,j) \in [n]^2 : i \neq j }  \hat{p}_n(d-e_i+e_j)$

\lIf{$\hat{p}_n(d) \leq \hat{p}_n ( d - e_{i^*} + e_{j^*} )$}{terminate}

\lElse{Set $d = d - e_{i^*} + e_{j^*}$}

}

\end{algorithm}

\new{First, we let $dom(\hat{p}_n) = \{ d \in \N_0^n : \sum_{i=1}^n d(i) = b_n \}$ and rewrite \eqref{eqOriginalOpt} as $\min_{d \in \Z^n} \hat{p}_n(d)$, where}
\begin{gather} \label{eqRewriteOpt}
\new{\hat{p}_n(d) = \begin{cases} \tilde{p}_n(d) , & d \in dom(\hat{p}_n) \\ \infty , & \textrm{otherwise} \end{cases} .}
\end{gather}
In words, we incorporated the constraints from \eqref{eqOriginalOpt} into the objective; we also used the (obvious) fact that the solution of \eqref{eqOriginalOpt} satisfies the budget constraint with equality. \new{The new objective $\hat{p}_n$ satisfies a certain discrete convexity property, which implies that $d$ minimizes $\hat{p}_n$ if and only if $\hat{p}_n(d) \leq \hat{p}_n(d+e_i-e_j)$ for any $i,j$ pair. Hence, we can find the minimizer by iteratively replacing $d$ with $d+e_i-e_j$ until the objective stops decreasing. This approach is known as \textit{steepest descent} \cite[Section 10.1.1]{murota2003discrete} and is provided in Algorithm \ref{algExact}. In
\iftoggle{arxiv}{%
Appendix \ref{appAlgDetails}, %
}{%
\cite[Appendix III-E]{vial2019local_arxiv}, %
}%
we show its runtime is $\Theta(n^2)$ in the best case and $O(n^2 b_n)$ in the general case.
}

\subsection{Approximation algorithm} \label{secApproxSolution}

\begin{algorithm}[t]
\DontPrintSemicolon
\caption{Approximate solution of \eqref{eqOriginalOpt}} \label{algApprox}

Compute $d_n^{rel}(i)$ as in \eqref{eqRelaxedSolnMain} and set $d_n^{rand}(i) = 0\ \forall\ i \in [n]$

\For{$j=1$ \KwTo $b_n$}{

Sample $W_j$ from the distribution $d_n^{rel} / \sum_{k = 1}^n d_n^{rel}(k)$, i.e.\ $\P ( W_j = i ) = d_n^{rel}(i) /  \sum_{k = 1}^n d_n^{rel}(k)\ \forall\ i \in [n]$

}

Set $d_n^{rand}(i) = \sum_{j=1}^{b_n} 1 ( W_j = i )\ \forall\ i \in [n]$

\end{algorithm}

Algorithm \ref{algExact}'s $\Omega ( n^2 )$ runtime is prohibitive for massive networks like Twitter, which motivates our approximation scheme. The idea is to first solve the relaxed problem
\begin{equation} \label{eqRelaxedOpt}
 \min_{d \in \R_+^n} \tilde{p}_n(d)\ s.t.\ \sum_{i=1}^n d(i) \leq b_n ,
\end{equation}
and then to sample bot locations in proportion to the relaxed solution. More formally, our approximate solution $d_n^{rand}$ is constructed via Algorithm \ref{algApprox}. We note that by definition, the budget constraint holds with equality for Algorithm \ref{algApprox}. Also, as shown in %
\iftoggle{arxiv}{%
Appendix \ref{proofLemRelaxedProblem}%
}{%
\cite[Appendix III-A]{vial2019local_arxiv}%
}%
, the solution of \eqref{eqRelaxedOpt} is
\begin{equation} \label{eqRelaxedSolnMain}
  d_n^{rel}(i) = d_{in}^A(i) \left( \frac{ \sqrt{r(i)} }{ h^* } - 1 \right)_+\ \forall\ i \in [n] ,
\end{equation}
where $x_+ = x 1 ( x > 0)$, $r(i) = d_{out}(i) / d_{in}^A(i)\ \forall\ i \in [n]$, $h^* = \max_{x \in \R_+} h(x)$, and
\begin{equation} \label{eqHmain}
h(x) = \frac{ \sum_{i \in [n] : r(i) \geq x^2} \sqrt{ d_{out}(i) d_{in}^A(i) } }{ b_n + \sum_{i \in [n] : r(i) \geq x^2} d_{in}^A(i) }\ \forall\ x \in \R_+ .
\end{equation}

This randomized scheme yields useful insights, in contrast to the optimal algorithm. In particular, the randomized and relaxed solutions $d_n^{rand}$ and $d_n^{rel}$ are equal in expectation, and the relaxed solution $d_n^{rel}$ satisfies some intuitive properties:
\begin{itemize}[leftmargin=10pt,itemsep=0pt,topsep=0pt]
\item $d_n^{rel}(i)$ grows with $r(i) = d_{out}(i) / d_{in}^A(i)$, i.e.\ the adversary targets agents $i$ with large $d_{out}(i)$ and small $d_{in}^A(i)$ under the relaxed solution. Here large $d_{out}(i)$ means $i$ is \textit{influential} (e.g.\ $i$ has many Twitter followers), while small $d_{in}^A(i)$ means $i$ is \textit{susceptible to influence} (e.g.\ $i$ has few Twitter followees, so bot tweets will appear prominently in $i$'s Twitter feed). %In short, the relaxed solution carefully balances \textit{influence} and \textit{susceptibility}.
\item If $r(i) < ( h^* )^2$, then $d_n^{rel}(i) = d_n^{rand}(i) = 0$. Hence, if $i$ is sufficiently non-influential, and/or sufficiently non-susceptible, targeting $i$ gives no value to the adversary.
\item If $r(i) = r(j) > ( h^* )^2$, the relaxed solution yields
\begin{equation}
 \frac{ d_{in}^A(i) }{ d_{in}^A(i) + d_n^{rel}(i) } = \frac{ d_{in}^A(j) }{ d_{in}^A(j) + d_n^{rel}(j) } .
\end{equation}
This can be interpreted as follows: the adversary strives for a similar proportion of fake news in the feeds of users with similar ratios of influence to susceptibility.
\end{itemize}
In short, our approximate solution strives to balance influence and susceptibility. While somewhat intuitive, the precise manner in which this balance occurs (in particular, the form of \eqref{eqRelaxedSolnMain}-\eqref{eqHmain}) is far from obvious.

In %
\iftoggle{arxiv}{%
Appendix \ref{appAlgDetails} %
}{%
\cite[Appendix III-E]{vial2019local_arxiv} %
}%
, we show Algorithm \ref{algApprox} has complexity $O(n \log n + b_n )$. In terms of accuracy, we next prove that with high probability, Algorithm \ref{algApprox} is a $(2+\delta)$-approximation algorithm for the constrained problem $\max_{d \in \N_0^n : \sum_ d(i) \leq b_n} ( 1 - \tilde{p}_n(d) )$, which is equivalent to \eqref{eqOriginalOpt}. More precisely, letting $d_n^{opt}$ be any solution of \eqref{eqOriginalOpt}, i.e.\
\begin{equation} \label{eqDefnDnOpt}
 d_n^{opt} \in \argmin_{ d \in \N_0^n : \sum_{i=1}^n d(i) \leq b_n} \tilde{p}_n ( d ) ,
\end{equation}
we have the following result.
\begin{theorem} \label{thmConstApprox}
Let $\delta > 0$ and $c_{\delta} = \frac{ \delta^2 }{ 4(2+\delta)^2}$. Then
\begin{align}
& \P \left(  1 - \tilde{p}_n ( d_n^{rand} ) \leq \frac{ 1 - \tilde{p}_n ( d_n^{opt} ) }{ 2 + \delta } \right) \\
& \quad\quad \leq \exp \left( -  \frac{ c_{\delta} m_n ( 1 - \tilde{p}_n(d_n^{rel}) ) }{ \max_{j \in [n]} r(j) } \right) .
\end{align}
\end{theorem}
\begin{proof}
As mentioned above, %
\iftoggle{arxiv}{%
Appendix \ref{proofLemRelaxedProblem} %
}{%
\cite[Appendix III-A]{vial2019local_arxiv} %
}%
shows \eqref{eqRelaxedSolnMain} solves \eqref{eqRelaxedOpt} (the proof amounts to verifying \textit{KKT conditions}, see e.g.\ \cite[Section 5.5.3]{boyd2004convex}). Hence, by definition \eqref{eqDefnDnOpt}, 
\begin{equation} \label{eqDrelDopt}
\tilde{p}_n ( d_n^{rel} ) \leq \tilde{p}_n ( d_n^{opt} ) .
\end{equation}
We next rewrite $1 - \tilde{p}_n ( d_n^{rand})$ in terms of the random vector $W = ( W_j )_{j=1}^{b_n}$ from Algorithm \ref{algApprox}. Toward this end, let $\bar{r} = \max_{j \in [n]} r(j)$, and for $w = (w_j)_{j=1}^{b_n} \in [n]^{b_n}$ define
\begin{equation}
g_n(w) = \frac{1}{\bar{r}} \sum_{j=1}^{b_n} \frac{ d_{out}(w_j) }{ d_{in}^A(w_j) + \sum_{k=1}^{b_n} 1 ( w_k = w_j ) } .
\end{equation}
Then a simple calculation yields
\begin{equation} \label{eqGnW}
g_n ( W ) = \frac{ m_n }{ \bar{r} } ( 1 - \tilde{p}_n ( d_n^{rand}) ) ,
\end{equation}
and using Jensen's inequality, one can show
\begin{equation} \label{eqEGnW}
\E g_n ( W ) \geq \frac{ m_n }{ 2 \bar{r} } ( 1 - \tilde{p}_n ( d_n^{rel}) )
\end{equation}
(see %
\iftoggle{arxiv}{%
Appendix \ref{appGn} %
}{%
\cite[Appendix III-B]{vial2019local_arxiv} %
}%
for details.) Combining \eqref{eqDrelDopt}-\eqref{eqEGnW},
\begin{equation}
1 - \tilde{p}_n ( d_n^{rand} ) \leq \frac{ 1 - \tilde{p}_n ( d_n^{opt} ) }{ 2 + \delta } \Rightarrow g_n(W) \leq \frac{2 \E g_n(W)}{2+\delta} .
\end{equation}
Also, using \eqref{eqEGnW} and recalling $\bar{r} = \max_{j \in [n]} r(j)$, we have
\begin{equation}
\frac{ c_{\delta} m_n (1-\tilde{p}_n(d_n^{rel} ) ) }{ \max_{j \in [n]} r(j) } \leq 2 c_{\delta} \E g_n(W) .
\end{equation}
By the previous two lines, the following implies the theorem:
\begin{equation} \label{eqGnTail}
\P \left( g_n(W) \leq \frac{2 \E g_n(W)}{2+\delta} \right) \leq \exp \left( - 2 c_{\delta} \E g_n(W)  \right) .
\end{equation}
Such an inequality would follow from a simple Hoeffding bound if $g_n(W)$ was simply $\sum_j W_j$; however, $g_n(W)$ is a much more complicated function. Fortunately, $g_n$ belongs to a special class called \textit{self-bounding functions} \cite[Section 3.3]{boucheron2013concentration}, for which concentration inequalities of the form \eqref{eqGnTail} are known. See %
\iftoggle{arxiv}{%
Appendix \ref{appSelfBounding} %
}{%
\cite[Appendix III-C]{vial2019local_arxiv} %
}%
for details.
\end{proof}

The tail bound in Theorem \ref{thmConstApprox} is opaque, as it relies on $\tilde{p}_n(d_n^{rel})$, which (in general) is difficult to interpret. Under certain assumptions, we can obtain more transparent results. For example, we have the following corollary.
\begin{corollary} \label{corConstApprox}
Let $\bar{r} = \max_{j \in [n]} r(j)$ as above. Assume $\lim_{n \rightarrow \infty} b_n = \infty$ and for some $\epsilon > 0$ independent of $n$,
\begin{equation}  \label{eqCorAssumption}
\lim_{n \rightarrow \infty} | \{ i \in [n] : r(i) \geq \epsilon \bar{r} \} | = \infty .
\end{equation}
Then $\exists\ \{ \delta_n \}_{n \in \N} \subset (0,\infty)$ s.t.\ $\lim_{n \rightarrow \infty} \delta_n = 0$ and
\begin{equation} 
\lim_{n \rightarrow \infty} \P \left(  1 - \tilde{p}_n ( d_n^{rand} ) \leq \frac{ 1 - \tilde{p}_n ( d_n^{opt} ) }{ 2 + \delta_n } \right) = 0 .
\end{equation}
\end{corollary}
\begin{proof}
Since $d_n^{rel}$ solves \eqref{eqRelaxedOpt}, we can weaken the bound in Theorem \ref{thmConstApprox} by replacing $\tilde{p}_n(d_n^{rel})$ with $\tilde{p}_n(d)$ for any $d \in \R_+^n$ with $\sum_i d(i) \leq b_n$. Thus, the proof chooses a particular $d$ that leads to a more tractable bound, and the assumptions ensure this bound vanishes. See %
\iftoggle{arxiv}{%
Appendix \ref{proofCorConstApprox} %
}{%
\cite[Appendix III-D]{vial2019local_arxiv} %
}%
for details.
\end{proof}

In words, the corollary shows our randomized scheme is (asymptotically) a $2$-approximation algorithm with probability tending to $1$. The assumption \eqref{eqCorAssumption} only precludes the case where only finitely many of the degree ratios $r(i)$ are comparable to the maximum $\bar{r}$. This restriction arises because our self-bounding concentration analysis in Theorem \ref{thmConstApprox} requires normalization by $\bar{r}$ (see %
\iftoggle{arxiv}{%
Appendix \ref{appSelfBounding}%
}{%
\cite[Appendix III-C]{vial2019local_arxiv}%
}%
.)

\subsection{Empirical results} \label{secEmpirical}

A fundamental assumption in our adversary solutions is that $\tilde{p}_n$ and $\theta_{T_n}(i^*)$ are correlated, in the sense that minimizing $\tilde{p}_n$ also minimizes $\theta_{T_n}(i^*)$. While Theorem \ref{thmMain} states this correlation holds for the random graph model of Section \ref{secGraphModel}, it is unclear if this correlation occurs in practice. To conclude this section, we present empirical results suggesting that this indeed occurs. In our experiments, we compare our proposed solutions against some natural heuristics:
\begin{itemize}[leftmargin=10pt,itemsep=0pt,topsep=0pt]
\item A naive baseline, which uses Algorithm \ref{algApprox} but samples each $W_j$ uniformly from $[n]$.
\item Three schemes which similarly use Algorithm \ref{algApprox}, along with the observed degrees: sampling $W_j$ proportional to $d_{out}$ (i.e.\ targeting influential nodes), $d_{in}^A$ (i.e.\ targeting susceptible nodes), and $d_{out} / d_{in}^A$ (i.e.\ naively balancing the two).
\item Sampling $W_j$ proportional to $\textrm{PageRank}(\epsilon)$ \cite{page1999pagerank}, where\footnote{In experiments, we compute the first $\lceil \log ( 0.99 ) / \log (1-\epsilon) \rceil$ summands, which guarantees an $l_1$ error bound of $0.01$.}
\begin{equation} \label{eqPRdefn}
 \textrm{PageRank}(\epsilon) = \frac{\epsilon \mathbf{1}_{n}}{n} \sum_{j=0}^{\infty} (1-\epsilon)^j \left( P_A^{\trans} \right)^j ,
\end{equation}
where $\epsilon \in (0,1)$, $\mathbf{1}_{n}$ is the length-$n$ ones vector, and $P_A$ is the agent sub-graph's column-normalized adjacency matrix, i.e.\ the matrix with $(i,j)$-th element
\begin{equation}
P_A(i,j) = \frac{ 1 ( i \rightarrow j \in E_n ) }{ d_{in}^A(j) }\ \forall\ i,j \in [n] .
\end{equation}
PageRank is a commonly-used measure of influence or centrality for graphs in many domains \cite{gleich2015pagerank} (and a richer such measure than $d_{out}$). 
\end{itemize}

We compare our proposed solutions with these heuristics using four datasets from \cite{snapnets}, described in Table \ref{tabDatasets}. We chose these datasets so we could test our proposed solutions on real social networks of two scales: Gnutella and Wiki-Vote have $n < 10^4$, a scale at which the exact solution Algorithm \ref{algExact} is feasible; Pokec and LiveJournal have $n > 10^6$, a scale that renders Algorithm \ref{algExact} infeasible but that more closely resembles social networks of interest. For the experiments, we set $\theta = 0.5$ (to maximize signal variance), $\eta = 0.9$ (to emphasize the effect of the network), and $T_n = 101$ (to ensure the code had reasonable runtime). We let $b_n = \lceil |E_n| / 400 \rceil$, so that 0.25\% of all agent in-edges are connected to bots. For each graph and each of five experimental trials, we chose $\{ d_{in}^B(i) \}_{i \in [n]}$ as described above, added bots to the original graph accordingly, and simulated the learning process from Section \ref{secLearnModel}. 

\begin{table}
\centering
\caption{Dataset details} \label{tabDatasets}
\begin{tabular}{|l|l|c|c|c|}
\hline 
\bf{Name} & \bf{Description} & \bf{Nodes} & \bf{Edges} \\ \hline 
Gnutella & Peer-to-peer network & 6,301 & 20,777  \\ \hline 
Wiki-Vote & Wiki admin elections & 7,115 & 103,689  \\ \hline 
Pokec & Slovakian social network & 1,632,803 & 30,622,564 \\ \hline
LiveJournal & Blogging platform & 4,847,571 & 68,993,773 \\ \hline
\end{tabular} 
\end{table}

In Figure \ref{fig_timeVsBel}, we plot the mean and standard deviation (across experimental trials) of $\theta_t(i^*)$ as a function of $t$. For all datasets, our proposed solutions outperform all heuristics, in the sense that our solutions yield the lowest average $\theta_t(i^*)$ for most values of $t$. Furthermore, we note the following:
\begin{itemize}[leftmargin=10pt,itemsep=0pt,topsep=0pt]
\item Across all graphs, our solutions outperform $\textrm{PageRank}(\epsilon)$ for all values of $\epsilon$ tested. This is quite surprising, because PageRank uses the entire \textit{graph topology}, whereas our solutions only use \textit{degree information}. Also, as $\epsilon$ becomes increasingly smaller, $\textrm{PageRank}(\epsilon)$ performs increasingly better, but this comes at the cost of higher runtime to estimate $\textrm{PageRank}(\epsilon)$.
\item Among the heuristics using (at most) degree information, $d_{out} / d_{in}^A$ performs best -- but still worse than Algorithm \ref{algApprox} -- across all datasets. Put differently, naively balancing influence and susceptibility is not enough; the non-obvious form of Algorithm \ref{algApprox} yields better performance.
\item For Gnutella and Wiki-Vote, Algorithm \ref{algExact} noticeably outperforms Algorithm \ref{algApprox}. Though the former is an exact solution and the latter is an approximation, this is still surprising, since it is unclear that these schemes are even optimizing the correct objective for real graphs.
\end{itemize}

While Figure \ref{fig_timeVsBel} only considers one choice of $b_n$, we believe our conclusions are robust. In particular, we also tested the cases $b_n = \lceil \tilde{b} |E_n| \rceil$ for each $\tilde{b} \in \{ \frac{1}{1600} ,\frac{1}{800}  , \frac{1}{400} , \frac{1}{200} , \frac{1}{100}  \}$, so that between $\approx 0.0625\%$ and $\approx 1\%$ of edges connected to bots (thus, Figure \ref{fig_timeVsBel} shows the intermediate case $\tilde{b} = \frac{1}{400}$). %
\iftoggle{arxiv}{%
Appendix \ref{appOtherBudgets} %
}{%
\cite[Appendix III-F]{vial2019local_arxiv} %
}%
contains a figure analogous to Figure \ref{fig_timeVsBel} for the other choices of $\tilde{b}$; the plots are qualitatively similar.

\begin{figure}
\centering
\includegraphics[width=3in]{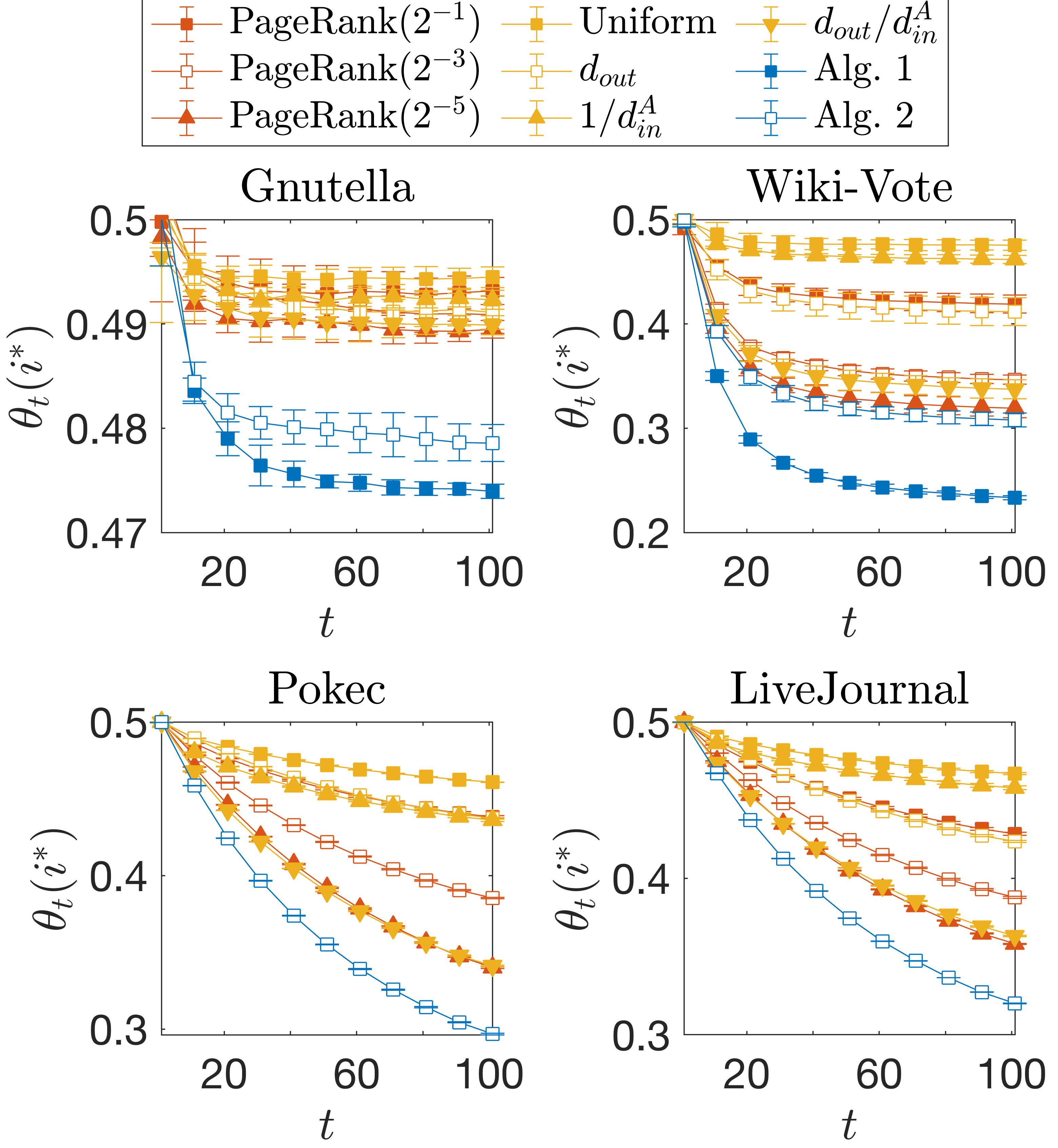}
\caption{Estimates when simulating our learning model on real datasets; Algorithms \ref{algExact} and \ref{algApprox} outperform intuititive heuristics.} \label{fig_timeVsBel}
\end{figure}

We have thus far shown that our solutions outperform heuristics, even those using graph topology. This is quite surprising: our solutions were derived under the fundamental assumption that \textit{minimizing $\theta_{T_n}(i^*)$ amounts to minimizing $\tilde{p}_n$}, but we only verified this assumption asymptotically for a class of random graphs. Thus, our empirical results suggest that \textit{even for real social networks, this assumption holds}. Indeed, in Figure \ref{fig_objVsFinBel} we show scatter plots of $\theta_{T_n}(i^*)$ against $\tilde{p}_n$ (each dot represents one experimental trial). For all datasets, the two quantities are closely correlated.

\begin{figure}
\centering
\includegraphics[width=3.2in]{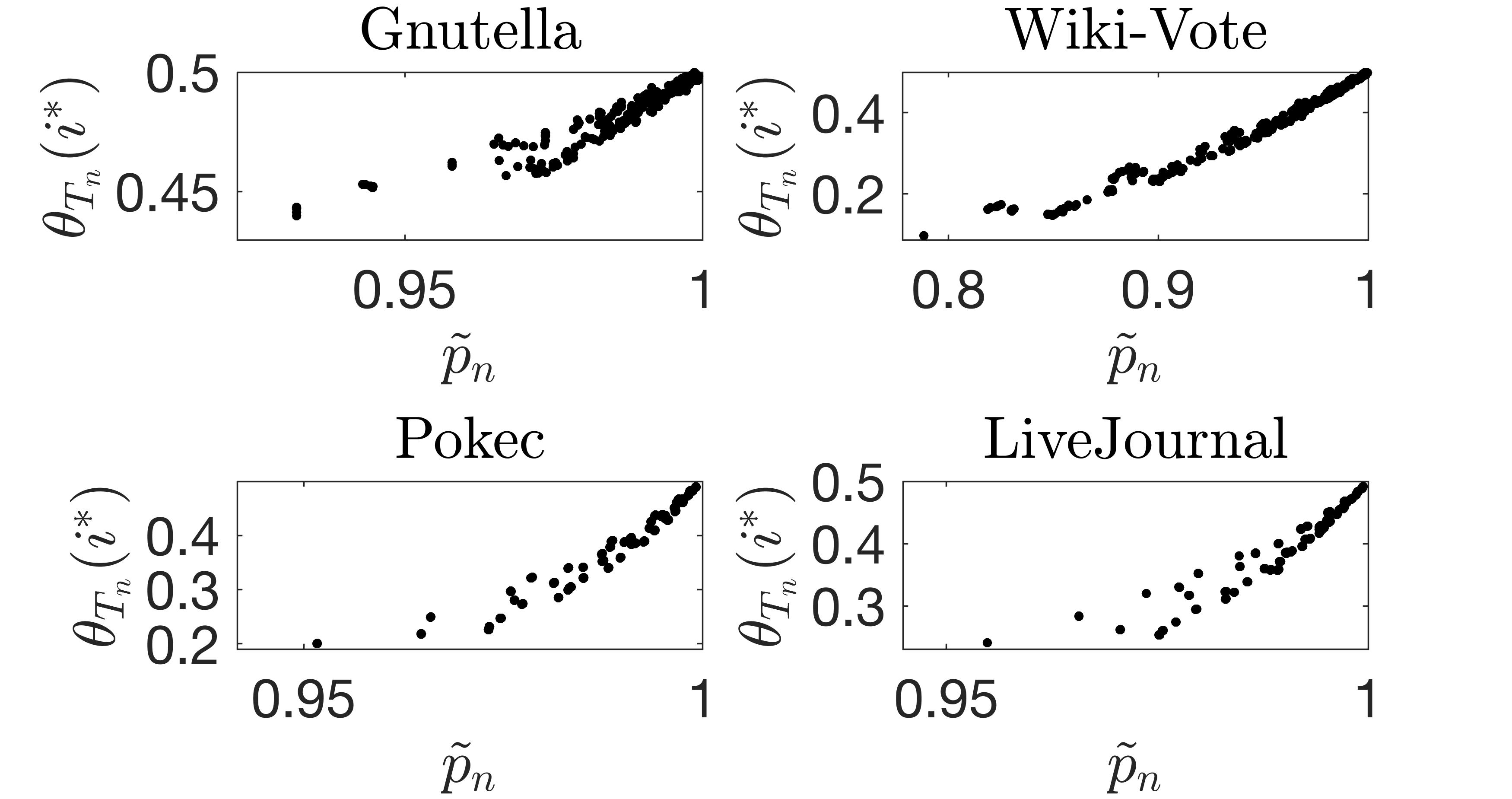}
\caption{As suggested by Figure \ref{fig_timeVsBel}, $\theta_{T_n}(i^*)$ and $\tilde{p}_n$ are closely correlated for real social networks.} \label{fig_objVsFinBel}
\end{figure}

%% file: related.tex
\section{Related work} \label{secRelated}

\new{As discussed in Section \ref{secLearnModel}, \eqref{eqParamUpdateInitial} resembles the non-Bayesian social learning model from \cite{jadbabaie2012non}, which uses belief update
\begin{equation} \label{eqOtherUpdate}
\mu_t(i) = \eta_{ii} \textrm{BU}(\mu_{t-1}(i),\omega_t(i)) +  \sum_{j \in N_{in}(i)} \eta_{ij} \mu_{t-1}(j) ,
\end{equation}
where $\sum_j \eta_{ij} = 1$, $\omega_t(i)$ is a signal, and $\textrm{BU}$ means Bayesian update. Hence, agents perform Bayesian updates and then average in terms of \textit{beliefs} in \cite{jadbabaie2012non} but \textit{parameters} in this work. The main advantage of the latter is that beliefs remain Beta distributions, which simplifies our analysis. This simplification, along with weights $\eta / d_{in}(j)$ instead of \eqref{eqOtherUpdate}, are needed since we consider a finite horizon and a graph which need not be connected, in contrast to \cite{jadbabaie2012non}. Another distinction is that agents in \cite{jadbabaie2012non} cannot learn the true state individually, and need the network for learning. In contrast, agents in our work can learn in isolation (simply by averaging their signals), so the network can either speed up learning or be a detriment. We highlight here the detriment with our model relevant to platforms like Twitter, where users who could have read accurate news in isolation instead of risking exposure to bots.

Our parameter update is also studied in \cite{azzimonti2018social}, which features bots defined in a slightly different manner but in the same spirit. However, \cite{azzimonti2018social} only includes theoretical results in the case $B = \emptyset$; the case $B \neq \emptyset$ is studied empirically. This allowed \cite{azzimonti2018social} to use a slightly richer model, including a time-varying graph and agent-dependent mixture parameters $\sum_{j \in N_{in}(i) \cup \{ i \} } \eta_{ij}$. Notably, the empirical results from \cite{azzimonti2018social} fix a learning horizon and do not investigate the effects of different timescales; in particular, the delicate relationship between timescale and bot prevalence from Theorem \ref{thmMain} is not brought to light. Beyond stubborn agents, \cite{mitra2020new,su2016non} propose different non-Bayesian updates to cope with Byzantine agents with arbitrary behavior.

From an analytical perspective, our approach of analyzing estimates by studying random walks is similar to the deGroot model \cite{degroot1974reaching}. Here the estimate vector $\theta_t = \{ \theta_t(i) \}_i$ is updated as $\theta_t = \theta_{t-1} W$ for some column-stochastic matrix $W$. Hence, $\theta_t = \theta_0 W^t$, so $i$'s belief is determined by the distribution of a $t$-step random walk from $i$. This observation has been exploited in the literature; see the surveys \cite[Section 3]{golub2017learning} and \cite[Section 4]{acemoglu2011opinion}, and the references therein. For example, assuming $W$ is irreducible and aperiodic, and therefore has a well-defined stationary distribution $\pi$, \cite{golub2010naive} establishes conditions for learning using the fact that $\theta_t(i) =  \theta_0 W^t e_i^{\trans} \approx \theta_0 \pi^{\trans}\ \forall\ i$ when $t$ is large. Roughly speaking, our model combines deGroot-like averaging with exogenous unbiased signals. As discussed, the averaging in our case exposes agents to biased beliefs (due to bots); the resulting tension between biased and unbiased information is a key feature in our model not present in deGroot's. Ours is arguably a richer model of platforms like Twitter, where there is a similar tension between legitimate news and bots. Beyond the deGroot model, agents in \cite{rahimian2015learning} perform Bayesian updates using the prior of a randomly-chosen neighbor, which yields a different connection to random walks; assuming strong connectedness, the authors exploit the fact that the walk visits every agent infinitely often (i.o.) to derive conditions for learning.}

Similar to \cite{jadbabaie2012non}, the papers of the previous paragraph typically assume strong connectedness and long learning horizons so as to leverage properties such as stationary distributions and i.o.\ visits. This is a fundamental distinction from our work. Indeed, even if we disregard stubborn agents, the random walk converges to a stationary distribution, but it does \textit{not} converge within our local learning horizon. This is because, as shown in \cite{bordenave2018random}, the DCM we consider has mixing time that exceeds
\begin{equation}
\frac{\log n}{ \sum_{i \in [n]} \log ( d_{in}^A(i) ) \frac{ d_{out}(i)}{ \sum_{i' \in A} d_{out}(i')} } %\geq \frac{\log n}{ \log ( \sum_{i \in [n]} d_{in}^A(i) \frac{d_{out}(i)}{\sum_{i' \in A} d_{out}(i')} ) } 
\gtrapprox \frac{\log n}{\log(\nu_3/\nu_1)} , 
\end{equation}
where we used Jensen's inequality and \eqref{eqApproxNu3Nu1}. The right side exceeds $T_n$ by \ref{assGraphHorizon}, i.e.\ our learning horizon occurs before the underlying random walk mixes. In fact, \cite{bordenave2018random} shows that the random walk on the DCM exhibits \textit{cutoff}, meaning that the $T_n$-step distribution of this walk can be maximally far from the stationary distribution (i.e.\ the total variation distance between these distributions can be 1 for certain starting locations of the walk). Hence, not only can we not use this stationary distribution, we cannot even use an approximation of it. \new{Again, this means our analysis cannot leverage global properties typically used when relating estimates to random walks. We circument this using the DCM, which has a well-behaved local structure.} We also note that our idea to simultaneously construct the graph and sample the walk is taken from \cite{bordenave2018random}.

Some other works have considered social learning with stubborn agents. For example, \cite{acemoglu2010spread} studies a model in which agents meet and either retain their own estimates, adopt the average of their estimates, or adopt a weighted average; the agent whose estimate has a larger weight is called a ``forceful'' agent. Here the authors show that all agent estimates converge to a common random variable and study its deviation from the true state. A crucial difference between this work and ours is that \cite{acemoglu2010spread} assumes even forceful agents occasionally observe other agents' opinions. This yields an underlying Markov chain that is irreducible (unlike ours); the analysis then relies on this chain having a well-defined stationary distribution.

Stubborn agents have also been considered in the consensus setting \cite{thenis1987consensus}, which asks whether agent estimates converge to a common value, i.e.\ a consensus. For example, \cite{acemoglu2011opinion2} considers a model in which regular agents adopt weighted averages of estimates upon meeting other agents, while stubborn agents always retain their own estimates. This intuitively prohibits a consensus from forming; indeed, it is shown that agent estimates fail to converge, i.e.\ disagreement can persist indefinitely. Another example is \cite{ghaderi2014opinion}, in which an agent's estimate at time $t+1$ is a weighted average of their own estimate at time 0 and their neighbors' estimates at time $t$. In this model, stubborn agents place all weight on their own estimate from time 0 and thus do not update their estimates. The analysis in \cite{ghaderi2014opinion} is similar to ours as it relates agent estimates to hitting probabilities of the stubborn agent set, but it differs as the learning horizon is infinite in \cite{ghaderi2014opinion}. Also in the consensus setting, \cite{rocket2018snowflake} investigates protocols for robust consensus that may lessen the undesirable effects of stubborn agents.

\new{The problem of deploying stubborn agents is studied in \cite{yildiz2013binary,sadler2020influence}, though for the voter model. Both assume knowledge of a matrix describing the graph topology (like $P_A$ from Section \ref{secEmpirical}), and the optimization requires inverting this matrix at complexity $n^3$. Our algorithms overcome both of these issues. We also note this inversion is common in more general influence maximization settings.

Without stubborn agents, \cite{anunrojwong2018naive} considers a non-Bayesian update for infinite horizons, where agents treat neighbors' beliefs as independent. Convergence rates are provided in \cite{lalitha2014social,shahrampour2015distributed,nedic2017fast} for \eqref{eqParamUpdateInitial} or similar Bayesian-plus-aggregation updates. An open question is how these models behave with stubborn agents, particularly for \cite{lalitha2014social,shahrampour2015distributed,nedic2017fast}, where the convergence may be slower than the propagation of stubborn agent bias.}

%% file: specialCase.tex
\section{Special case} \label{secSecondResult}

While Theorem \ref{thmMain} establishes convergence for the estimate of a typical agent, a natural question to ask is how many agents have convergent estimates. Our second result, Theorem \ref{thmSec}, provides a partial answer to this question. To prove the result, we require slightly stronger assumptions than those required for Theorem \ref{thmMain} (we will return shortly to comment on why these are needed). First, we strengthen \ref{assGraphDegSeq} and \ref{assBranchDegSeq} to include particular rates of convergence for the probabilities $\P(\Omega_{n,i}), i \in \{1,2\}$. Second, we strengthen \ref{assBranchHorizon} with a minimum rate at which $T_n \rightarrow \infty$ (specifically, $T_n = \Omega(\log n)$). Third, and perhaps most restrictively, we require $p_n \rightarrow p < 1$ in \ref{assGraphDegSeq}. As a result, Theorem \ref{thmSec} only applies to the case $T_n(1-p_n) \rightarrow \infty$, for which Theorem \ref{thmMain} states the estimate of a uniform agent converges to zero. In this setting, Theorem \ref{thmSec} provides an upper bound on how many agents' estimates do \textit{not} converge to zero. In particular, this bound is $O(n^k)$ for some $k < 1$.

\begin{theorem} \label{thmSec}
Assume $\exists\ \kappa, \mu > 0$ and $N' \in \N$ independent of $n$ s.t.\ the following hold:
\begin{itemize}
\item \ref{assGraphDegSeq}, with $\P(\Omega_{n,1}) = O(n^{-\kappa})$.
\item \ref{assGraphHorizon}.
\item \ref{assBranchDegSeq}, with $\P(\Omega_{n,2}) = O(n^{-\kappa})$ and $p < 1$.
\item \ref{assBranchHorizon}, with $T_n \geq \mu \log n\ \forall\ n \geq N'$.
\end{itemize}
Then for any $\epsilon > 0$, $k > 1 - \min \{ (1/2) - \zeta , \mu ( \epsilon \eta (1-p) / \theta )^2 / 16 , \kappa \}$, and $K > 0$, all independent of $n$,
\begin{equation}
\lim_{n \rightarrow \infty} \P \left( \left| \left\{ i \in [n] : \theta_{T_n}(i) > \epsilon \right\} \right| > K n^k \right) = 0 .
\end{equation}
\end{theorem}

We reiterate that $\zeta < 1/2$ by \ref{assGraphHorizon} and $\mu , \kappa > 0$ by the theorem statement. Hence, $\min \{ (1/2) - \zeta , \mu ( \epsilon \eta (1-p) / \theta )^2 / 16 , \kappa \} > 0$, so one can choose $k < 1$ in Theorem \ref{thmSec} to show that the size of the non-convergent set of agents vanishes relative to $n$. We suspect that such a result is the best one could hope for; in particular, we suspect that showing \textit{all} agent estimates converge to zero is impossible. This is in part because our assumptions do not preclude the graph from being disconnected. Hence, there may be small connected components composed of agents but no bots; in such components, agent estimates will converge to $\theta$ (not zero). Additionally, while the lower bound for $k$ in Theorem \ref{thmSec} is somewhat unwieldy, certain terms are easily interpretable: the bound sharpens as $\eta$ grows (i.e.\ as agents place less weight on their unbiased signals), as $p$ decays (i.e.\ as the number of bots grows), and as $\theta$ decays (i.e.\ as signals are more likely to be zero, pushing estimates to zero).

As for Theorem \ref{thmMain}, the proof of Theorem \ref{thmSec} is outlined in Appendix \ref{appMainProofOutline} with details provided in Appendix \ref{appMainProofDetails}. The crux of the proof involves obtaining a sufficiently fast rate for the convergence in Theorem \ref{thmMain}; namely, we show that for some $\gamma > 0$, $\P ( \theta_{T_n}(i^*) > \epsilon ) = O(n^{-\gamma})$.\footnote{One may wonder why we derive a separate bound for Theorem \ref{thmSec}, since we have already bounded $\P ( \theta_{T_n}(i^*) > \epsilon )$ to prove Theorem \ref{thmMain}. The reason for this is that the bound for Theorem \ref{thmMain} does not decay quickly enough as $n \rightarrow \infty$ to prove Theorem \ref{thmSec}; on the other hand, the bound for Theorem \ref{thmSec} does not decay at all as $n \rightarrow \infty$ for the case $T_n(1-p_n) \rightarrow [0,\infty)$ and therefore cannot be used for all cases of Theorem \ref{thmMain}. See Appendix \ref{appStep2compare} for details.} At a high level, obtaining such a bound requires bounding three probabilities by $O(n^{-\gamma})$, which also helps explain the stronger assumptions of Theorem \ref{thmSec}:
\begin{itemize}
\item As for Theorem \ref{thmMain}, we first locally approximate the graph construction with a branching process so as to analyze the estimates on a tree. Here strengthening \ref{assGraphDegSeq} with $\P(\Omega_{n,1}) = O(n^{-\kappa})$ is necessary to ensure this approximation fails with probability at most $O(n^{-\gamma})$.
\item To analyze the estimates on a tree, we first condition on the random tree structure and treat the estimate as a weighted sum of i.i.d.\ signals using an approach similar to Hoeffding's inequality. Namely, we obtain the Hoeffding-like tail $O(e^{-2 \epsilon^2 T_n})$; strengthening \ref{assBranchHorizon} with $T_n \geq \mu \log n$ is necessary to show this tail is $O(e^{-2 \epsilon^2 \mu \log n}) = O(n^{-2 \epsilon^2 \mu}) = O(n^{-\gamma})$.
\item Finally, after conditioning on the tree structure, we show this structure is close to its mean. More specifically, letting $\E [ \hat{\vartheta}_{T_n}(\phi) | \T ]$ denote the expected estimate for the root node in the tree conditioned on the random tree structure (see Appendix \ref{appMainProofOutline} for details), we show
\begin{equation}
\P( \E [ \hat{\vartheta}_{T_n}(\phi) | \T ] > \epsilon) = O \left( n^{-\gamma} \right) .
\end{equation}
Note the only source of randomness in $\E [ \hat{\vartheta}_{T_n}(\phi) | \T ]$ is the random tree; because this tree is recursively generated, it has a martingale-like structure that can be analyzed using an approach similar to the Azuma-Hoeffding inequality for bounded-difference martingales. Here we require $\P(\Omega_{n,2}) = O(n^{-\kappa})$ to ensure the degree sequence is ill-behaved with probability at most $O(n^{-\gamma})$; we also require $p_n \rightarrow p < 1$ in this step (and only in this step).
\end{itemize}

We now address the most notable difference between Theorems \ref{thmMain} and \ref{thmSec}; namely, that the latter only applies when $p_n \rightarrow p < 1$. We believe this reflects a fundamental distinction between the cases $p_n \rightarrow p < 1$ and $p_n \rightarrow 1$ and is \textit{not} an artifact of our analysis. An intuitive reason for this is that more bots are present in the former case, so fewer random signals are present (recall we model bot signals as being deterministically zero). As a result, $\theta_{T_n}(i^*)$ is ``less random'', so its concentration around its mean is stronger. Towards a more rigorous explanation, we first note that Appendix \ref{appExtendSecThmToOthers} provides the following condition for extending Theorem \ref{thmSec} to other cases of $p_n$:
\begin{equation} \label{eqSuffCond}
\exists\ \gamma' > 0\ s.t.\ \P( | \E [ \hat{\vartheta}_{T_n}(\phi) | \T ] - L(p_n) |  > \epsilon) = O \left( n^{-\gamma'} \right) ,
\end{equation}
where $L(p_n)$ is the limit from Theorem \ref{thmMain} based on the relative asymptotics of $T_n$ and $p_n$, i.e.\
\begin{equation}
L(p_n) = \begin{cases} \theta , & T_n(1-p_n) \xrightarrow[n \rightarrow \infty]{} 0 \\  \theta (1-e^{-c \eta})/(c \eta) , & T_n(1-p_n) \xrightarrow[n \rightarrow \infty]{} c \in (0,\infty) \\ 0 , & T_n(1-p_n) \xrightarrow[n \rightarrow \infty]{} \infty \end{cases} .
\end{equation}

It is the convergence of $| \E [ \hat{\vartheta}_{T_n}(\phi) | \T ] - L(p_n) |$ in \eqref{eqSuffCond} that we suspect is fundamentally different in the cases $p_n \rightarrow p < 1$ and $p_n \rightarrow 1$. To illustrate this, we provide empirical results in Figure \ref{figSecThmStats}. In the leftmost plot, we show $1-\tilde{p}_n$ versus $T_n$; here the plot is on a log-log scale, so a line with slope $m$ means $(1-\tilde{p}_n) \propto T_n^m$. Hence, we are comparing four cases: $m \approx 0$, so that $p_n \approx p < 1$ (blue circles); $m \approx -0.5$, so that $T_n(1-p_n) \rightarrow \infty$ and $p_n \rightarrow 1$ (orange squares); $m \approx -1$, so that $T_n(1-p_n) \rightarrow 1$ (yellow diamonds); and $m \approx -1.5$, so that $T_n(1-p_n) \rightarrow 0$ (purple triangles). The second plot reflects the corresponding cases of $L(p_n)$: $\E [ \hat{\vartheta}_{T_n}(\phi) | \T ]$ decays to zero in the first two cases, grows towards $\theta = 0.5$ in the fourth case, and approaches an intermediate limit in the third case. The final two plots illustrate the convergence (or lack thereof) in \eqref{eqSuffCond}. Here the empirical mean of the error term $| \E [ \hat{\vartheta}_{T_n}(\phi) | \T ] - L(p_n) |$ decays quickly for the first case but decays more slowly (or is even non-monotonic) in the other cases. More strikingly, the empirical variance of this error term is several orders of magnitude smaller in the first case. This suggests that $\P( | \E [ \hat{\vartheta}_{T_n}(\phi) | \T ] - L(p_n) |  > \epsilon)$ decays much more rapidly in the case $p_n \rightarrow p < 1$, which is why we believe this is the only case for which \eqref{eqSuffCond} is satisfied.

\begin{figure}
\centering
\includegraphics[height=1.9in]{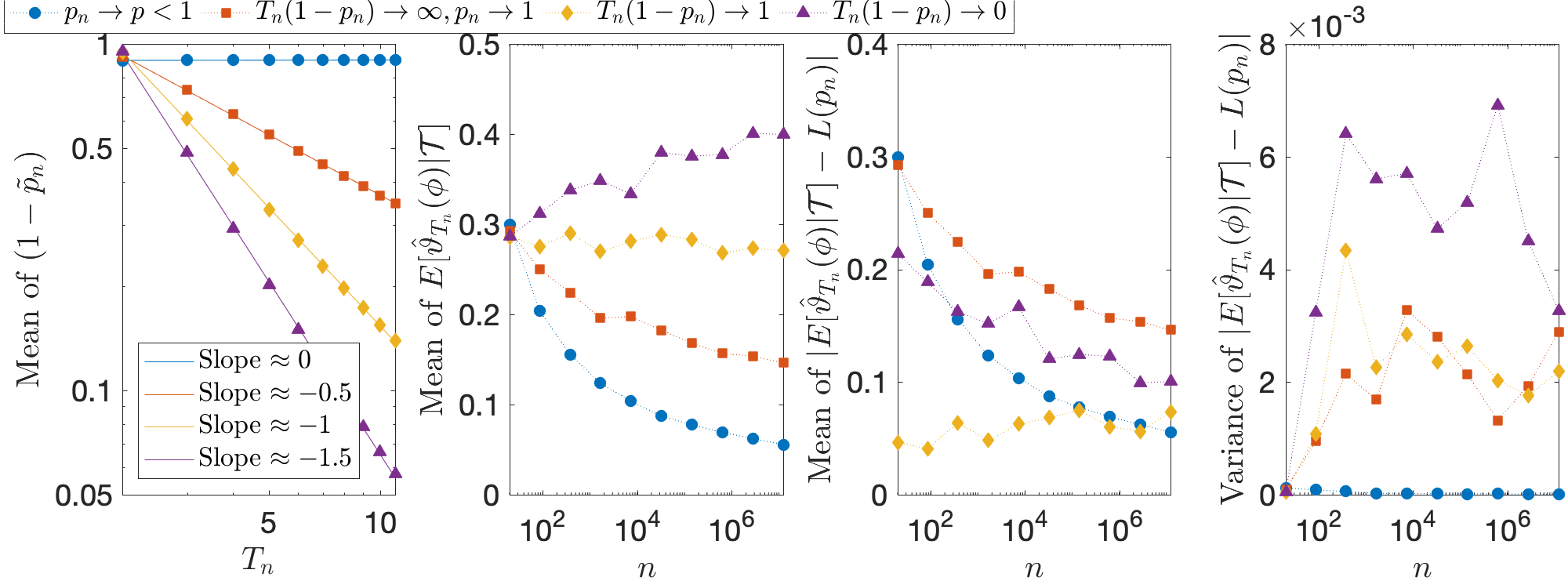}
\caption{Empirical comparison of the cases $p_n \rightarrow p < 1$, $T_n(1-p_n) \rightarrow \infty$ with $p_n \rightarrow 1$, $T_n(1-p_n) \rightarrow 1$, and $T_n(1-p_n) \rightarrow 0$ (leftmost plot). On average, $\E [ \hat{\vartheta}_{T_n}(\phi) | \T ]$ approaches the corresponding limit from Theorem \ref{thmMain} in all cases (second plot from left). However, the error term $| \E [ \hat{\vartheta}_{T_n}(\phi) | \T ] - L(p_n) |$ behaves markedly differently in the case $p_n \rightarrow p < 1$, with a faster decay on average (second plot from right) and a strikingly lower variance (rightmost plot); we believe this is why Theorem \ref{thmSec} only applies in this case.} \label{figSecThmStats}
\end{figure}

In addition to the summary statistics shown in Figure \ref{figSecThmStats}, we also show histograms of error term $| \E [ \hat{\vartheta}_{T_n}(\phi) | \T ] - L(p_n) |$ across the 400 trials in Figure \ref{figSecThmPmfs}. As discussed above, this term must converge to zero (in probability) at a sufficiently fast rate to prove Theorem \ref{thmSec}. In Figure \ref{figSecThmPmfs}, these histograms appear to converge quickly to a point mass at zero in the case $p_n \rightarrow p < 1$; in other cases, such behavior does \textit{not} occur, further suggesting a fundamental difference between the cases.

\begin{figure}
\centering
\includegraphics[width=5in]{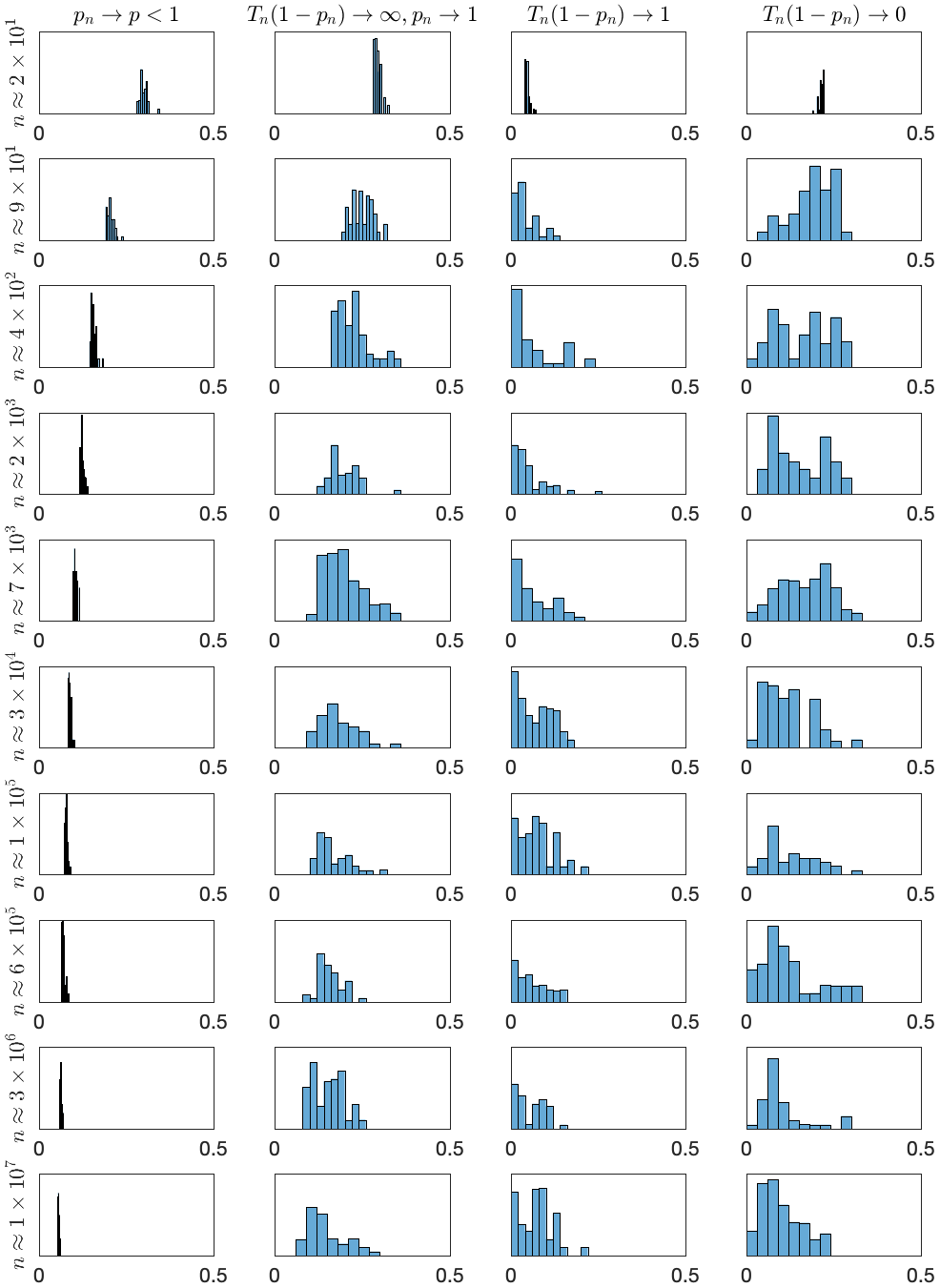}
\caption{Histograms of 400 samples of the error term $| \E [ \hat{\vartheta}_{T_n}(\phi) | \T ] - L(p_n) |$. When $p_n \rightarrow p < 1$, the histogram appears to decay quickly to a point mass on zero; in other cases, this does not occur.} \label{figSecThmPmfs}
\end{figure}

We note here that basic workflow of the experiment above proceeded as follows:
\begin{itemize}
\item Choose a sequence of time horizons $T_n$ that increase linearly, then set $n$ accordingly.
\item Realize the degrees $\{ d_{out}(i) , d_{in}^A(i) , d_{in}^B(i) \}_{i \in [n]}$ after selecting $n$.
\item Define the empirical distributions $f_n, f_n^*$ using the degrees as in \eqref{eqEmpDist}.
\item Evaluate quantity of interest $\E [ \hat{\vartheta}_{T_n}(\phi) | \T ]$ empirically via \eqref{eqMeanTreeBeliefGivenTree} using $f_n, f_n^*$.
\end{itemize}
We repeated this experiment 400 times to obtain 400 samples of $\E [ \hat{\vartheta}_{T_n}(\phi) | \T ]$; the plots in Figure \ref{figSecThmStats} show empirical means and variances across these 400 samples. We used the following parameters:
\begin{itemize}
\item We set $\eta = 0.9$ to emphasize the effect of the network.
\item We let $d_{in}^A(i) = 1 + \textrm{Poisson}(\lambda_A-1)\ \forall\ i \in [n]$, so that $\E [ d_{in}^A(i) ] = \lambda_A$; we choose $\lambda_A$ independent of $n$ so that $\E [ d_{in}^A(i) ] = O(1)$, as required by \ref{assGraphDegSeq}. In particular, we choose $\lambda_A = 2.1$.
\item After realizing $\{ d_{in}^A(i) \}_{i \in [n]}$, we assign one outgoing edge to each $i \in [n]$, then assign each of the remaining $\sum_{i \in [A]} d_{in}^A(i) - n$ outgoing edges independently and uniformly at random. Note that this implies $d_{in}^A(i), d_{out}(i) > 0$ and $\sum_{i \in [n]} d_{in}^A(i) = \sum_{i \in [n]} d_{out}(i)$, as required by \eqref{eqDegSeqBasicAss}.
\item We let $d_{in}^B(i) = \textrm{Poisson}(\lambda_B)$, with $\lambda_B = \lambda_A (1-p_n) / p_n$, so that
\begin{equation}
\E d_{in}^A(i) / ( \E d_{in}^A(i) + \E d_{in}^B(i) )= \lambda_A / ( \lambda_A+\lambda_B) = 1 / ( 1 + (1-p_n)/p_n ) = p_n .
\end{equation}
(This is not precisely what we desire, since \ref{assBranchDegSeq} assumes $p_n \approx \tilde{p}_n = \E_n [ \frac{d_{in}^A(v^*)}{  d_{in}^A(v^*) + d_{in}^B(v^*) } ]$ for $v^*$ sampled proportional to out-degree; however, as shown in the second plot in Figure \ref{figSecThmStats}, this empirically yields distinct cases rates of convergence for $(1-p_n) \rightarrow 0$.)
\item We compare four cases of $p_n$: $p_n = p$ and $p_n =  1 - c_i T_n^{(-i+1)/2}$ for $i \in \{2,3,4\}$, with $p$ and $c_i$ independent of $n$. Note that the three latter cases satisfy
\begin{equation}
(1-p_n)  \propto T_n^{(-i+1)/2} \in \left\{  T_n^{-1/2}  , T_n^{-1} ,  T_n^{-3/2} \right\} ,
\end{equation}
as shown in Figure \ref{figSecThmStats}. Here $p$ and $c_i$ were chosen via trial-and-error so that all four cases behaved roughly the same at the smallest value of $n$ (as in Figure \ref{figSecThmStats}). In particular, we chose
\begin{equation}
p = 0.9, \quad c_2 = 1.3, \quad c_3 = 1.9, \quad c_4 = 2.7 .
\end{equation}
\item We let $T_n \in \{2,3,\ldots,11\}$; here the minimum of 2 was chosen since $T_n = 1$ is a trivial case and the maximum of 11 was chosen due to computational limitations.
\item Given $T_n$, we let $n = \lceil \lambda_A^{2 T_n} \rceil$. Note that this implies $T_n \approx ( \log n ) / ( 2 \log \lambda_A )$, roughly the upper bound in \ref{assGraphHorizon}. With our choice of $T_n$ and $\lambda_A$, $n$ ranged from 20 to (roughly) 12 million.
\end{itemize}

%% file: mainProofOutline.tex
\section{Proof of Theorems \ref{thmMain} and \ref{thmSec} (outline)} \label{appMainProofOutline}

The proofs of Theorems \ref{thmMain} and \ref{thmSec} proceed in two steps. First, we show that the graph construction can be locally approximated by a certain branching process. Second, we analyze the estimates of agents in the graph by instead analyzing the estimates of agents in the tree resulting from the branching process. We note that studying tree agent estimates rather than graph agent estimates is advantageous because the tree has a comparatively simple structure that is more amenable to analysis.

The first step is identical for both theorems, while the second step requires a different analysis for each theorem. In Appendix \ref{appBranchApproxProofOutline}, we outline the first step, and in Appendices \ref{appRootBeliefProofOutline} and \ref{appRootBelief2ProofOutline}, respectively, we outline the second step for Theorems \ref{thmMain} and \ref{thmSec}, respectively. To highlight the key ideas of our analysis, we defer many details to Appendix \ref{appMainProofDetails}; in particular, proofs pertaining to Appendices \ref{appBranchApproxProofOutline}, \ref{appRootBeliefProofOutline}, and \ref{appRootBelief2ProofOutline} , respectively, can be found in Appendices \ref{appBranchApproxProofDetails}, \ref{appRootBeliefProofDetails}, and \ref{appRootBelief2ProofDetails}, respectively. Finally, we note that throughout the analysis we use $\P_n$ and $\E_n$, respectively, to denote probability and expectation, respectively, conditioned on the degree sequence $\{ d_{out}(i) , d_{in}^A(i), d_{in}^B(i) \}_{i \in [n]}$.

\subsection{Branching process approximation (Step 1 for proofs of Theorems \ref{thmMain} and \ref{thmSec})} \label{appBranchApproxProofOutline}

We first show that the estimate of any agent in the graph depends (asymptotically) only on the structure of the agent's local neighborhood and on certain signals realized within this neighborhood. This will facilitate the definition of the branching process with which we will approximate the graph construction. Importantly, the graph agent's estimate will \textit{not} depend on the prior parameters $\alpha_0, \beta_0$ (asymptotically). This is necessary as we have not specified these priors (beyond assuming they are bounded by some $\bar{\alpha}, \bar{\beta}$ independent of $n$, as discussed in Section \ref{secLearnModel}). 

To begin, we require some notation. Let $P$ denote the graph's column-normalized adjacency matrix, i.e.\ $P(i,j) = | \{ i' \rightarrow j' \in E : i' = i, j' = j \} | / d_{in}(j)$, and set $Q = (1-\eta) I + \eta P$, where $I$ is the identity matrix of appropriate dimension. (Recall from Section \ref{secGraphModel} that $E$ is in general a multi-set; hence, the numerator in $P(i,j)$ may exceed 1.) Next, for $t \in \N$, let $s_t$ denote the collection of signals $\{ s_t(i) \}_{i \in A \cup B}$ in vector form. Finally, for $i \in A$ define
\begin{equation} \label{eqVarThetaGraph}
\vartheta_{T_n}(i) = \frac{1}{T_n} \sum_{t = 0}^{T_n-1} s_{T_n-t} Q^t e_i^{\trans} .
\end{equation}
We note that \eqref{eqVarThetaGraph} can be rewritten as
\begin{equation} \label{eqVarThetaGraphExpand}
\vartheta_{T_n}(i) = \frac{1}{T_n} \sum_{t = 0}^{T_n-1} \sum_{j \in A} s_{T_n-t}(j) e_j Q^t e_i^{\trans} ,
\end{equation}
where we have used the fact that $s_t(j) = 0\ \forall\ t \in \N, j \in B$. From this expression, it is clear that $\vartheta_{T_n}(i)$ only depends on the structure of the $T_n$-step neighborhood into $i$ (since only this sub-graph affects the $e_j Q^t e_i^{\trans}$ terms) and on certain signals within this neighborhood, as mentioned above. We can then establish the following. 
\begin{lemma} \label{lemGraphBeliefApprox}
Given \ref{assBranchHorizon}, $\forall\ \epsilon > 0\ \exists\ N$ s.t. $\forall\ n \geq N$, $| \theta_{T_n}(i) - \vartheta_{T_n}(i) | < \epsilon\ a.s.\ \forall\ i \in A$.
\end{lemma}
\begin{proof}
See Appendix \ref{appProofGraphBeliefApprox}.
\end{proof} 

Before defining the aforementioned branching process, we formally define the graph construction described in Section \ref{secGraphModel}. For this, we will use the following additional notation.
\begin{itemize}
\item We let $A_l, l \in \N_0$ denote the set of agents at distance $l$ from the initial agent $i^*$, i.e.\ $i \in A_l$ means a path from $i$ to $i^*$ of length $l$ exists, but no shorter path exists (hence, $A_0 = \{i^*\}$, $A_1 = N_{in}(i^*)$, etc.). Similarly, we let $B_l , l \in \N_0$ denote the set of bots at distance $l$ from $i^*$.
\item We let $\{ (i,j) : j  \in [d_{out}(i)] \}$ denote the set of outstubs belonging to $i \in A$; we let $O_A$ denote the set of all such outstubs.
\item For each $(i,j) \in O_A$, we define a label $g( (i,j) ) \in \{1,2,3\}$ as follows:
\begin{equation} \label{eqOutstubLabels}
g( (i,j) ) = \begin{cases} 1, & \textrm{$i$ does not yet belong to graph} \\ 2, & \textrm{$i$ belongs to graph but $(i,j)$ has not been paired} \\ 3, & \textrm{$i$ belongs to graph and $(i,j)$ has been paired} \end{cases} .
\end{equation}
We will explain the utility of these labels shortly.
\end{itemize}
With this notation in place, we present the formal graph construction as Algorithm \ref{algGraph}. We offer some further comments to help explain the algorithm:
\begin{itemize}
\item The algorithm takes as input the degree sequence $\{ d_{out}(i) , d_{in}^A(i) , d_{in}^B(i) \}_{i \in A}$, which is used in Line \ref{algGraphInitSet} to define $O_A$. Also in Line \ref{algGraphInitSet}, we label all outstubs as 1 (since no agents have been added to the graph), and we initialize the set of bots to the empty set.
\item In Line \ref{algGraphFirstNode}, we sample the agent $i^*$ from which the graph construction begins. Since $i^*$ then belongs to the graph, we change the labels of its outstubs to 2.
\item For the remainder of the algorithm, we proceed in a breadth-first-search fashion, looping over distance $l$ and agents $i$ at distance $l$ from $i^*$. For each such agent, we do the following:
\begin{itemize}
\item For each of the $d_{in}^A(i)$ instubs of $i$ intended for pairing with agent outstubs, we sample an agent outstub uniformly (Line \ref{algGraphInitSample}), resampling until an unpaired outstub (i.e.\ one with label 1 or 2) has been found (Line \ref{algGraphResample}). Upon finding such an outstub, denoted $(i',j')$, we pair it with $i$'s instub to form an edge from $i'$ to $i$ (Line \ref{algGraphAgentPair}). Note that $g((i',j')) = 1$ implies $i'$ was added to the graph when edge $i' \rightarrow i$ was formed; hence, because $i \in A_l$, $i'$ is at distance $l+1$ from $i^*$ and must be added to $A_{l+1}$ (Line \ref{algGraphUpdate1}). Finally, we update the labels of the outstubs of $i'$ via \eqref{eqOutstubLabels} (Lines \ref{algGraphUpdate1}-\ref{algGraphUpdate2}). (Line \ref{algGraphTau} will be used in the branching process approximation and will be discussed shortly.)
\item For each of the $d_{in}^B(i)$ instubs of $i$ intended for pairing with bot outstubs, we add a new bot with a self-loop and an unpaired outstub to the set of bots, updating $B_{l+1}$ accordingly (Line \ref{algGraphAddBot}), and then add an edge from the new bot to $i$ (Line \ref{algGraphBotPair}). Note here that $B = \emptyset$ at the start of the construction; it follows that the $k$-th bot added to the graph is $n+k+1$, so $B = n + [ \sum_{i \in A} d_{in}^B(i) ]$ is the set of bots at the end of the construction.
\item Finally, if all agent outstubs have been paired, the construction terminates (Line \ref{algGraphTerm}).
\end{itemize}
\end{itemize}

\begin{algorithm}[t]
\DontPrintSemicolon
\caption{$\texttt{Graph-Construction}$ } \label{algGraph}

%\KwIn{$\{ d_{out}(i) , d_{in}^A(i) , d_{in}^B(i) \}_{i \in A}$}

Set $O_A = \{ (i,j) : i \in A, j \in [d_{out}(i)] \}$, $g( (i,j) ) = 1\ \forall\ (i,j) \in O_A$, $B = \emptyset$ \label{algGraphInitSet}

Sample $i^*$ uniformly from $A$; set $g( (i^*,j) ) = 2\ \forall\ j \in [ d_{out}(i^*) ]$; set $A_0 = \{ i^* \}$ \label{algGraphFirstNode}

\For{$l = 0$ \KwTo $\infty$}{

Set $A_{l+1} =B_{l+1} = \emptyset$

\For{$i \in A_{l}$}{

% check here if walk currently at $i$; if so, let $j \sim [d_{in}(i)]$ uniformly and let next step of walk be whatever connects to $(i,j)$. in this case, $X_i$ realized before degrees of $X_i$, as in random walk analysis

\For{$j = 1$ \KwTo $d_{in}^A(i)$}{

Sample $(i',j')$ from $O_A$ uniformly \label{algGraphInitSample}

\lIf{$g( (i',j') ) \neq 1$ \KwAnd $\tau_n = \infty$}{set $\tau_n = l$} \label{algGraphTau}

\lWhile{$ g((i',j')) = 3$}{sample $(i',j')$ from $O_A$ uniformly} \label{algGraphResample}

Add directed edge from $i'$ to $i$ \label{algGraphAgentPair}

\lIf{$g( (i',j') ) = 1$}{set $A_{l+1} = A_{l+1} \cup \{ i' \}$, $g ( (i',j') ) = 3$, $g( (i',j'') ) = 2\ \forall\ j'' \in [ d_{out}(i') ] \setminus \{ j' \}$} \label{algGraphUpdate1}
\lElseIf{$g( (i',j') ) = 2$}{set $g( (i',j') ) = 3$} \label{algGraphUpdate2}

}

\For{$j = 1$ \KwTo $d_{in}^B(i)$}{

Add bot $b = n + |B| + 1$ with self-loop and unpaired outstub, set $B = B \cup \{b\}, B_{l+1} = B_{l+1} \cup \{b\}$ \label{algGraphAddBot}

Add directed edge from $b$ to $i$ \label{algGraphBotPair}

}

\lIf{$g((i',j')) = 3\ \forall\ (i',j') \in O_A$}{\KwReturn} \label{algGraphTerm}

}

}

\end{algorithm}

We now return to discuss Line \ref{algGraphTau} of Algorithm \ref{algGraph}. Here $\tau_n$ denotes the first iteration at which an outstub with label 2 or 3 is sampled for pairing with an instub. Put differently, $\tau_n > l$ means that for the first $l$ iterations of the construction, only outstubs with label 1 have been sampled. This has two consequences. First, no edges have been added between two nodes both at distance $\leq l$ from $i^*$, i.e.\ the $l$-step incoming neighborhood of $i^*$ is a tree (except for the self-loops attached to bots). Second, no resampling of outstubs has occurred (Line \ref{algGraphResample}); this implies that the outstub $(i',j')$ paired in Line \ref{algGraphAgentPair} is chosen uniformly from $O_A$, so the degrees $(d_{out}(i'), d_{in}^A(i'), d_{in}^B(i'))$ of $i'$ are distributed according to the out-degree distribution $f_n$ defined in \eqref{eqEmpDist}.

These observations motivate a tree construction that we define next. In particular, we will construct a tree (except for bot self-loops) with edges pointing towards the root. Agents will be added to the tree with degrees sampled from $f_n$, except for the root node, whose degrees are sampled from $f_n^*$ (also defined in \eqref{eqEmpDist}), corresponding to the degrees of $i^*$ in the graph construction.

The tree construction requires further notation. First, we let $\hat{A}_l$ ($\hat{B}_l$, respectively) denote agents (bots, respectively) at distance $l$ from the tree's root. We also set $\hat{A} = \cup_{l=0}^{\infty} \hat{A}_l, \hat{B} = \cup_{l=0}^{\infty} \hat{B}_l$. (Here and moving forward, we use $\hat{\cdot}$ to distinguish tree-related objects from similarly-defined graph-related ones.) At times, we will use branching process terminology and e.g.\ refer to $\hat{A}_l$ as the $l$-th \textit{generation} of agents. We let $\phi$ denote the root node, so that $\hat{A}_0 = \{ \phi \}$. We will denote generic node in $\hat{A}_l \cup \hat{B}_l$ as $\i \in \N^l$; here $\i = (i_1,\ldots,i_l)$ encodes the ancestry of $\i$, i.e.\ $(i_1,\ldots,i_l)$ is the child of $(i_1,\ldots,i_{l-1})$, who is in turn the child of $(i_1,\ldots,i_{l-2})$, etc. Finally, for such $\i$ and for $j \in \N$, $(\i,j) = (i_1,\ldots,i_l,j)$ is the concatenation operation and $\i | j = (i_1,\ldots,i_j)$ denotes $\i$'s ancestor in generation $j$, with $\i | 0 = \phi$ by convention (note also that $\i | l = \i$ for such $\i$).

With this notation in place, we define the tree construction in Algorithm \ref{algTree}. We offer several more explanatory comments:
\begin{itemize}
\item Lines \ref{algTreeInitWalk} and \ref{algTreeUpdateWalkStart}-\ref{algTreeUpdateWalkEnd} define a particular random walk that will be used in Appendix \ref{appRootBeliefProofOutline}; they do not affect the tree structure and we defer further explanation to Appendix \ref{appRootBeliefProofOutline}.
\item As mentioned above, the root node $\phi$ has degrees sampled from $f_n^*$ (Line \ref{algTreeInitStruct}), while all other agents have degrees sampled from $f_n$ (Line \ref{algTreeDegSample}).
\item In Line \ref{algTreeAgentPair}, a directed edge is added from $(\i,j)$ to $\i$; the other $d_{out}((\i,j))-1$ outstubs of $(\i,j)$ are left unpaired so that the tree structure is preserved (except for bot self-loops).
\item At the conclusion of the $l$-th iteration, $\i \in \hat{A}_l$ has incoming neighbor set (offspring, in the branching process terminology) $\{ (\i,j) : j \in [ d_{in}^A(\i) + d_{in}^B(\i) ] \}$. More specifically, the subset $(\i,1),\ldots,(\i,d_{in}^A(\i) )$ of $\i$'s incoming neighbors are agents (Line \ref{algTreeAgentPair}), while the subset $(\i,d_{in}^A(\i)+1),\ldots,(\i,d_{in}^A(\i)+d_{in}^B(\i))$ of $\i$'s incoming neighbors are bots (Line \ref{algTreeAddBot}).
\item Unlike the graph construction, the tree construction continues indefinitely, yielding an infinite tree (except for bot self-loops) with edges pointing towards the root node $\phi$.
\end{itemize}

\begin{algorithm}[t]
\DontPrintSemicolon
\caption{$\texttt{Tree-Construction}$ } \label{algTree}

%\KwIn{$\{ d_{out}(i) , d_{in}^A(i) , d_{in}^B(i) \}_{i \in A}$}

Define $f_n, f_n^*$ via \eqref{eqEmpDist}, set $\hat{A}_0 = \{ \phi \}$, sample $(d_{out}(\phi),d_{in}^A(\phi),d_{in}^B(\phi))$ from $f_n^*$ \label{algTreeInitStruct}

Set $X_0^1 = X_0^2 = \phi$ \label{algTreeInitWalk}

\For{$l = 0$ \KwTo $\infty$}{

Set $\hat{A}_{l+1} =\hat{B}_{l+1} = \emptyset$

\For{$\i \in \hat{A}_{l}$}{ 

\For{$k \in \{1,2\}$}{ \label{algTreeUpdateWalkStart}

\If{$X_l^k = \i$}{

Sample $j^*$ from $[ d_{in}^A(\i) + d_{in}^B(\i) ]$ uniformly, set $X_{l+1}^k = (\i,j^*)$ \label{algTreeNextStep}

\lIf{$j^* > d_{in}^A(\i)$}{set $X_{l'}^k = (\i,j^*)\ \forall\ l' \in \{l+2,l+3,\cdots\}$} \label{algTreeAbsorb}

}

} \label{algTreeUpdateWalkEnd}

\For{$j = 1$ \KwTo $d_{in}^A(\i)$}{

Sample $(d_{out}((\i,j)),d_{in}^A((\i,j)),d_{in}^B((\i,j)))$ from $f_n$ \label{algTreeDegSample}

Add directed edge from $(\i,j)$ to $\i$, set $\hat{A}_{l+1} = \hat{A}_{l+1} \cup \{ (\i,j) \}$ \label{algTreeAgentPair}

}

\For{$j = 1$ \KwTo $d_{in}^B(\i)$}{

Add bot $b = (\i,d_{in}^A(\i)+j)$ with self-loop and unpaired outstub, set $\hat{B}_{l+1} = \hat{B}_{l+1} \cup \{b\}$ \label{algTreeAddBot}

Add directed edge from $b$ to $i$ \label{algTreeBotPair}

}

}

}

\end{algorithm}

Having defined the tree construction, we also define $\hat{\vartheta}_{T_n}(\phi)$ as in \eqref{eqVarThetaGraph} but using the tree from Algorithm \ref{algTree} instead of the graph from Algorithm \ref{algGraph}. Specifically, we let
\begin{equation} \label{eqVarThetaTree}
\hat{\vartheta}_{T_n}(\phi) = \frac{1}{T_n} \sum_{t = 0}^{T_n-1} \hat{s}_{T_n-t} \hat{Q}^t e_{\phi}^{\trans} = \frac{1}{T_n} \sum_{t = 0}^{T_n-1} \sum_{\i \in \hat{A}} \hat{s}_{T_n-t}(\i) e_{\i} \hat{Q}^t e_{\phi}^{\trans} ,
\end{equation}
where $\hat{s}_t(\i) \sim \textrm{Bernoulli}(\theta)\ \forall\ t \in \N, \i \in \hat{A}$; $\hat{s}_t(\i) = 0\ \forall\ t \in \N, \i \in \hat{B}$; $\hat{Q} = (1-\eta) I + \eta \hat{P}$; and $\hat{P}$ is the column-normalized adjacency matrix of the tree from Algorithm \ref{algTree}. We pause to note that
\begin{equation} \label{eqBeliefIn01}
0 \leq \hat{\vartheta}_{T_n}(\phi) \leq \frac{1}{T_n} \sum_{t=0}^{T_n-1} \mathbf{1} \hat{Q}^t e_{\phi}^{\trans} = 1 ,
\end{equation}
where the first inequality holds since \eqref{eqVarThetaTree} is a sum of nonnegative terms, the second follows since $\sum_{\i \in \hat{A}} \hat{s}_{T_n-t}(\i) e_{\i} \leq \mathbf{1}$ component-wise (where $\mathbf{1}$ is the all ones vectors) and since $\hat{Q}^t e_{\phi}^{\trans}$ is element-wise nonnegative, and the equality holds by column stochasticity of $\hat{Q}$.

We can now state Lemma \ref{lemGraphTreeRelation}, which relates the estimate of a uniformly random agent in the graph with the estimate of the root node in the tree. For the first statement in the lemma, we argue that, conditioned on $\tau_n > T_n$, the $T_n$-step neighborhood of $i^*$ in the graph and the $T_n$-step neighborhood of $\phi$ in the tree are constructed via the same procedure; since the signals are defined in the same manner as well, this implies $\vartheta_{T_n}(i^*)$ and $\hat{\vartheta}_{T_n}(\phi)$ have the same distribution. The second statement of the lemma says that the condition $\tau_n > T_n$ occurs with high probability; it is essentially implied by \cite[Lemma 5.4]{chen2017generalized}. We note that the assumptions \ref{assGraphDegSeq} and \ref{assGraphHorizon} are required for this second statement to hold, and are standard assumptions needed to locally approximate a sparse random graph construction with a branching process. Finally, we recall $\zeta < 1/2$ by \ref{assGraphHorizon}, which is why the limit shown in Lemma \ref{lemGraphTreeRelation} holds.
\begin{lemma} \label{lemGraphTreeRelation}
Assume \ref{assGraphDegSeq} and \ref{assGraphHorizon} hold, and let $\stackrel{\mathscr{D}}{=}$ denote equality in distribution. Then
\begin{equation}
\vartheta_{T_n}(i^*) | \{ \tau_n > T_n \} \stackrel{\mathscr{D}}{=} \hat{\vartheta}_{T_n}(\phi) , \quad \P(\tau_n \leq T_n | \Omega_{n,1} ) = O \left( n^{ \zeta - 1/2 } \right) \xrightarrow[n \rightarrow \infty]{} 0.
\end{equation}
\end{lemma}
\begin{proof}
See Appendix \ref{appProofGraphTreeRelation}.
\end{proof}

We can now state and prove Lemma \ref{lemStep1Main}, which is the main result for Step 1 of the proofs of the theorems. This result will allow us to analyze convergence of $\theta_{T_n}(i^*)$ (the graph agent estimate) by instead analyzing convergence of $\hat{\vartheta}_{T_n}(\phi)$ (the tree agent estimate).
\begin{lemma} \label{lemStep1Main}
Assume \ref{assGraphDegSeq}, \ref{assGraphHorizon}, and \ref{assBranchHorizon} hold. Then $\forall\ x \in \R$ and all $n \in \N$ sufficiently large,
\begin{equation}
\P ( | \theta_{T_n}(i^*) - x | > \epsilon ) \leq \P ( | \hat{\vartheta}_{T_n}(\phi) - x | > \epsilon / 2  ) + \P ( \Omega_{n,1}^C ) + O \left( n^{ \zeta - 1/2 } \right) .
\end{equation}
\end{lemma}
\begin{proof}
First, given $\epsilon  > 0$, we have for sufficiently large $n$,
\begin{align}
\P ( | \theta_{T_n}(i^*) - x | > \epsilon ) & \leq \P ( | \theta_{T_n}(i^*) - \vartheta_{T_n}(i^*) | + | \vartheta_{T_n}(i^*) - x | > \epsilon ) \leq \P ( | \vartheta_{T_n}(i^*) - x | > \epsilon / 2 ) ,
\end{align}
where the first inequality uses the triangle inequality and in the second we used Lemma \ref{lemGraphBeliefApprox} to bound $| \theta_{T_n}(i^*) - \vartheta_{T_n}(i^*) |$ by $\epsilon/2$ $a.s.$ Furthermore, by the law of total probability, we have
\begin{align}
\P ( | \vartheta_{T_n}(i^*) - x | > \epsilon / 2 ) & \leq \P ( | \vartheta_{T_n}(i^*) - x | > \epsilon / 2 | \tau_n > T_n ) + \P ( \tau_n \leq T_n | \Omega_{n,1} ) + \P ( \Omega_{n,1}^C ) .
\end{align}
Combining the previous two inequalities and using Lemma \ref{lemGraphTreeRelation} (which applies since \ref{assGraphDegSeq}, \ref{assGraphHorizon} are assumed to hold), we obtain
\begin{align} 
\P ( | \theta_{T_n}(i^*) - x | > \epsilon ) & \leq \P ( | \hat{\vartheta}_{T_n}(\phi) - x | > \epsilon/2 ) + O \left( n^{ \zeta - 1/2 } \right) + \P ( \Omega_{n,1}^C )  ,
\end{align}
which is what we set out to prove.
\end{proof}

Before proceeding, we state another lemma that will be used in Step 2 of the proofs for both theorems. This lemma uses the fact that each agent in the tree has a unique path to the root. As a result, we can obtain an alternate expression for the terms $e_{\i} \hat{Q}^t e_{\phi}^{\trans}$ appearing in \eqref{eqVarThetaTree}.
\begin{lemma} \label{lemTreeBeliefExpand}
For each $n \in \N$,
\begin{equation} \label{eqTreeBeliefExpand}
\hat{\vartheta}_{T_n}(\phi) = \frac{1}{T_n} \sum_{t=0}^{T_n-1} \sum_{l=0}^t { t \choose l } \eta^l (1-\eta)^{t-l}  \sum_{\i \in \hat{A}_l}  \hat{s}_{T_n-t}(\i) \prod_{j=0}^{l-1} d_{in}(\i|j)^{-1}\ a.s.,
\end{equation}
where by convention $\prod_{j=0}^{l-1} d_{in}(\i|j)^{-1} = 1$ when $l = 0$.
\end{lemma}
\begin{proof}
See Appendix \ref{appProofTreeBeliefExpand}.
\end{proof}

\subsection{Step 2 for proof of Theorem \ref{thmMain}} \label{appRootBeliefProofOutline}

Our next goal is to establish convergence of $\hat{\vartheta}_{T_n}(\phi)$, from which convergence of $\theta_{T_n}(i^*)$ will follow via Lemma \ref{lemStep1Main}. For this, we will use Chebyshev's inequality, so we begin with two lemmas describing the limiting behavior of the mean and variance of $\hat{\vartheta}_{T_n}(\phi)$. Here and moving forward, for random variables $X$ and $Y$ we use $\Var_n(X) = \E_n [ X^2 ] - ( \E_n [ X ] )^2$ and $\Cov_n(X,Y) = \E_n [ X Y ] - \E_n [ X ] \E_n [ Y ]$ to denote variance and covariance conditional on the degree sequence.
\begin{lemma} \label{lemFirstMomBelief}
Given \ref{assBranchDegSeq} and \ref{assBranchHorizon}, we have the following:
\begin{align}
\lim_{n \rightarrow \infty} T_n(1-p_n) = 0 & \Rightarrow \lim_{n \rightarrow \infty} | \E_n [ \hat{\vartheta}_{T_n}(\phi) ] - \theta | 1(\Omega_{n,2}) = 0\ a.s. \\
\lim_{n \rightarrow \infty} T_n(1-p_n) = c \in (0,\infty) & \Rightarrow \lim_{n \rightarrow \infty} \left| \E_n [ \hat{\vartheta}_{T_n}(\phi) ] - \theta \frac{1-e^{-c \eta}}{c \eta} \right| 1(\Omega_{n,2}) = 0\ a.s. \\
\lim_{n \rightarrow \infty} T_n(1-p_n) = \infty & \Rightarrow \lim_{n \rightarrow \infty} | \E_n [ \hat{\vartheta}_{T_n}(\phi) ] | 1(\Omega_{n,2}) = 0\ a.s.
\end{align}
\end{lemma}
\begin{proof}
See Appendix \ref{appProofFirstMomBelief}.
\end{proof}

\begin{lemma} \label{lemSecondMomBelief}
Given \ref{assBranchDegSeq} and \ref{assBranchHorizon}, $\lim_{n \rightarrow \infty} \Var_n ( \hat{\vartheta}_{T_n}(\phi) ) 1(\Omega_{n,2}) = 0\ a.s.$
\end{lemma}
\begin{proof}
See Appendix \ref{appProofSecondMomBelief}.
\end{proof}

Before proceeding, we briefly describe our approach to proving these lemmas. First, we note that in analyzing the moments of $\hat{\vartheta}_{T_n}(\phi)$, the i.i.d.\ Bernoulli random variables $\hat{s}_{T_n-t}(\i)$ in \eqref{eqTreeBeliefExpand} are easily dealt with; the difficulty arises from the $\prod_{j=0}^{l-1} d_{in}(\i|j)^{-1}$ terms. Luckily, there is a simple interpretation of these terms that guides our analysis and that proceeds as follows. First, define a random walk $\{ X_l^1 \}_{l \in \N_0}$ with $X_0^1 = \phi$ and $X_l^1$ chosen uniformly from the incoming neighbors of $X_{l-1}^1$, for each $l \in \N$. Then, as shown in \eqref{eqWalkProbAppears} in Appendix \ref{appProofFirstMomBelief},
\begin{equation}
\E \sum_{\i \in \hat{A}_l} \prod_{j=0}^{l-1} d_{in}(\i|j)^{-1} = \P ( X_l^1 \in \hat{A}_l ) .
\end{equation}
In short, computing the mean of $\hat{\vartheta}_{T_n}(\phi)$ amounts to computing hitting probabilities of the form $\P ( X_l^1  \in \hat{A}_l )$. Similarly, to analyze the second moment of $\hat{\vartheta}_{T_n}(\phi)$, we compute hitting probabilities of the form $\P ( X_l^1 \in \hat{A}_l , X_l^2 \in \hat{A}_l )$, where $X_l^2$ is defined in the same manner as $X_l^1$ and is conditionally independent of $X_l^1$ given the tree structure. We note that, in principal, the $k$-th moment of $\hat{\vartheta}_{T_n}(\phi)$ can be computed by analyzing $k$ walks. However, the calculations become exceedingly complex as $k$ grows, and because we only require two moments, we do not study any case $k > 2$.

This interpretation explains Lines \ref{algTreeInitWalk} and \ref{algTreeUpdateWalkStart}-\ref{algTreeUpdateWalkEnd} of Algorithm \ref{algTree}: in Line \ref{algTreeInitWalk}, we begin two random walks at the root node $\phi$; each time Lines \ref{algTreeUpdateWalkStart}-\ref{algTreeUpdateWalkEnd} are reached, we advance the random walks one step. Importantly, we simultaneously sample the walks and construct the tree. In particular, the $l$-th step of the walk is taken at Line \ref{algTreeNextStep}, \textit{before} the degrees of the corresponding node are realized in Line \ref{algTreeDegSample}; this is crucial to our computation of the aforementioned hitting probabilities. Finally, we note that in Line \ref{algTreeAbsorb} of Algorithm \ref{algTree}, the condition $j^* > d_{in}^A(\i)$ implies the walk reaches the set of bots $\hat{B}$; since bots have self-loops but no other incoming edges, they act as absorbing states on the  walk. This is why the entire future trajectory of the walk can be defined in Line \ref{algTreeAbsorb}.

In Lemmas \ref{lemHitProbOneWalk} and \ref{lemHitProbTwoWalks}, we compute the hitting probabilities needed for the proofs of Lemmas \ref{lemFirstMomBelief} and \ref{lemSecondMomBelief}. We note that, in addition to the random variables $\tilde{p}_n, \tilde{p}_n^*, \tilde{q}_n$ defined in \eqref{eqWalkRVsMain} in Section \ref{secGraphModel}, Lemma \ref{lemHitProbTwoWalks} requires the definition of several similar random variables; we define these in \eqref{eqWalkRVs} (and also recall the definitions of $\tilde{p}_n, \tilde{p}_n^*, \tilde{q}_n$ for convenience). We discuss these in more detail shortly.
\begin{equation} \label{eqWalkRVs}
\begin{aligned}
\tilde{p}_n & = \sum_{j \in \N, k \in \N_0} \frac{j}{j+k} \sum_{i \in \N} f_n(i,j,k) \\
\tilde{q}_n & = \sum_{j \in \N, k \in \N_0} \frac{j}{j+k} \frac{1}{j+k} \sum_{i \in \N} f_n(i,j,k) \\
\tilde{r}_n & = \sum_{j \in \N, k \in \N_0} \frac{j}{j+k} \frac{j-1}{j+k} \sum_{i \in \N} f_n(i,j,k) 
\end{aligned} \quad\quad
\begin{aligned}
\tilde{p}_n^* & = \sum_{j \in \N, k \in \N_0} \frac{j}{j+k} \sum_{i \in \N} f_n^*(i,j,k) \\
\tilde{q}_n^* & = \sum_{j \in \N, k \in \N_0} \frac{j}{j+k} \frac{1}{j+k} \sum_{i \in \N} f_n^*(i,j,k) \\
\tilde{r}_n^* & = \sum_{j \in \N, k \in \N_0} \frac{j}{j+k} \frac{j-1}{j+k} \sum_{i \in \N} f_n^*(i,j,k) 
\end{aligned}
\end{equation}

\begin{lemma} \label{lemHitProbOneWalk}
We have
\begin{equation}
\P_n ( X_l^1 \in \hat{A} ) = \begin{cases} \tilde{p}_n^* \tilde{p}_n^{l-1}  , & l \in \N \\ 1 , & l = 0 \end{cases} .
\end{equation}
\end{lemma}
\begin{proof}
See Appendix \ref{appProofHitProbOneWalk}.
\end{proof}

\begin{lemma} \label{lemHitProbTwoWalks}
For $l' > l$, we have
\begin{equation}
\P_n ( X_l^1 \in \hat{A} , X_{l'}^2 \in \hat{A} ) = \begin{cases} \P_n ( X_l^1 \in \hat{A} , X_l^2 \in \hat{A} ) \tilde{p}_n^{l'-l}  , & l \in \N \\ \tilde{p}_n^* \tilde{p}_n^{l'-1} , & l = 0 \end{cases}.
\end{equation}
Furthermore,
\begin{equation} \label{eqHitProbTwoWalksState}
\P_n ( X_l^1 \in \hat{A} , X_l^2 \in \hat{A} ) = \begin{cases} \tilde{r}_n^* \tilde{p}_n^{2(l-1)}  +  \sum_{t=2}^{l} \tilde{q}_n^* \tilde{q}_n^{t-2} \tilde{r}_n \tilde{p}_n^{2(l-t)} + \tilde{q}_n^* \tilde{q}_n^{l-1}  , & l \in \{2,3,\ldots\} \\ \tilde{r}_n^* + \tilde{q}_n^* , & l = 1  \\ 1 , & l = 0 \end{cases} .
\end{equation}
\end{lemma}
\begin{proof}
See Appendix \ref{appProofHitProbTwoWalks}.
\end{proof}
Before proceeding, we comment on the form of \eqref{eqHitProbTwoWalksState}, which helps explain the definitions in \eqref{eqWalkRVs}. Namely, in \eqref{eqHitProbTwoWalksState}, $\tilde{r}_n^* \tilde{p}_n^{2(l-1)}$ is the probability of the two random walks visiting different agents on the first step of the walk ($\tilde{r}_n^*$ term), then separately remaining in the agent set for the next $l-1$ steps of the walk ($\tilde{p}_n^{2(l-1)}$ term); similarly, $\tilde{q}_n^* \tilde{q}_n^{t-2} \tilde{r}_n \tilde{p}_n^{2(l-t)}$ is the probability of the walks visiting the same agents for $t-1$ steps ($\tilde{q}_n^* \tilde{q}_n^{t-2}$ term), then visiting a different agent on the $t$-th step ($\tilde{r}_n$ term), then separately remaining in the agent set for $l-t$ steps ($\tilde{p}_n^{2(l-t)}$ term); finally, $ \tilde{q}_n^* \tilde{q}_n^{l-1}$ is the probability of the walks remaining together and in the agent set for $l$ steps. Each of these arguments follows from \eqref{eqWalkRVs}: $\tilde{p}_n$ gives the probability of a single walk proceeding to an agent ($j/(j+k)$ term), $\tilde{q}_n$ gives the probability of two walks proceeding to the same agent ($j/(j+k)$ term for the first walk, $1/(j+k)$ term for the second walk), and $\tilde{r}_n$ gives the probability of two walks proceeding to different agents ($j/(j+k)$ term for the first walk, $(j-1)/(j+k)$ term for the second walk). Similar arguments apply to $\tilde{p}_n^*, \tilde{q}_n^*, \tilde{r}_n^*$, except these pertain to the first steps of the walks.

Equipped with Lemmas \ref{lemFirstMomBelief} and \ref{lemSecondMomBelief}, we can prove Theorem \ref{thmMain}. First, suppose $T_n(1-p_n) \rightarrow 0$. Given $\epsilon > 0$, we can use Lemma \ref{lemStep1Main} to obtain (provided the limits exist)
\begin{align} \label{eqLemStep1MainCor}
\lim_{n \rightarrow \infty} \P ( | \theta_{T_n}(i^*) - \theta | > \epsilon ) & \leq  \lim_{n \rightarrow \infty} \left( \P ( | \hat{\vartheta}_{T_n}(\phi) - \theta | > \epsilon / 2  ) + \P ( \Omega_{n,1}^C ) + O \left( n^{ \zeta - 1/2 } \right) \right) \\
& =  \lim_{n \rightarrow \infty}  \P ( | \hat{\vartheta}_{T_n}(\phi) - \theta | > \epsilon / 2  ) ,
\end{align}
where we have used $\P(\Omega_{n,1}^C) \rightarrow 0$ by \ref{assGraphDegSeq} and $\zeta < 1/2$ by \ref{assGraphHorizon}. Next, using total probability,
\begin{equation} \label{eqLimitBeliefRootExpand1}
\P ( | \hat{\vartheta}_{T_n}(\phi) - \theta | > \epsilon /2 ) \leq \P ( | \hat{\vartheta}_{T_n}(\phi) - \theta | > \epsilon/2 , \Omega_{n,2} ) + \P( \Omega_{n,2}^C ) .
\end{equation}
We can further expand the first summand in \eqref{eqLimitBeliefRootExpand1} as 
\begin{align} 
& \P ( | \hat{\vartheta}_{T_n}(\phi) - \theta | > \epsilon/2 , \Omega_{n,2} ) \leq \P ( | \hat{\vartheta}_{T_n}(\phi) - \E_n \hat{\vartheta}_{T_n}(\phi) | + | \E_n \hat{\vartheta}_{T_n}(\phi) - \theta | > \epsilon/2 , \Omega_{n,2} ) \\
%& \leq \P \left( \left( \left\{ | \hat{\vartheta}_{T_n}(\phi) - \E_n \hat{\vartheta}_{T_n}(\phi) | > \frac{\epsilon}{2} \right\} \cup \left\{ | \E_n \hat{\vartheta}_{T_n}(\phi) - \theta | > \frac{\epsilon}{2} \right\} \right) \cap \Omega_{n,2} \right) \\
& \quad \leq \P \left( | \hat{\vartheta}_{T_n}(\phi) - \E_n \hat{\vartheta}_{T_n}(\phi) | > \frac{\epsilon}{4} , \Omega_{n,2} \right) + \P \left( | \E_n \hat{\vartheta}_{T_n}(\phi) - \theta | > \frac{\epsilon}{4} , \Omega_{n,2} \right) , \label{eqLimitBeliefRootExpand2}
\end{align}
where we have simply used the triangle inequality and the union bound. Now for the first summand in \eqref{eqLimitBeliefRootExpand2}, we have (via total expectation and the conditional form of Chebyshev's inequality)
\begin{align} \label{eqLimitBeliefRootCheb}
\P \left( | \hat{\vartheta}_{T_n}(\phi) - \E_n \hat{\vartheta}_{T_n}(\phi) | > \frac{\epsilon}{4} , \Omega_{n,2} \right) & = \E \left[ \P_n \left( | \hat{\vartheta}_{T_n}(\phi) - \E_n \hat{\vartheta}_{T_n}(\phi) | > \frac{\epsilon}{4} \right) 1 ( \Omega_{n,2} ) \right] \\
&  \leq \frac{16}{\epsilon^2} \E \left[ \Var_n ( \hat{\vartheta}_{T_n}(\phi) ) 1 ( \Omega_{n,2} ) \right] \xrightarrow[n \rightarrow \infty]{} 0 ,
\end{align}
where the limit holds by Lemma \ref{lemSecondMomBelief}. For second summand in \eqref{eqLimitBeliefRootExpand2}, we write
\begin{align} \label{eqLimitBeliefMean}
\P \left( | \E_n \hat{\vartheta}_{T_n}(\phi) - \theta | > \frac{\epsilon}{4} , \Omega_{n,2} \right) & = \E \left[ 1 \left( | \E_n \hat{\vartheta}_{T_n}(\phi) - \theta | > \frac{\epsilon}{4}  \right) 1 ( \Omega_{n,2} ) \right] \\
& \leq \frac{4}{\epsilon} \E [ | \E_n \hat{\vartheta}_{T_n}(\phi) - \theta |  1 ( \Omega_{n,2} ) ] \xrightarrow[n \rightarrow \infty]{} 0 ,
\end{align}
where the first two lines use total expectation and the inequality $1(x > y) \leq x/y$ for $x , y > 0$ (which is easily proven by considering the cases $x > y$ and $x \leq y$), and the limit holds by Lemma \ref{lemFirstMomBelief}. Finally, combining \eqref{eqLemStep1MainCor}, \eqref{eqLimitBeliefRootExpand1}, \eqref{eqLimitBeliefRootExpand2}, \eqref{eqLimitBeliefRootCheb}, and \eqref{eqLimitBeliefMean}, and recalling that $\P(\Omega_{n,2}^C) \rightarrow 0$ by \ref{assBranchDegSeq}, we obtain
\begin{equation}
0 \leq \lim_{n \rightarrow \infty} \P ( | \theta_{T_n}(i^*) - \theta | > \epsilon ) \leq \lim_{n \rightarrow \infty}  \P ( | \hat{\vartheta}_{T_n}(\phi) - \theta | > \epsilon / 2  ) = 0 .
\end{equation}
Since $\epsilon > 0$ was arbitrary, we conclude that $\theta_{T_n}(i^*)$ converges to $\theta$ in probability, completing the proof in the case $T_n(1-p_n) \rightarrow 0$. For the cases $T_n(1-p_n) \rightarrow c \in (0,\infty)$ and $T_n(1-p_n) \rightarrow \infty$, respectively, we can replace $\theta$ with $\theta (1-e^{-c \eta}) / (c \eta )$ and $0$, respectively (the corresponding cases from Lemma \ref{lemFirstMomBelief}), but otherwise follow the same approach. 

\subsection{Step 2 for proof of Theorem \ref{thmSec}} \label{appRootBelief2ProofOutline}

Similar to the second step in the proof of Theorem \ref{thmMain}, we begin by analyzing the limiting behavior of $\hat{\vartheta}_{T_n}(\phi)$. However, we will use a different approach than that used in Theorem \ref{thmMain}. This approach is made possible by the stronger assumptions of Theorem \ref{thmSec}, and it will yield a fast rate of convergence that will allow us to prove the theorem.

To explain our approach, we first recall that Lemma \ref{lemTreeBeliefExpand} states
\begin{equation} 
\hat{\vartheta}_{T_n}(\phi) = \frac{1}{T_n} \sum_{t=0}^{T_n-1} \sum_{l=0}^t { t \choose l } \eta^l (1-\eta)^{t-l}  \sum_{\i \in \hat{A}_l}  \hat{s}_{T_n-t}(\i) \prod_{j=0}^{l-1} d_{in}(\i|j)^{-1} .
\end{equation}
Hence, letting $\T$ denote the collection of random variables defining the tree structure,
\begin{align} \label{eqMeanTreeBeliefGivenTree}
\E [ \hat{\vartheta}_{T_n}(\phi) | \T ] & = \frac{1}{T_n} \sum_{t=0}^{T_n-1} \sum_{l=0}^t { t \choose l } \eta^l (1-\eta)^{t-l}  \sum_{\i \in \hat{A}_l} \E [ \hat{s}_{T_n-t}(\i) | \T ] \prod_{j=0}^{l-1} d_{in}(\i|j)^{-1} \\
& = \frac{\theta}{T_n} \sum_{t=0}^{T_n-1} \sum_{l=0}^t { t \choose l } \eta^l (1-\eta)^{t-l}  \sum_{\i \in \hat{A}_l} \prod_{j=0}^{l-1} d_{in}(\i|j)^{-1} ,
\end{align}
where we have simply used the fact that the signals are i.i.d.\ $\textrm{Bernoulli}(\theta)$ random variables. Our basic approach will now proceed in two steps. First, in Lemma \ref{lemSecThmSignals} we condition on the tree structure, so that $\hat{\vartheta}_{T_n}(\phi)$ is simply a weighted sum of i.i.d.\ $\textrm{Bernoulli}(\theta)$ random variables; the lemma shows that this weighted sum is close to its conditional mean $\E [ \hat{\vartheta}_{T_n}(\phi) | \T ]$ with high probability. Second, in Lemma \ref{lemSecThmStructure}, we show that the conditional mean $\E [ \hat{\vartheta}_{T_n}(\phi) | \T ]$ converges to zero in probability. Before proceeding, we also note that an argument similar to \eqref{eqBeliefIn01} implies
\begin{equation} \label{eqBeliefGivenTreeIn0theta}
0 \leq \E [ \hat{\vartheta}_{T_n}(\phi) | \T ] \leq \theta\ a.s.,
\end{equation}
which will be used in the proofs of the lemmas in this appendix.

We now state Lemma \ref{lemSecThmSignals}. As mentioned, the proof involves analyzing a weighted sum of i.i.d.\ random variables; hence, our analysis is similar to the derivation of Hoeffding's inequality.
\begin{lemma} \label{lemSecThmSignals}
Assume $\exists\ \mu > 0$ and $N' \in \N$ independent of $n$ s.t.\ the following hold:
\begin{itemize}
\item \ref{assBranchHorizon}, with $T_n \geq \mu \log n\ \forall\ n \geq N'$.
\end{itemize}
Then $\forall\ \epsilon > 0$,
\begin{equation}
\P ( | \hat{\vartheta}_{T_n}(\phi) - \E [ \hat{\vartheta}_{T_n}(\phi) | \T ] | > \epsilon ) = O \left( n^{-2 \epsilon^2 \mu} \right) .
\end{equation}
\end{lemma}
\begin{proof}
See Appendix \ref{appProofSecThmSignals}.
\end{proof}

Lemma \ref{lemSecThmStructure} states that conditional mean $\E [ \hat{\vartheta}_{T_n}(\phi) | \T ]$ converges to zero in probability. Note that the only source of randomness in $\E [ \hat{\vartheta}_{T_n}(\phi) | \T ]$ is the tree structure. Since the tree structure is generated recursively, $\E [ \hat{\vartheta}_{T_n}(\phi) | \T ]$ has a martingale-like structure; this allows us to use an approach similar to the Azuma-Hoeffding inequality for bounded-difference martingales.
\begin{lemma} \label{lemSecThmStructure}
Assume $\exists\ \kappa, \mu > 0$ and $N' \in \N$ independent of $n$ s.t.\ the following hold:
\begin{itemize}
\item \ref{assBranchDegSeq}, with $P(\Omega_{n,2}) = O(n^{-\kappa})$ and $p < 1$.
\item \ref{assBranchHorizon}, with $T_n \geq \mu \log n\ \forall\ n \geq N'$.
\end{itemize}
Then $\forall\ \epsilon > 0$,
\begin{equation}
\P( \E [ \hat{\vartheta}_{T_n}(\phi) | \T ] > \epsilon) = O \left( n^{- \min \{ \mu ( \epsilon \eta (1-p) / \theta )^2 , \kappa \} } \right) .
\end{equation}
\end{lemma}
\begin{proof}
See Appendix \ref{appProofSecThmStructure}. 
\end{proof}

With Lemmas \ref{lemSecThmSignals} and \ref{lemSecThmStructure} in place, we can prove Theorem \ref{thmSec}. First, since $\theta_{T_n}(i^*), \hat{\vartheta}_{T_n}(\phi) \geq 0$, taking $x = 0$ in Lemma \ref{lemStep1Main} yields
\begin{align} 
\P ( \theta_{T_n}(i^*) > \epsilon ) & \leq \P (  \hat{\vartheta}_{T_n}(\phi)  > \epsilon / 2  ) + \P ( \Omega_{n,1}^C ) + O \left( n^{ \zeta - 1/2 } \right) \\
& = \P (  \hat{\vartheta}_{T_n}(\phi)  > \epsilon / 2  ) + O \left( n^{-\kappa} \right) + O \left( n^{ \zeta - 1/2 } \right) \label{eqLemStep1appThmSec} ,
\end{align}
where the equality is by the theorem assumptions. For the first summand in \eqref{eqLemStep1appThmSec}, we write
\begin{align}
\P (  \hat{\vartheta}_{T_n}(\phi)  > \epsilon / 2  ) & = \P (  ( \hat{\vartheta}_{T_n}(\phi) - \E [  \hat{\vartheta}_{T_n}(\phi) | \T ] ) +   \E [  \hat{\vartheta}_{T_n}(\phi) | \T ] > \epsilon / 2  ) \\
& \leq \P (  | \hat{\vartheta}_{T_n}(\phi) - \E [  \hat{\vartheta}_{T_n}(\phi) | \T ] | + \E [  \hat{\vartheta}_{T_n}(\phi) | \T ] > \epsilon / 2  ) \\
& \leq \P (  | \hat{\vartheta}_{T_n}(\phi) - \E [  \hat{\vartheta}_{T_n}(\phi) | \T ] | > \epsilon / 4  ) +  \P (   \E [  \hat{\vartheta}_{T_n}(\phi) | \T ] > \epsilon / 4 ) \\
& = O \left( n^{-\epsilon^2 \mu / 8} + n^{-\min \{ \mu ( \epsilon \eta (1-p) / \theta )^2 / 16 , \kappa \} } \right) = O \left( n^{-\min \{ \mu ( \epsilon \eta (1-p) / \theta )^2 / 16 , \kappa \} } \right)  ,
\end{align}
where the first equality adds and subtracts a term, the first inequality is immediate, the second inequality uses the union bound, the second equality uses Lemmas \ref{lemSecThmSignals} and \ref{lemSecThmStructure}, and the final equality holds since $\eta,p \in (0,1)$ implies $\epsilon^2 \mu / 8 >   \mu  ( \epsilon \eta (1-p) / \theta )^2 / 16$. Substituting into \eqref{eqLemStep1appThmSec},
\begin{equation} \label{eqThmSecFinalTailBound}
\P ( \theta_{T_n}(i^*) > \epsilon ) = O \left( n^{-\min \{ (1/2) - \zeta , \mu ( \epsilon \eta (1-p) / \theta )^2 / 16 , \kappa \} } \right) .
\end{equation}
We can then write
\begin{align}
\E  \left| \left\{ i \in [n] : \theta_{T_n}(i) > \epsilon \right\} \right| & = \sum_{i \in [n]} \E 1 ( \theta_{T_n}(i) > \epsilon ) = \sum_{i \in [n]} \P ( \theta_{T_n}(i) > \epsilon )  \\
& = n \P ( \theta_{T_n}(i^*) > \epsilon ) = O \left( n^{1-\min \{ (1/2) - \zeta , \mu ( \epsilon \eta (1-p) / \theta )^2 / 16 , \kappa \} } \right) ,
\end{align}
where we have used \eqref{eqThmSecFinalTailBound}. Hence, by Markov's inequality,
\begin{align}
\P \left( \left| \left\{ i \in [n] : \theta_{T_n}(i) > \epsilon \right\} \right| > K n^k \right) & \leq K^{-1} n^{-k} \E  \left| \left\{ i \in [n] : \theta_{T_n}(i) > \epsilon \right\} \right| \\
& = O \left( n^{-k + ( 1 - \min \{ (1/2) - \zeta , \mu ( \epsilon \eta (1-p) / \theta )^2 / 16 , \kappa \} ) }  \right) \xrightarrow[n \rightarrow \infty]{} 0 ,
\end{align}
where the limit holds by the assumption on $k$ in the statement of the theorem.

\subsection{Other remarks}

\subsubsection{A sufficient condition for extending Theorem \ref{thmSec}} \label{appExtendSecThmToOthers}

Here we show that the condition \eqref{eqSuffCond} from Appendix \ref{secSecondResult} is sufficient to extend Theorem \ref{thmSec} to other cases of $p_n$. Recall this condition is
\begin{equation} \label{eqSuffCond2}
\exists\ \gamma' > 0\ s.t.\ \P( | \E [ \hat{\vartheta}_{T_n}(\phi) | \T ] - L(p_n) |  > \epsilon) = O \left( n^{-\gamma'} \right) ,
\end{equation}
where $L(p_n)$ is the limit from Theorem \ref{thmMain} based on the relative asymptotics of $T_n$ and $p_n$, i.e.\
\begin{equation} \label{eqLofPn}
L(p_n) = \begin{cases} \theta , & T_n(1-p_n) \xrightarrow[n \rightarrow \infty]{} 0 \\  \theta (1-e^{-c \eta})/(c \eta) , & T_n(1-p_n) \xrightarrow[n \rightarrow \infty]{} c \in (0,\infty) \\ 0 , & T_n(1-p_n) \xrightarrow[n \rightarrow \infty]{} \infty \end{cases} .
\end{equation}
Suppose \eqref{eqSuffCond2} holds in the case $T_n(1-p_n) \rightarrow 0$, so that $L(p_n) = \theta$. In this case, we have
\begin{align}
& \P ( | \theta_{T_n}(i^*) - \theta | > \epsilon ) \leq \P ( | \hat{\vartheta}_{T_n}(\phi) - \theta | > \epsilon / 2 ) + O \left( n^{- \min \{ \kappa, (1/2) - \zeta \} } \right) \\ 
& \quad \leq  \P ( | \hat{\vartheta}_{T_n}(\phi) - \E [ \hat{\vartheta}_{T_n}(\phi) | \T ] | > \epsilon / 4 ) + \P ( | \E [ \hat{\vartheta}_{T_n}(\phi) | \T ] - \theta | > \epsilon / 4 )  + O \left( n^{- \min \{ \kappa, (1/2) - \zeta \} } \right) \\
& \quad \leq  O \left( n^{-\epsilon^2 \mu / 8 } \right) + O \left( n^{-\gamma'} \right) + O \left( n^{- \min \{ \kappa, (1/2) - \zeta \} } \right) = O \left( n^{-\min \{  \epsilon^2 \mu / 8 , \gamma' , \kappa , (1/2)-\zeta \} } \right) , 
\end{align}
where the first inequality is Lemma \ref{lemStep1Main} (which holds for all cases of $p_n$) with $\P(\Omega_{n,1}) = O(n^{-\kappa})$ and the third uses Lemma \ref{lemSecThmSignals} (which holds for all cases of $p_n$) and the sufficient condition \eqref{eqSuffCond2}. Hence, by the argument following \eqref{eqThmSecFinalTailBound}, we obtain for any $\epsilon > 0$, $K > 0$, and $k' > 1 - \min \{  \epsilon^2 \mu / 8 , \gamma' , \kappa , (1/2)-\zeta \}$,
\begin{equation}
\P \left( | \{ i \in [n] : | \theta_{T_n}(i) - \theta | > \epsilon \} | > K n^{k'} \right) \xrightarrow[n \rightarrow \infty]{} 0 ,
\end{equation}
i.e.\ Theorem \ref{thmSec} holds with $k$ replaced by $k'$. The same argument shows that Theorem \ref{thmSec} holds (with only a change of $k$) in the cases $T_n(1-p_n) \rightarrow c \in (0,\infty)$ and $T_n(1-p_n) \rightarrow \infty$ with $p_n \rightarrow 1$.

\subsubsection{Comparing Step 2 for proofs of Theorems \ref{thmMain} and \ref{thmSec}}  \label{appStep2compare}

As shown in Appendices \ref{appRootBeliefProofOutline} and \ref{appRootBelief2ProofOutline}, Step 2 for the proofs of both theorems involves bounding $\P ( | \hat{\vartheta}_{T_n}(\phi) - L(p_n) | > \epsilon / 2 )$ for the appropriate $L(p_n)$. One may wonder why we have conducted a different analysis for the two theorems. The reason is that, as shown in Appendix \ref{appWhereFailWhenPnTo1}, the analysis for Step 2 of Theorem \ref{thmSec} yields a bound that does not decay with $n$ in the case $T_n(1-p_n) \rightarrow c \in [0,\infty)$. Hence, we have derived a bound for Theorem \ref{thmMain} that encompasses all cases of $\lim_{n \rightarrow \infty} T_n(1-p_n)$. On the other hand, the bound from Theorem \ref{thmMain} only states $\P ( | \hat{\vartheta}_{T_n}(\phi) - L(p_n) | > \epsilon / 2 ) \rightarrow 0$ but does not provide a rate of convergence so cannot be used to prove Theorem \ref{thmSec}. We also note Appendix \ref{appWhereFailWhenPnTo1} shows that, while the bound for Step 2 of Theorem \ref{thmSec} \textit{does} decay in $n$ for the case $T_n(1-p_n) \rightarrow \infty$ with $p_n \rightarrow 1$, it does not decay quickly enough to establish \eqref{eqSuffCond}.

%% file: adversarialDetails.tex
\section{Section \ref{secAdversarial} proof and experiment details} \label{appAdvDetails}

\subsection{Solution of the relaxed problem} \label{proofLemRelaxedProblem}

We aim to show $d_n^{rel} \in \argmin_{ d \in \R_+^n : \sum_{i=1}^n d(i) \leq b_n } \tilde{p}_n(d)$, where (we recall)
\begin{equation}
d_n^{rel}(i) = d_{in}^A(i) \left( \frac{ \sqrt{r(i)} }{ h^* } - 1 \right)_+\ \forall\ i \in [n]  .
\end{equation}
We also recall $x_+ = x 1 ( x > 0)$, $r(i) = d_{out}(i) / d_{in}^A(i)\ \forall\ i \in [n]$, $h^* = \max_{x \in \R_+} h(x)$, and
\begin{equation}
h(x) = \frac{ \sum_{i \in [n] : r(i) \geq x^2} \sqrt{ d_{out}(i) d_{in}^A(i) } }{ b_n + \sum_{i \in [n] : r(i) \geq x^2} d_{in}^A(i) }\ \forall\ x \in \R_+ .
\end{equation}
First note strict convexity of $y \mapsto 1 / y$  for $y \in \R$ implies strict convexity of $\tilde{p}_n$, i.e.\ for any $d \neq d' \in \R^n_+$ and $\rho \in (0,1)$,
\begin{equation}
\tilde{p}_n(  \rho d + (1-\rho) d' ) < \rho \tilde{p}_n(d) + (1-\rho) \tilde{p}_n(d') .
\end{equation}
Also note we can rewrite the relaxed problem \eqref{eqRelaxedOpt} as
\begin{equation} \label{eqRewriteProblem}
\min_{d \in \R^n} \tilde{p}_n(d)\ s.t.\ c(d) \leq 0 , c_i(d) \leq 0\ \forall\ i \in [n] ,   
\end{equation}
where $c(d) = \sum_{i=1}^n d(i) - b_n, c_i(d) = -d(i)$. Given $\lambda , \lambda_i \geq 0$, we also define the Lagrangian
\begin{equation}
L(d,\lambda,\lambda_1,\ldots,\lambda_n) = \tilde{p}_n(d) + \lambda c(d) + \sum_{i=1}^n \lambda_i c_i(d) .
\end{equation}
Finally, we set $\lambda^* = ( h^* )^2 / m_n, \lambda_i^* = ( ( h^* )^2 - r(i) )_+ / m_n$ (clearly, $\lambda^*, \lambda_i^* \geq 0$). Now to prove the theorem, it suffices to establish the following \textit{KKT conditions} (see e.g.\ \cite[Section 5.5.3]{boyd2004convex}):
\begin{enumerate}
\item $c( d_n^{rel} ) , c_1 ( d_n^{rel} ) , \ldots , c_n ( d_n^{rel} ) \leq 0\ \forall\ i \in [n]$,  i.e.\ $d_n^{rel}$ is a feasible point of \eqref{eqRewriteProblem}.
\item $\nabla L ( d_n^{rel} , \lambda^*, \lambda_1^* , \ldots , \lambda_n^* ) = 0$, i.e.\ the first-order condition is satisfied.
\item $\lambda^* c ( d_n^{rel} ) = \lambda_1^* c_1 ( d_n^{rel} ) = \cdots = \lambda_n^* c_n ( d_n^{rel} ) = 0$, i.e.\ complementary slackness holds.
\end{enumerate}
We proceed to the proofs of these three statements.
\begin{enumerate}
\item Clearly, $c_i(d_n^{rel} ) \leq 0\ \forall\ i \in [n]$. To show $c( d_n^{rel} ) \leq 0$, we claim (and will return to prove) that $h^*$ is a fixed point of $h$, i.e.\ $h^* = h(h^*)$. Assuming this claim holds, we have
\begin{align} \label{eqConstraintWithEq}
c \left( d_n^{rel} \right) & = \frac{1}{h^*} \sum_{i \in [n] : r(i) \geq ( h^* )^2 } d_{in}^A(i) \sqrt{r(i)} - \sum_{i \in [n] : r(i) \geq ( h^* )^2 } d_{in}^A(i) - b_n \\
& = \frac{1}{ h(h^*) }  \sum_{i \in [n] : r(i) \geq ( h^* )^2 } \sqrt{d_{out}(i) d_{in}^A(i)} - \sum_{i \in [n] : r(i) \geq ( h^* )^2 } d_{in}^A(i) - b_n = 0 ,
\end{align}
where the last two equalities use the fixed point claim and the definition of $h$, respectively.
\item First, let $i \in [n]$ satisfy $r(i) > ( h^* )^2$, so that $d_n^{rel}(i) = - d_{in}^A(i) + d_{in}^A(i) \sqrt{r(i)}/h^* , \lambda_i^* = 0$. Then
\begin{align}
& \frac{ \partial L }{ \partial d(i) } \left( d_n^{rel} , \lambda^*, \lambda_1^*, \ldots, \lambda_n^* \right) = - \frac{ d_{out}(i) d_{in}^A(i) }{ m_n}  \frac{1}{ ( d_{in}^A(i) + d_n^{rel}(i)  )^2 } + \lambda^*  \\
&  \quad = - \frac{ d_{out}(i) d_{in}^A(i) }{ m_n ( d_{in}^A(i) \sqrt{r(i)} / h^* )^2 } + \frac{( h^* )^2}{m_n} = - \frac{  d_{out}(i) d_{in}^A(i) }{ m_n \left( \sqrt{ d_{out}(i) d_{in}^A(i) } / h^* \right)^2 } + \frac{( h^* )^2}{m_n}  = 0 .
\end{align}
Next, let $i \in [n]$ satisfy $r(i) \leq ( h^* )^2$, so that $d_n^{rel}(i) = 0, \lambda_i = ( ( h^* )^2 - r(i) ) / m_n$. Then
\begin{align}
& \frac{ \partial L }{ \partial d(i) } \left( d_n^{rel} , \lambda^*, \lambda_1^*, \ldots, \lambda_n^* \right) = - \frac{ d_{out}(i) d_{in}^A(i) }{ m_n}  \frac{1}{ ( d_{in}^A(i)  )^2 } + \lambda^* - \lambda_i^* = - \frac{r(i)}{m_n} + \frac{( h^* )^2}{m_n} - \frac{ ( h^* )^2 - r(i) }{ m_n } = 0 .
\end{align}
\item For any $i \in [n]$, we have
\begin{equation}
\lambda_i^* c_i ( d_n^{rel} ) = -d_{in}^A(i) \left( \frac{ ( h^* )^2 - r(i) }{ m_n } \right)_+ \left( \frac{ \sqrt{r(i)} }{ h^* } - 1 \right)_+ .
\end{equation}
Clearly, the first $(\cdot)_+$ term is zero if $r(i) > ( h^* )^2$, the second is zero if $r(i) < ( h^* )^2$, and both are zero if $r(i) = ( h^* )^2$. Finally, $\lambda^* c ( d_n^{rel} ) = 0$ holds by \eqref{eqConstraintWithEq}.
\end{enumerate}

We return to establish the fixed point claim. We in fact prove the slightly stronger result
\begin{equation} \label{eqSubFixedPt}
h (x ) \leq h ( h ( x ) )\ \forall\ x \in \R_+ .
\end{equation}
The fixed point claim then follows, since $h^* \geq h ( h^* )$ by definition and $h^* \leq h ( h^* )$ by \eqref{eqSubFixedPt} with  $x = x^*$, where $x^*$ is a maximizer of $h$. Thus, it suffices to prove \eqref{eqSubFixedPt}. Towards this end, fix $x \in \R_+$. We first assume $x \geq h(x)$ and will return to address the other case. For any $y,z \in \R \cup \{ \infty \}$, we define
\begin{equation}
N(y,z) = \sum_{i \in [n] : r(i) \in [y^2,z^2)} \sqrt{d_{out}(i) d_{in}^A(i)} , \quad D(y,z) =  \sum_{i \in [n] : r(i) \in [y^2,z^2)} d_{in}^A(i) ,
\end{equation}
where by convention $N(y,z) = D(y,z) = 0$ if $y,z$ are such that $\{ i \in [n] : r(i) \in [y^2,z^2) \} = \emptyset$ (i.e.\ if the sums are over empty sets). Then by definition of $h$, $N$, and $D$, we have
\begin{equation}
h( h(x)) = \frac{N(h(x),\infty)}{b_n+D(h(x),\infty)} = \frac{N(x,\infty) + N(h(x),x)}{ b_n + D(x,\infty) + D(h(x),x)} .
\end{equation}
Again by definition of $h$, $N$, and $D$, and recalling $r(i) = d_{out}(i) / d_{in}^A(i)$, we also have
\begin{align}
N(h(x),x) & = \sum_{i \in [n] : r(i) \in [h(x)^2,x^2)} \sqrt{r(i)} d_{in}^A(i) \geq h(x) \sum_{i \in [n] : r(i) \in [h(x)^2,x^2)} d_{in}^A(i) \\
& = h(x) D(h(x),x) = \frac{N(x,\infty) }{b_n+D(x,\infty)} D(h(x),x)
\end{align}
Thus, combining the previous two equations, we obtain
\begin{equation}
h(h(x)) \geq \frac{N(x,\infty) + \frac{N(x,\infty) }{b_n+D(x,\infty)} D(h(x),x) }{ b_n + D(x,\infty) + D(h(x),x)} = \frac{N(x,\infty) }{b_n+D(x,\infty)} = h(x) .
\end{equation}
If instead $x \leq h(x)$, we can use the same argument to obtain
\begin{align}
& h( h(x)) =  \frac{N(x,\infty) - N(x,h(x))}{ b_n + D(x,\infty) - D(x,h(x))} , \quad N(x,h(x)) \leq \frac{N(x,\infty) }{b_n+D(x,\infty)} D ( x,h(x)) , \\
& \Rightarrow h(h(x)) \geq \frac{ N(x,\infty) - \frac{N(x,\infty) }{b_n+D(x,\infty)} D ( x,h(x)) }{ b_n + D(x,\infty) - D(x,h(x)) } = \frac{N(x,\infty) }{b_n+D(x,\infty)} = h(x) .
\end{align}

\subsection{Rewriting the objective function} \label{appGn}
 
We aim to prove \eqref{eqGnW} and \eqref{eqEGnW}, which we restate here for convenience:
\begin{equation} \label{eqRewriteObj}
g_n ( W ) = \frac{ m_n }{ \bar{r} } ( 1 - \tilde{p}_n ( d_n^{rand}) ) , \quad 
\E g_n ( W ) \geq \frac{ m_n }{ 2 \bar{r} } ( 1 - \tilde{p}_n ( d_n^{rel}) ) .
\end{equation}
For the equality in \eqref{eqRewriteObj}, we write
\begin{align}
g_n(W) & = \frac{1}{\bar{r}} \sum_{j=1}^{b_n} \frac{ d_{out}(W_j) }{ d_{in}^A(W_j) + \sum_{k=1}^{b_n} 1 ( W_k = W_j ) } = \frac{1}{\bar{r}} \sum_{j=1}^{b_n} \sum_{i=1}^n 1 ( W_j = i ) \frac{ d_{out}(i) }{ d_{in}^A(i) + \sum_{k=1}^{b_n} 1 ( W_k = i ) } \\
& = \frac{1}{\bar{r}} \sum_{i=1}^n \frac{ d_{out}(i) \sum_{j=1}^{b_n} 1 ( W_j = i ) }{ d_{in}^A(i) + \sum_{k=1}^{b_n} 1 ( W_k = i )  }  = \frac{1}{\bar{r}} \sum_{i=1}^n \frac{ d_{out}(i) d_n^{rand}(i) }{ d_{in}^A(i) + d_n^{rand}(i)  }  = \frac{ m_n }{ \bar{r} } ( 1 - \tilde{p}_n ( d_n^{rand}) ) ,
\end{align}
where the first, fourth, and fifth equalities hold by definition of $g_n$ (see proof of Theorem \ref{thmConstApprox}), $d_n^{rand}$ (see Algorithm \ref{algApprox}), and $\tilde{p}_n(d)$ (see \eqref{eqTildePnFunction}), respectively, and the others are straightforward. For the inequality in \eqref{eqRewriteObj}, we first write (simliar to above)
\begin{equation} \label{eqRewriteEgnW}
\E g_n(W) = \frac{1}{\bar{r}} \sum_{i=1}^n \sum_{j=1}^{b_n} d_{out}(i) \E\left[  \frac{ 1 ( W_j = i )  }{ d_{in}^A(i) + \sum_{k=1}^{b_n} 1 ( W_k = i ) } \right]  .
\end{equation}
Now fix $i \in [n]$ and $j \in [b_n]$. Then by the smoothing property,
\begin{align}
\E \left[ \frac{ 1 ( W_j = i )  }{ d_{in}^A(i) + \sum_{k=1}^{b_n} 1 ( W_k = i ) } \right] & = \E \left[ \E \left[ \frac{ 1 ( W_j = i )  }{ d_{in}^A(i) + \sum_{k=1}^{b_n} 1 ( W_k = i ) } \middle| \{ W_k \}_{k \in [b_n] \setminus \{j\} }  \right] \right] \\
& = \E \left[ \frac{  \P ( W_j = i | \{ W_k \}_{k \in [b_n] \setminus \{j\} } )  }{ d_{in}^A(i) + 1 + \sum_{k \in [b_n] \setminus \{j\} } 1 ( W_k = i ) }\right] .
\end{align}
We next observe
\begin{equation}
\P ( W_j = i | \{ W_k \}_{k \in [b_n] \setminus \{j\} } ) = \P ( W_j = i ) = \frac{ d_n^{rel}(i)}{  \sum_{k = 1}^n d_n^{rel}(k) } = \frac{ d_n^{rel}(i)}{  b_n } ,
\end{equation}
where we used independence of $\{ W_j \}_{j=1}^k$, Algorithm \ref{algApprox}, and \eqref{eqConstraintWithEq}, respectively. Combining the previous two identities,
\begin{align}
\E \left[ \frac{ 1 ( W_j = i )  }{ d_{in}^A(i) + \sum_{k=1}^{b_n} 1 ( W_k = i ) } \right]   & = \frac{1}{b_n} \E \left[\frac{ d_n^{rel}(i) }{  d_{in}^A(i) + 1 + \sum_{k \in [b_n] \setminus \{j\} } 1 ( W_k = i )  } \right] \\
& \geq \frac{1}{b_n} \frac{ d_n^{rel}(i) }{  d_{in}^A(i) + 1 + \sum_{k \in [b_n] \setminus \{j\} } \P ( W_k = i )  }\\
& = \frac{1}{b_n} \frac{ d_n^{rel}(i) }{  d_{in}^A(i) + 1 + \frac{b_n-1}{b_n} d_n^{rel}(i)  } > \frac{1}{2 b_n} \frac{ d_n^{rel}(i) }{  d_{in}^A(i) + d_n^{rel}(i)  }
\end{align}
where the first inequality is Jensen's and the second holds since $1 \leq d_{in}^A(i)$ and $\frac{b_n-1}{b_n} < 2$. Substituting into \eqref{eqRewriteEgnW}, we obtain
\begin{equation}
\E g_n(W) \geq \frac{1}{ 2 \bar{r}} \sum_{i=1}^n \frac{ d_{out}(i) d_n^{rel}(i) }{  d_{in}^A(i) + d_n^{rel}(i)  }  = \frac{m_n}{ 2 \bar{r}} ( 1 - \tilde{p}_n ( d_n^{rel} ) ) ,
\end{equation}
where the equality holds by definition of $\tilde{p}_n$ \eqref{eqTildePnFunction}.

\subsection{Self-bounding concentration} \label{appSelfBounding}

As mentioned in the main text, we exploit the theory of \textit{self-bounding} functions.
\begin{definition} \label{defnSelfBound}
\cite[Section 3.3]{boucheron2013concentration} Let $\mathcal{X}$ be a measurable space, $l \in \N$, and $f : \mathcal{X}^l \rightarrow [0,\infty)$. We say $f$ is a \textit{self-bounding} function if there exists auxiliary functions $f_{-i} : \mathcal{X}^{l-1} \rightarrow \R, i \in [l]$ such that, for any $x = ( x_1, \ldots, x_l ) \in \mathcal{X}^l$,
\begin{gather}
0 \leq f(x) - f_{-i}(x_{-i}) \leq 1\ \forall\  i \in [l] , \quad \sum_{i=1}^l \left( f(x) - f_{-i}(x_{-i}) \right) \leq f(x) ,
\end{gather}
where $x_{-i} = ( x_1, \ldots , x_{i-1} , x_{i+1} , \ldots , x_l )\ \forall\ i \in [l]$.
\end{definition}

\begin{theorem} \label{thmSelfBounding}
\cite[Theorem 6.12]{boucheron2013concentration} Let $X_1, \ldots, X_l$ be independent $\mathcal{X}$-valued random variables, define $X = (X_1,\ldots,X_l)$, and let $f : \mathcal{X}^l \rightarrow [0,\infty)$ be self-bounding. Then for any $t \in (0, \E f(X) ]$,
\begin{equation}
\P ( f(X) \leq \E f(X) - t ) \leq \exp \left( - \frac{t^2}{2 \E f(X) } \right) .
\end{equation}
\end{theorem}

Assuming for the moment that $g_n$ is self-bounding, we can apply the theorem with $t = \delta \E g_n(W) / (2+\delta)$ to obtain
\begin{equation} 
\P \left( g_n(W) \leq \frac{2 \E g_n(W)}{2+\delta} \right) = \P \left( g_n(W) \leq \E g_n(W) - \frac{\delta \E g_n(W)}{2+\delta} \right)  \leq \exp \left( \frac{ - ( \frac{\delta \E g_n(W)}{2+\delta}  )^2 }{ 2 \E g_n(W) } \right) = \exp \left( - 2 c_{\delta} \E g_n(W) \right) ,
\end{equation}
where $c_{\delta} = \frac{ \delta^2 }{ 4(2+\delta)^2}$ as in Theorem \ref{thmConstApprox}. This completes the proof of \eqref{eqGnTail} from the main text.

To verify $g_n$ is self-bounding, we use the most obvious choice of auxiliary functions: let
\begin{equation}
g_{n,-i} ( w_{-i} ) = \frac{1}{ \max_{j \in [n]} r(j) } \sum_{j=1, j \neq i}^{b_n} \frac{ d_{out}(w_j) }{ d_{in}^A(w_j) + \sum_{k=1, k \neq i}^{b_n} 1 ( w_k = w_j ) } ,
\end{equation}
where $w_{-i} = ( w_1, \ldots, w_{i-1}, w_{i+1}, \ldots , w_{b_n} )$ for $w \in (w_1,\ldots,w_{b_n}) \in [n]^{b_n}$, i.e.\  we simply ignore the $i$-th coordinate of $w$. Towards bounding $g_n(w) - g_{n,-i}(w_{-i})$, we first observe
\begin{align}
& \sum_{j=1, j \neq i}^{b_n} d_{out}(w_j) \left( \frac{1}{d_{in}^A(w_j) + \sum_{k=1}^{b_n} 1 ( w_k = w_j ) } - \frac{1}{ d_{in}^A(w_j) + \sum_{k=1, k \neq i}^{b_n} 1 ( w_k = w_j ) } \right) \label{eqSelfBoundUglyTerm1} \\
& \quad = \sum_{j=1, j \neq i}^{b_n} \frac{ - 1 ( w_i = w_j ) d_{out}(w_j) }{   (d_{in}^A(w_j) + \sum_{k=1}^{b_n} 1 ( w_k = w_j )) ( d_{in}^A(w_j) + \sum_{k=1,k\neq i}^{b_n} 1 ( w_k = w_j ) )   }  \label{eqSelfBoundUglyTerm2}  \\
& \quad = \sum_{j=1, j \neq i}^{b_n} \frac{ - 1 ( w_i = w_j ) d_{out}(w_i) }{   (d_{in}^A(w_i) + \sum_{k=1}^{b_n} 1 ( w_k = w_i )) ( d_{in}^A(w_i) + \sum_{k=1,k\neq i}^{b_n} 1 ( w_k = w_i ) )   }  \label{eqSelfBoundUglyTerm3}  \\
& \quad =\frac{ -  d_{out}(w_i)  }{   d_{in}^A(w_i) + \sum_{k=1}^{b_n} 1 ( w_k = w_i ) } \times \frac{ \sum_{k=1,k\neq i}^{b_n} 1 ( w_k = w_i ) }{  d_{in}^A(w_i) + \sum_{k=1,k\neq i}^{b_n} 1 ( w_k = w_i )    } \label{eqSelfBoundUglyTerm4}  \\
& \quad \in \left( \frac{ -  d_{out}(w_i) }{ d_{in}^A(w_i) + \sum_{k=1}^{b_n} 1 ( w_k = w_i ) } , 0 \right) \label{eqSelfBoundUglyTerm5}  ,
\end{align}
where in \eqref{eqSelfBoundUglyTerm2} we computed the difference of fractions in \eqref{eqSelfBoundUglyTerm1}, in \eqref{eqSelfBoundUglyTerm3} we replaced $w_j$ by $w_i$ (which is permitted due to the indicator $1(w_i = w_j)$), and in \eqref{eqSelfBoundUglyTerm4} we rearranged the expression; the upper bound in \eqref{eqSelfBoundUglyTerm5} is obvious, while the lower bound holds since the second factor in \eqref{eqSelfBoundUglyTerm4} is less than 1. Using the upper bound in \eqref{eqSelfBoundUglyTerm5}, we can then obtain
\begin{align}
g_n(w) - g_{n,-i}(w_{-i}) & = \frac{  \sum_{j=1, j \neq i}^{b_n} d_{out}(w_j) \left( \frac{1}{d_{in}^A(w_j) + \sum_{k=1}^{b_n} 1 ( w_k = w_j ) } - \frac{1}{ d_{in}^A(w_j) + \sum_{k=1, k \neq i}^{b_n} 1 ( w_k = w_j ) } \right)  }{ \max_{j \in [n]} r(j) } \label{eqSelfBoundDiff1} \\
& \quad\quad + \frac{ 1}{ \max_{j \in [n]} r(j) }  \frac{ d_{out}(w_i) }{ d_{in}^A(w_i) + \sum_{k=1}^{b_n} 1 ( w_k = w_i ) }  \label{eqSelfBoundDiff2} \\ 
& < \frac{ 1}{ \max_{j \in [n]} r(j) }  \frac{ d_{out}(w_i) }{ d_{in}^A(w_i) + \sum_{k=1}^{b_n} 1 ( w_k = w_i ) } < \frac{ r(w_i) }{ \max_{j \in [n]} r(j) } \leq 1 . \label{eqSelfBoundDiffUpper}
\end{align}
On the other hand, using the lower bound in \eqref{eqSelfBoundUglyTerm5}, along with \eqref{eqSelfBoundDiff1}-\eqref{eqSelfBoundDiff2}, we immediately obtain $g_n(w) - g_{n,-i}(w_{-i})  > 0$. Together with \eqref{eqSelfBoundDiffUpper}, the first condition in Definition \ref{defnSelfBound} holds. To verify the second condition in Definition \ref{defnSelfBound}, we use the leftmost expression in \eqref{eqSelfBoundDiffUpper} to obtain
\begin{align}
\sum_{i=1}^{b_n} \left( g_n(w) - g_{n,-i}(w_{-i}) \right) < \frac{ 1}{ \max_{j \in [n]} r(j) } \sum_{i=1}^{b_n} \frac{ d_{out}(w_i) }{ d_{in}^A(w_i) + \sum_{k=1}^{b_n} 1 ( w_k = w_i ) } = g_n(w) .
\end{align}

\subsection{Proof of Corollary \ref{corConstApprox}} \label{proofCorConstApprox}

Let $\mathcal{I}_n = \{ i \in [n] : r(i) \geq \epsilon \bar{r} \}$; recall $|\mathcal{I}_n| \rightarrow \infty$ as $n \rightarrow \infty$ by assumption. Define $d_n \in \R_+^n$ by
\begin{equation}
d_n(i) =  \frac{d_{out}(i) 1 ( i \in \mathcal{I}_n ) b_n }{ \sum_{j \in \mathcal{I}_n } d_{out}(j) }  .
\end{equation}
Clearly, $\sum_{i=1}^n d_n(i) = b_n$, so $\tilde{p}_n ( d_n^{rel} ) \leq \tilde{p}_n ( d_n )$ (see Appendix \ref{proofLemRelaxedProblem}). Hence, by Theorem \ref{thmConstApprox}, we aim to find $\{ \delta_n \}_{n \in \N}$ s.t.\
\begin{equation} \label{eqCorProofGoal}
\lim_{n \rightarrow \infty} \delta_n = 0 , \quad \lim_{n \rightarrow \infty} c_{\delta_n} \frac{m_n}{ \bar{r} } ( 1 - \tilde{p}_n ( d_n ) ) = \infty ,
\end{equation}
where $c_{\delta_n} = \frac{\delta_n^2}{4(2+\delta_n)^2}$ as in Theorem \ref{thmConstApprox}. In fact, it suffices to show $m_n ( 1 - \tilde{p}_n ( d_n ) ) / \bar{r} \rightarrow \infty$, since then we choose $\delta_n$ such that (for example) $c_{\delta_n} = \sqrt{ \frac{ \bar{r} }{ m_n ( 1 - \tilde{p}_n ( d_n ) ) } }$ to ensure \eqref{eqCorProofGoal} holds. Toward this end, first note that for any $i \in \mathcal{I}_n$,
\begin{equation}
d_n(i) = \frac{ r(i) d_{in}^A(i) b_n }{ \sum_{j \in \mathcal{I}_n } d_{out}(j) } \quad \Rightarrow \quad \frac{ d_n(i) }{ d_n(i) + d_{in}^A(i) } = \frac{ 1 }{ 1 + \frac{ \sum_{j \in \mathcal{I}_n } d_{out}(j) }{ r(i) b_n }  }  \geq \frac{ 1 }{ 1 + \frac{ \sum_{j \in \mathcal{I}_n } d_{out}(j) }{ \epsilon \bar{r} b_n }  } .
\end{equation}
Hence, by definition of $\tilde{p}_n$ \eqref{eqTildePnFunction},
\begin{equation} \label{eqComputePnDn}
\frac{m_n}{ \bar{r} } ( 1 - \tilde{p}_n ( d_n ) ) = \frac{1}{ \bar{r} } \sum_{i \in [n]}  \frac{ d_{out} (i) d_n(i) }{ d_n(i) + d_{in}^A(i) } \geq  \frac{1}{ \bar{r} } \sum_{i \in \mathcal{I}_n } \frac{ d_{out}(i) }{ 1 + \frac{ \sum_{j \in \mathcal{I}_n } d_{out}(j) }{ \epsilon \bar{r} b_n }  } = \frac{ 1  }{ \frac{ \bar{r} }{ \sum_{i \in \mathcal{I}_n } d_{out}(i) }  + \frac{1}{ \epsilon b_n }   } .
\end{equation}
Recall $| \mathcal{I}_n | \rightarrow \infty$ by assumption, so
\begin{equation}
\frac{ \sum_{i \in \mathcal{I}_n } d_{out}(i) }{ \bar{r} } = \frac{ \sum_{i \in \mathcal{I}_n } d_{in}^A(i) r(i) }{ \bar{r} } \geq \epsilon \sum_{i \in \mathcal{I}_n } d_{in}^A(i) \geq \epsilon | \mathcal{I}_n | \rightarrow \infty .
\end{equation}
Since also $\epsilon b_n \rightarrow \infty$ by assumption, the final expression in \eqref{eqComputePnDn} diverges, as desired.

\subsection{Other algorithmic details} \label{appAlgDetails}

We first show Algorithm \ref{algExact} solves \eqref{eqOriginalOpt}. We require a basic fact about discrete convexity.
\begin{definition}
\cite[Section 1.4.2]{murota2003discrete} Let $f : \Z^n \rightarrow \R \cup \{ \infty \}$ and $dom(f) = \{ x \in \Z^n : f(x) \in \R \}$. Then $f$ is called \textit{M-convex} if for any $x, y \in dom(f)$ and any $i \in [n]$ satisfying $x(i) > y(i)$, there exists $j \in [n]$ satisfying
\begin{equation}
y(j) > x(j), \quad f(x) + f(y) \geq f ( x - e_i + e_j )  + f ( y + e_i - e_j ) .
\end{equation}
\end{definition}
\begin{theorem} \label{thmMoptimality}
\cite[Theorem 6.26]{murota2003discrete} Let $f$ be M-convex, and let $x \in dom(f)$. Then
\begin{equation}
f(x) \leq f(y)\ \forall\ y \in \Z^n  \Leftrightarrow  f(x) \leq f(x-e_i+e_j)\ \forall\ i,j \in [n] .
\end{equation}
\end{theorem}
In words, the theorem says that $x$ minimizes $f$ if and only if $f$ cannot be decreased by an ``exchange,'' wherein $x$ is replaced by $x-e_i+e_j$. Note that Algorithm \ref{algExact} terminates precisely when this criteria is satisfied, so if we can show that \eqref{eqRewriteOpt} is M-convex, we obtain as a corollary that Algorithm \ref{algExact} solves \eqref{eqOriginalOpt}. 

To show M-convexity, let $d,d' \in dom(\hat{p}_n), i \in [n]$ s.t.\ $d(i) > d'(i)$. Then since $\sum_{k=1}^n d(k) = \sum_{k=1}^n d'(k) = b_n$, we clearly have $d'(j) > d(j)$ for some $j \in [n]$. From $\sum_{k=1}^n d(k) = \sum_{k=1}^n d'(k) = b_n$ and $d(i), d'(j) \geq 1$, it is also clear that $d - e_i + e_j , d' + e_i - e_j \in dom(\hat{p}_n)$. Hence, letting $\mu(k) = d_{out}(k) d_{in}^A(k) / m_n$,
\begin{align} \label{eqEfficientObjectiveUpdate}
\hat{p}_n ( d - e_i + e_j ) & = \sum_{k \in [n] \setminus \{i,j\} } \frac{ \mu(k) }{ d_{in}^A(k) + d(k) } + \frac{ \mu(i) }{ d_{in}^A(i) + d(i) - 1 }  + \frac{ \mu(j) }{ d_{in}^A(j) + d(j) +1 } \\
&  = \hat{p}_n ( d) +  \frac{ \mu(i) }{ ( d_{in}^A(i) + d(i) - 1) ( d_{in}^A(i) + d(i) ) } - \frac{ \mu(j) }{ ( d_{in}^A(j) + d(j) +1 ) ( d_{in}^A(j) + d(j) ) } , 
\end{align}
where we have simply used the definitions of $\hat{p}_n, \tilde{p}_n$. Similarly, we obtain
\begin{equation}
\hat{p}_n ( d' + e_i - e_j ) = \hat{p}_n ( d' ) - \frac{ \mu(i) }{ ( d_{in}^A(i) + d'(i) + 1) ( d_{in}^A(i) + d'(i) ) } + \frac{ \mu(j) }{ ( d_{in}^A(j) + d'(j) -1 ) ( d_{in}^A(j) + d'(j) ) } .
\end{equation}
Adding the previous two equations, and using the inequalities $d(i) \geq d'(i)+1, d'(j) \geq d(j)+1$ (where the first holds since $d(i) > d'(i)$ and $d(i), d'(i) \in \Z$, and the second holds similarly) gives $\hat{p}_n ( d - e_i + e_j ) + \hat{p}_n ( d' + e_i - e_j ) \leq \hat{p}_n(d) + \hat{p}_n(d')$.

For the runtime of Algorithm \ref{algExact}, we note the following:
\begin{itemize}
\item The complexity of each iteration is dominated by the computation of $\{ \hat{p}_n ( d - e_{i} + e_{j} ) \}_{i,j \in [n] }$. By \eqref{eqEfficientObjectiveUpdate}, we can compute $\hat{p}_n ( d - e_{i} + e_{j} )$ in $O(1)$ time per $i,j$ pair, which yields $O(n^2)$ complexity per iteration.
\item In the best case, the initial choice of $d$ is actually a solution. However, it still requires one iteration to verify this, so the best-case complexity is $O(n^2)$.
\item In the general case, \cite[Section 10.1.1]{murota2003discrete} provides a tie-breaking rule for the choice of $(i^*,j^*)$ that guarantees termination in $\max \{ \| d - d' \|_1 : d, d' \in dom ( \hat{p}_n ) \} = O(b_n)$ iterations, which means $O(n^2 b_n)$ complexity.
\end{itemize}

For the randomized scheme (Algorithm \ref{algApprox}), first observe that by definition of $h$, $\{ h(x) \}_{x \in \R_+} = \{ h ( \sqrt{r(i)} ) \}_{ i \in [n] }$. Furthermore, $\{ h ( \sqrt{r(i)} ) \}_{ i \in [n] }$, and thus $\{ h(x) \}_{x \in \R_+}$, can be computed in time $O(n \log n)$ as follows: 
\begin{itemize}
\item Compute a vector containing $\{ r(i) \}_{ i \in [n] }$ sorted in decreasing order ($O(n \log n)$ time).
\item Iteratively compute the sums in \eqref{eqHmain} at each $x \in \{ \sqrt{r(i)} \}_{ i \in [n] }$ ($O(n)$ time).
\item Compute $\{ h ( \sqrt{r(i)} ) \}_{ i \in [n] }$ ($O(n)$ time).
\end{itemize}
In summary, $\{ h(x) \}_{x \in \R_+}$ (which contains at most $n$ elements) can be computed in $O(n \log n)$ time. After computing this set, $h^*$, and subsequently $d_n^{rel}$, can each be computed in linear time. Thus, computing the relaxed solution \eqref{eqRelaxedSolnMain} requires $O(n \log n)$ complexity. Finally, assuming we can obtain one sample from $d_n^{rel}$ in $O(1)$ time after $O(n \log n)$ pre-processing time (using e.g.\ the alias method \cite[Section 3.4.1]{knuth2014art}), Algorithm \ref{algApprox} has total complexity $O(n \log n + b_n)$.

\subsection{Additional experiments} \label{appOtherBudgets}

Figure \ref{fig_timeVsBelMore} shows an analogue of Figure \ref{fig_timeVsBel} with budget $b_n = \lceil \tilde{b} |E_n| \rceil$ for each $\tilde{b} \in \{ \frac{1}{1600} ,\frac{1}{800} , \frac{1}{200} , \frac{1}{100}  \}$. The results are qualitatively similar to Figure \ref{fig_timeVsBel} (Algorithm \ref{algExact} outperforms Algorithm \ref{algApprox}, which itself outperforms the heuristics). We also observe the gap between between the heuristics and our algorithms generally increases as the budget decreases for a fixed social network. Put differently, if an adversary with a limited budget spends this budget intelligently (i.e.\ using our proposed solutions), they can still disrupt learning; in contrast, an adversary with a large budget need not be as careful.

\begin{figure}
\centering
\begin{subfigure}{\textwidth}
\centering \includegraphics[height=2in]{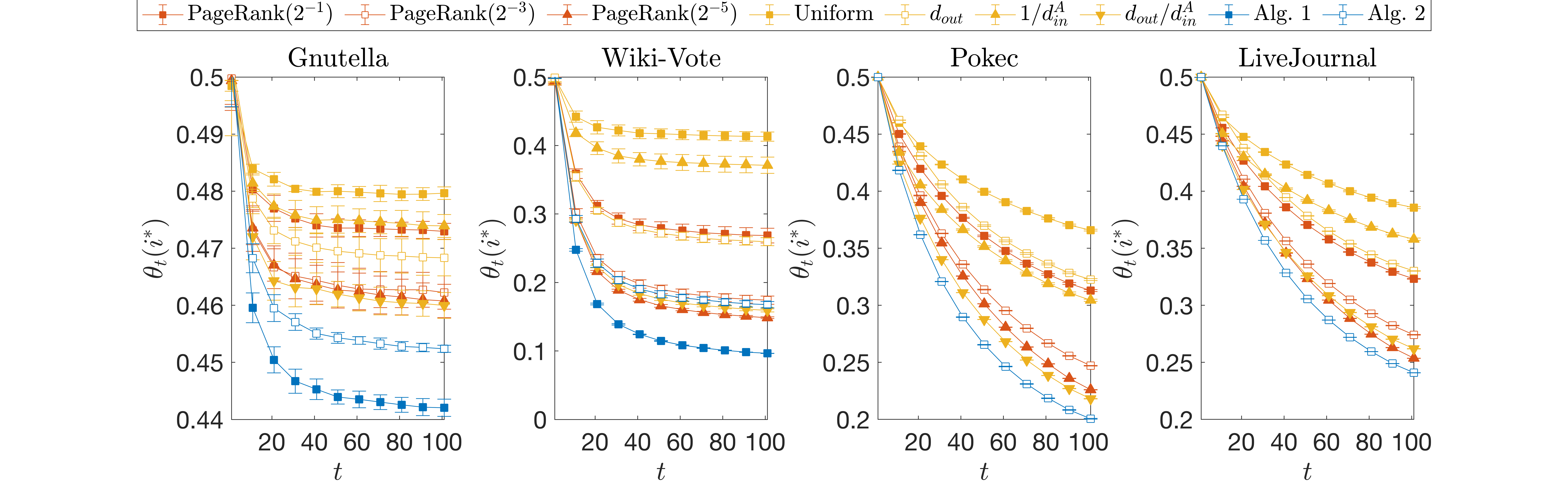} \caption{$\tilde{b} = 1/100$}
\end{subfigure}
\begin{subfigure}{\textwidth}
\centering \includegraphics[height=2in]{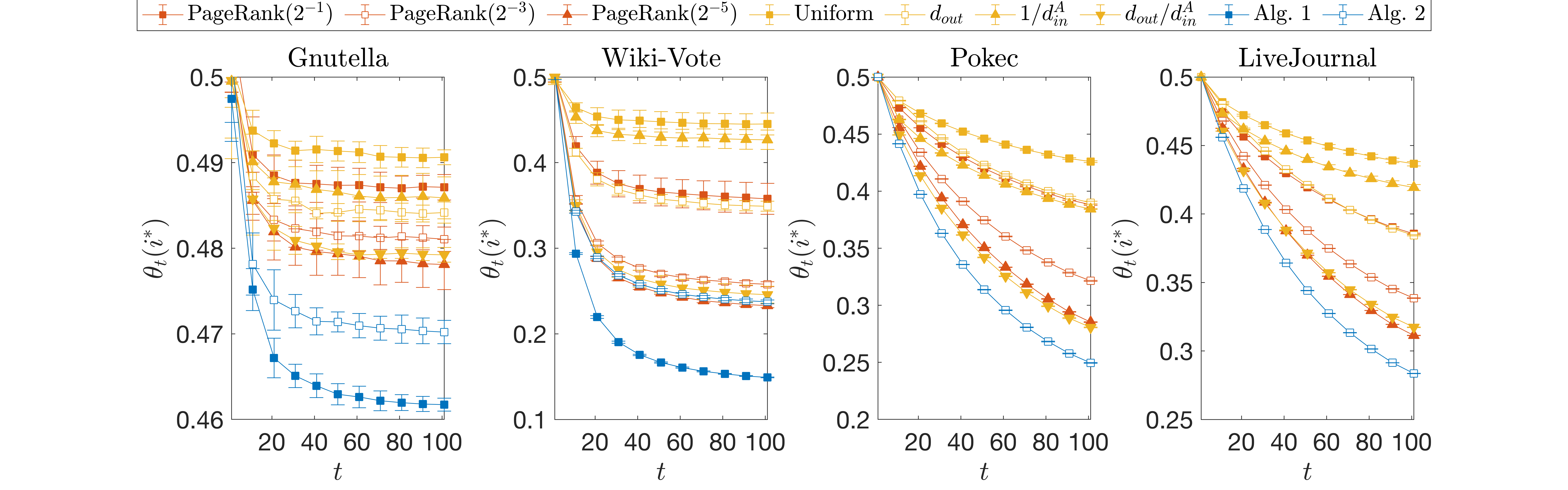}  \caption{$\tilde{b} = 1/200$}
\end{subfigure}
\begin{subfigure}{\textwidth}
\centering \includegraphics[height=2in]{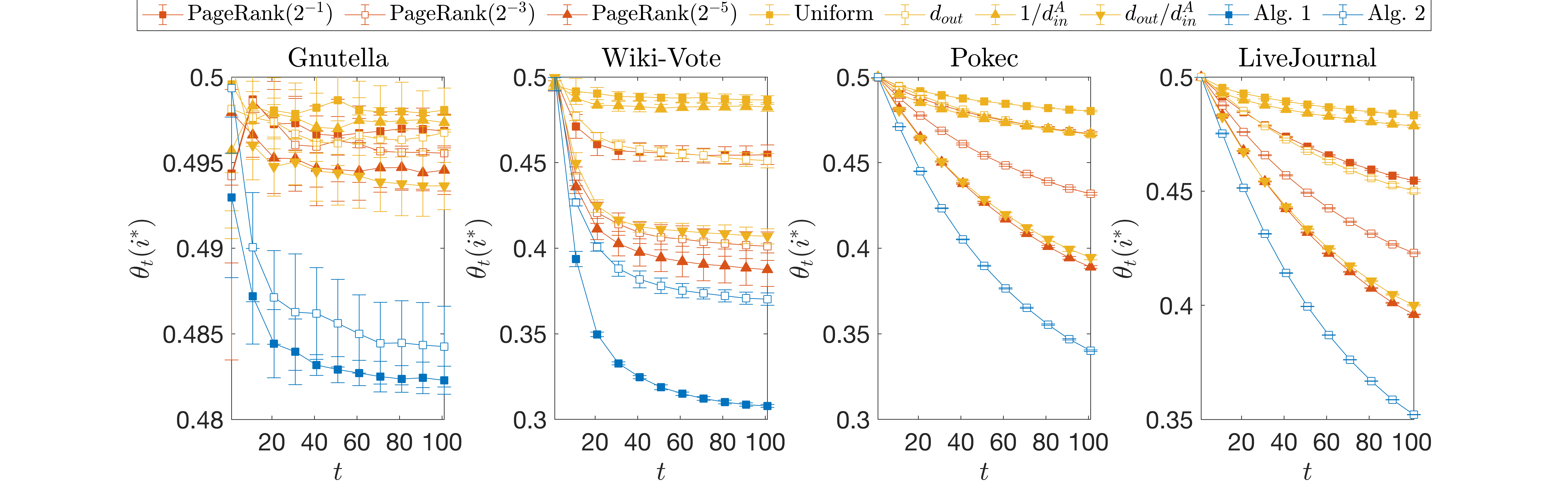}  \caption{$\tilde{b} = 1/800$}
\end{subfigure}
\begin{subfigure}{\textwidth}
\centering \includegraphics[height=2in]{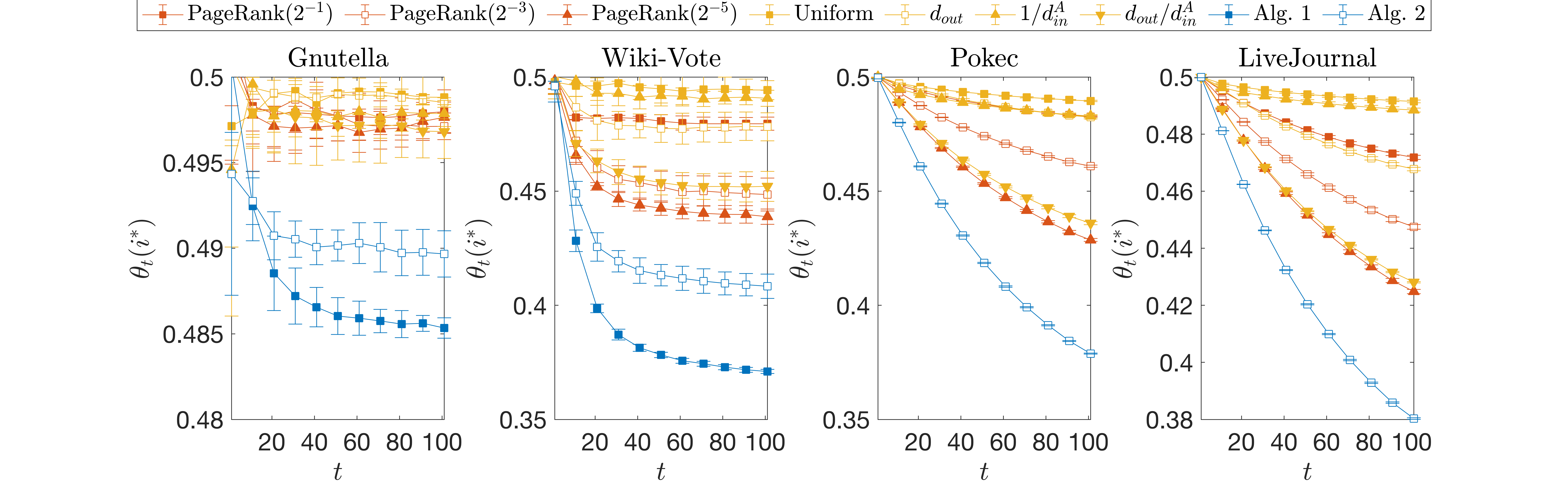}  \caption{$\tilde{b} = 1/1600$}
\end{subfigure}
\caption{Varying budget for Figure \ref{fig_timeVsBel} experiment.} \label{fig_timeVsBelMore}
\end{figure}

%% file: mainProofDetails.tex
\section{Proof of Theorems \ref{thmMain} and \ref{thmSec} (details)} \label{appMainProofDetails} 

\subsection{Branching process approximation (Step 1 for proofs of Theorems \ref{thmMain} and \ref{thmSec})} \label{appBranchApproxProofDetails}

\subsubsection{Proof of Lemma \ref{lemGraphBeliefApprox}} \label{appProofGraphBeliefApprox}

For $t \in \N_0$, let $\alpha_t, \beta_t$ denote the parameters $\{ \alpha_t(i) \}_{i \in A \cup B}, \{ \beta_t(i) \}_{i \in A \cup B}$ in vector form, and let $\mathbf{1}$ denote the all ones vector. We claim
\begin{gather} \label{eqGraphBeliefMatrixForm}
\alpha_t = (1-\eta) \sum_{\tau=1}^{t} s_{\tau} Q^{t-\tau}  + \alpha_0 Q^t  , \quad \beta_t = (1-\eta) \sum_{\tau=1}^{t} (\mathbf{1}-s_{\tau}) Q^{t-\tau}  + \beta_0 Q^t \quad \forall\ t \in \N .
\end{gather}
We prove \eqref{eqGraphBeliefMatrixForm} for $\alpha_t$; the proof for $\beta_t$ follows the same approach. First, we use the parameter update equations \eqref{eqParamUpdateInitial}, and the definitions of $P$ and $Q$ from Appendix \ref{appBranchApproxProofOutline} ($P$ being the column-normalized adjacency matrix and $Q = (1-\eta)I+\eta P$) to write the parameter update equation in vector form as
\begin{equation} \label{eqParamUpdateVector}
\alpha_t = (1-\eta) ( \alpha_t  + s_t ) +\eta \alpha_{t-1} P = (1-\eta) s_t + \alpha_{t-1} Q .
\end{equation}
We next use induction. For $t = 1$, \eqref{eqGraphBeliefMatrixForm} is equivalent to \eqref{eqParamUpdateVector}. Assuming \eqref{eqGraphBeliefMatrixForm} holds for $t-1$, we have
\begin{align}
\alpha_t & = (1-\eta) s_t + \alpha_{t-1} Q = (1-\eta) s_t + \left( (1-\eta) \sum_{\tau=1}^{t-1} s_{\tau} Q^{(t-1)-\tau} + \alpha_0 Q^{t-1} \right) Q \\
& = (1-\eta) s_t  + (1-\eta) \sum_{\tau=1}^{t-1} s_{\tau} Q^{t-\tau} + \alpha_0 Q^t = (1-\eta) \sum_{\tau=1}^{t} s_{\tau} Q^{t-\tau}  + \alpha_0 Q^t ,
\end{align}
which completes the proof. Next, recalling $e_i$ is the vector with 1 in the $i$-th position and 0 elsewhere,
\begin{align}
\theta_{T_n}(i) & = \frac{\alpha_{T_n}(i)}{\alpha_{T_n}(i)+\beta_{T_n}(i)} = \frac{ (1-\eta) \sum_{\tau=1}^{T_n} s_{\tau} Q^{T_n-\tau} e_i^{\trans}  + \alpha_0 Q^{T_n} e_i^{\trans} }{ (1-\eta) \sum_{\tau=1}^{T_n} \mathbf{1} Q^{T_n-\tau} e_i^{\trans} + ( \alpha_0 + \beta_0 ) Q^{T_n} e_i^{\trans} } \\
& = \frac{ (1-\eta) \sum_{\tau=1}^{T_n} s_{\tau} Q^{T_n-\tau} e_i^{\trans}  + \alpha_0 Q^{T_n} e_i^{\trans} }{ (1-\eta) T_n + ( \alpha_0 + \beta_0 ) Q^{T_n} e_i^{\trans} } = \frac{ \frac{1}{T_n} \sum_{\tau=1}^{T_n} s_{\tau} Q^{T_n-\tau} e_i^{\trans}  + \frac{1}{(1-\eta)T_n} \alpha_0 Q^{T_n} e_i^{\trans} }{ 1 + \frac{1}{(1-\eta)T_n} ( \alpha_0 + \beta_0 ) Q^{T_n} e_i^{\trans} } ,
\end{align}
where the equalities hold by definition, by \eqref{eqGraphBeliefMatrixForm}, since the columns of $Q$ sum to 1 by definition, and by multiplying numerator and denominator by $\frac{1}{(1-\eta)T_n}$, respectively. Next, recall from Section \ref{secLearnModel} that $\alpha_0(j) \in [0,\bar{\alpha}]\ \forall\ j \in A \cup B$ for some $\bar{\alpha} > 0$. Hence, $\alpha_0$ is element-wise upper bounded by $\bar{\alpha} \mathbf{1}$, so $\alpha_0 Q^{T_n} e_i^{\trans} \leq \bar{\alpha} \mathbf{1} Q^{T_n} e_i^{\trans} = \bar{\alpha}$, where we have used column stochasticity of $Q$. Additionally, $\alpha_0 Q^{T_n} e_i^{\trans} \geq 0$ (since the three terms in the product are elementwise nonnegative). By a similar argument, $0 \leq \beta_0 Q^{T_n} e_i^{\trans} \leq \bar{\beta}$. Taken together, we can use the previous equation to obtain
\begin{equation} \label{eqThetaVarThetaUpperAndLower}
 \frac{ \frac{1}{T_n} \sum_{\tau=1}^{T_n} s_{\tau} Q^{T_n-\tau} e_i^{\trans}   }{ 1 + \frac{\bar{\alpha} + \bar{\beta} }{(1-\eta)T_n}  }  \leq \theta_{T_n}(i) \leq  \frac{1}{T_n} \sum_{\tau=1}^{T_n} s_{\tau} Q^{T_n-\tau} e_i^{\trans}  + \frac{\bar{\alpha}}{(1-\eta)T_n} .
\end{equation}
Finally, recall from Section \ref{secLearnModel} that $\bar{\alpha}$ and $\bar{\beta}$ are independent of $n$. Hence, because $T_n \rightarrow \infty$ as $n \rightarrow \infty$ (by \ref{assBranchHorizon} in the statement of the lemma), $\bar{\alpha} / T_n, \bar{\beta} / T_n \rightarrow 0$ as $n \rightarrow \infty$. It follows that, for given $\epsilon > 0$ and $n$ sufficiently large, $| \theta_{T_n}(i) -  \frac{1}{T_n} \sum_{\tau=1}^{T_n} s_{\tau} Q^{T_n-\tau} e_i^{\trans} | < \epsilon$. Finally, by changing the index of summation, it is clear that $\frac{1}{T_n} \sum_{\tau=1}^{T_n} s_{\tau} Q^{T_n-\tau} e_i^{\trans} = \vartheta_{T_n}(i)$, completing the proof.

\subsubsection{Proof of Lemma \ref{lemGraphTreeRelation}} \label{appProofGraphTreeRelation}

We begin by arguing $\vartheta_{T_n}(i^*) | \{ \tau_n > T_n \} \stackrel{\mathscr{D}}{=} \hat{\vartheta}_{T_n}(\phi)$. For this, first consider the sub-graph containing only edges between two agents formed during the first $T_n$ iterations of Algorithm \ref{algGraph}. Conditioned on $\tau_n > T_n$, this sub-graph is constructed as follows:
\begin{itemize}
\item The initial agent $i^*$ is sampled uniformly from $A$ (Line \ref{algGraphFirstNode}), which implies its degrees $( d_{out}(i^*),$ $d_{in}^A(i^*), d_{in}^B(i^*) )$ are distributed as $f_n^*$. %Additionally, every agent $i$ in the sub-graph has $d_{in}^B(i)$ bots in its incoming neighbor set, each with a single self-loop and no other edges. 
(In fact, this holds even if $\tau_n \leq T_n$.)
\item Each time an edge is added to the sub-graph (Line \ref{algGraphAgentPair}), the paired outstub $(i',j')$ is sampled uniformly from $O_A$ (else, $\tau_n > T_n$ is contradicted by Line \ref{algGraphTau}-\ref{algGraphResample}), so the degrees $( d_{out}(i'), d_{in}^A(i'),$ $d_{in}^B(i') )$ of the corresponding agent $i'$ are distributed as $f_n$.
\item The initial agent $i^*$ has no paired outstubs, while all other agents in the sub-graph have one paired outstub (otherwise, an outstub with label 2 was paired within the first $T_n$ iterations, contradicting $\tau_n > T_n$ by Line \ref{algGraphTau}); in particular, the sub-graph has $| \cup_{l=0}^{T_n} A_l |$ nodes and $| \cup_{l=0}^{T_n} A_l | - 1$ edges. Also, every agent in the sub-graph has a path to $i^*$ by the breadth-first-search nature of the construction, so, neglecting edge polarities, we obtain a connected graph with $| \cup_{l=0}^{T_n} A_l |$ nodes and $| \cup_{l=0}^{T_n} A_l | - 1$ edges, i.e.\ a tree. Finally, since all edges point towards $i^*$ (see Line \ref{algGraphAgentPair}), the sub-graph is a directed tree pointed towards $i^*$.
\end{itemize}
In summary, the sub-graph is a directed tree pointing towards an agent with degrees distributed as $f_n^*$, in which all other nodes have degrees distributed as $f_n$. This is precisely the procedure used to construct the sub-graph of agents during the first $T_n$ iterations of Algorithm \ref{algTree}. Additionally, Algorithms \ref{algGraph} and \ref{algTree} add bots in the same manner (Lines \ref{algGraphAddBot}-\ref{algGraphBotPair} in Algorithm \ref{algGraph}, Lines \ref{algTreeAddBot}-\ref{algTreeBotPair} in Algorithm \ref{algTree}). Taken together, we conclude that, conditioned on $\tau_n > T_n$, the $T_n$-step neighborhood into $i^*$ is constructed in the same manner in Algorithm \ref{algGraph} as the $T_n$-step neighborhood into $\phi$ is constructed in Algorithm \ref{algTree}. Furthermore, by \eqref{eqVarThetaGraphExpand} and \eqref{eqVarThetaTree}, it is clear that $\vartheta_{T_n}(i)$ and $\hat{\vartheta}_{T_n}(\phi)$, respectively, depend only on these respective neighborhoods, and on the signals $s_{T_n-t}(i)$ and $\hat{s}_{T_n-t}(\i)$, respectively. Since the signals $s_{T_n-t}(i)$ and $\hat{s}_{T_n-t}(\i)$ are also defined in the same manner ($s_{T_n-t}(i), \hat{s}_{T_n-t}(\i) \sim\textrm{Bernoulli}(\theta)$ for $i \in A, \i \in \hat{A}$; $s_{T_n-t}(i) = \hat{s}_{T_n-t}(\i) = 0$ for $i \in B, \i \in \hat{B}$), we ultimately conclude that $\vartheta_{T_n}(i^*)$ and $\hat{\vartheta}_{T_n}(\phi)$ have the same distribution when $\tau_n > T_n$ holds.

We next argue $\{ \tau_n > T_n\}$ occurs with high probability when $\Omega_{n,1}$ holds. For this, we note that Algorithm \ref{algGraph} is nearly identical to the graph construction described in \cite[Section 5.2]{chen2017generalized}. More specifically, the only difference is that the construction in \cite{chen2017generalized} does not include the pairing of agent instubs with bots in Lines \ref{algGraphAddBot}-\ref{algGraphBotPair} of Algorithm \ref{algGraph}. However, these lines do not affect $\tau_n$. Moreover, when \ref{assGraphDegSeq} holds, the assumptions of \cite[Lemma 5.4]{chen2017generalized} are satisfied. This lemma states that, if $t_n < ( \log n ) / ( 2 \log (\nu_3 / \nu_1) )$ and $\nu_3 > \nu_1$ (with $\nu_1, \nu_3$ defined as in \ref{assGraphDegSeq}), then $P(\tau_n \leq t_n | \Omega_{n,1} ) = O ( (\nu_3/\nu_1)^{t_n} / \sqrt{n} )$. In particular, by \ref{assGraphHorizon} we have $T_n \leq  \zeta \log ( n ) /  \log ( \nu_3 / \nu_1  )$ for $n$ sufficiently large, with $\zeta \in (0,1/2)$ independent of $n$; substituting gives
\begin{equation}
\P(\tau_n \leq T_n | \Omega_{n,1} ) = O \left( \frac{(\nu_3/\nu_1)^{ \zeta \log ( n ) /  \log ( \nu_3 / \nu_1  )  }}{ \sqrt{n} } \right) = O \left( n^{ \zeta - 1/2 } \right) .
\end{equation}

\subsubsection{Proof of Lemma \ref{lemTreeBeliefExpand}} \label{appProofTreeBeliefExpand}

We first claim that for $l \in \N_0$ and $\i \in \hat{A}_l$,
\begin{equation} \label{eqTreeBeliefExpandId2}
e_{\i} \hat{P}^{l'} e_{\phi} = \begin{cases} \prod_{j=0}^{l-1} d_{in}(\i|j)^{-1} , & l' = l \\ 0 , & l' \in \N_0 \setminus \{l\} \end{cases} .
\end{equation}
(Recall $\hat{P}$ is the column-normalized adjacency matrix.) We prove \eqref{eqTreeBeliefExpandId2} separately for $l = 0$ and $l \in \N$. When $l = 0$, the only case is $\i = \phi$ (since $\hat{A}_0 = \{ \phi \}$); if $l' = 0$, the left side is clearly 1 and the right side is 1 by convention; if $l' \in \N$, the left side is 0 since $e_{\phi} \hat{P}^{l'} = 0$ ($\phi$ has no outgoing neighbors in the tree). Next, we aim to prove \eqref{eqTreeBeliefExpandId2} for $\i \in \hat{A}_l$ and $l \in \N$. For such $\i$, there is a unique path from $\i$ to $\phi$ with length $l$ that visits the nodes $\i|l = \i , \i | l-1 , \ldots , \i | 0 = \phi$. By definition of $\hat{P}$, it follows that
\begin{equation}
e_{\i} \hat{P}^{l} e_{\phi} = \hat{P}( \i | l , \i | l-1 ) \hat{P} ( \i | l-1 , \i | l-2 ) \cdots \hat{P} ( \i | 1 , \i | 0 ) = \frac{1}{d_{in}(\i|l-1)} \frac{1}{d_{in}(\i|l-2)} \cdots \frac{1}{d_{in}(\phi)} .
\end{equation}
On the other hand, if $l' \neq l$, no path of length $l'$ from $\i$ to $\phi$ exists, so $e_{\i} \hat{P}^{l'} e_{\phi} = 0$. This proves \eqref{eqTreeBeliefExpandId2}.

Recalling that $\hat{Q} = (1-\eta)I+\eta \hat{P}$, we next claim that $\forall\ t \in \N_0$,
\begin{equation} \label{eqTreeBeliefExpandId1}
\hat{Q}^t = \sum_{l=0}^t { t \choose l } \eta^l (1-\eta)^{t-l} \hat{P}^l .
\end{equation}
We prove \eqref{eqTreeBeliefExpandId1} inductively: both sides equal $I$ when $t = 0$; assuming \eqref{eqTreeBeliefExpandId1} is true for $t$, we have
\begin{align}
\hat{Q}^{t+1} & = ( (1-\eta) I + \eta \hat{P} ) \sum_{l=0}^t { t \choose l } \eta^l (1-\eta)^{t-l} \hat{P}^l \\ 
%& = \sum_{l=0}^t {t \choose l} \eta^l (1-\eta)^{t+1-l} \hat{P}^l + \sum_{l=0}^t {t \choose l} \eta^{l+1} (1-\eta)^{t-l} \hat{P}^{l+1} \\
& = \sum_{l=0}^t {t \choose l} \eta^l (1-\eta)^{t+1-l} \hat{P}^l + \sum_{l=1}^{t+1} {t \choose l-1} \eta^{l} (1-\eta)^{t+1-l} \hat{P}^{l} \\ 
& = (1-\eta)^{t+1} I + \sum_{l=1}^{t} \left( {t \choose l} + {t \choose l-1}  \right) \eta^{l} (1-\eta)^{t+1-l} \hat{P}^{l} + \eta^{t+1} \hat{P}^{t+1} \\
& = (1-\eta)^{t+1} I + \sum_{l=1}^{t} {t+1 \choose l} \eta^{l} (1-\eta)^{t+1-l} \hat{P}^{l} + \eta^{t+1} \hat{P}^{t+1} ,
\end{align}
where in the first line we have used the definition of $\hat{Q}$ and the inductive hypothesis, the second line simply uses the distributive property, the third rearranges summations, and the fourth uses Pascal's rule ($[t+1]$ has ${t+1 \choose l}$ subsets of cardinality $l$; ${t \choose l-1}$ that contain 1 and ${t \choose l}$ that do not contain 1). This completes the proof of \eqref{eqTreeBeliefExpandId1}. 

Having established \eqref{eqTreeBeliefExpandId1} and \eqref{eqTreeBeliefExpandId2}, we can combine them to obtain $\forall\ t \in \N_0, \i \in \hat{A}_l$,
\begin{equation}
e_{\i} \hat{Q}^{t} e_{\phi} =  \begin{cases} { t \choose l } \eta^{l} (1-\eta)^{t-l} \prod_{j=0}^{l-1} d_{in}(\i|j)^{-1} , &  l \leq t \\ 0 , &  l > t \end{cases}.
\end{equation}
Finally, substituting the previous equation into \eqref{eqVarThetaTree}, and recalling $\hat{A} = \cup_{l=0}^{\infty} \hat{A}_l$, we obtain
\begin{equation}
\hat{\vartheta}_{T_n}(\phi) =  \frac{1}{T_n} \sum_{t =0}^{T_n-1} \sum_{l=0}^t \sum_{\i \in \hat{A}_l} \hat{s}_{T_n-t}(\i) { t \choose l } \eta^{l} (1-\eta)^{t-l} \prod_{j=0}^{l-1} d_{in}(\i|j)^{-1} ,
\end{equation}
which completes the proof.

\subsection{Step 2 for proof of Theorem \ref{thmMain}} \label{appRootBeliefProofDetails}

\subsubsection{Proof of Lemma \ref{lemFirstMomBelief}} \label{appProofFirstMomBelief}

First, letting $\D$ denote the degree sequence and $\T$ denote the set of random variables defining the tree structure, we can use Lemma \ref{lemTreeBeliefExpand} to write
\begin{align}
& \E_n [ \hat{\vartheta}_{T_n}(\phi) ] = \frac{1}{T_n} \sum_{t=0}^{T_n-1} \sum_{l=0}^t { t \choose l } \eta^l (1-\eta)^{t-l} \E_n \left[ \sum_{\i \in \hat{A}_l} \E [ \hat{s}_{T_n-t}(\i) | \D,\mathcal{T} ]  \prod_{j=0}^{l-1} d_{in}(\i|j)^{-1} \right] \\
& \quad  = \frac{\theta}{T_n} \sum_{t=0}^{T_n-1} \sum_{l=0}^t { t \choose l } \eta^l (1-\eta)^{t-l} \E_n \left[ \sum_{\i \in \hat{A}_l} \prod_{j=0}^{l-1} d_{in}(\i|j)^{-1}  \right] = \frac{\theta}{T_n} \sum_{t=0}^{T_n-1} \sum_{l=0}^t { t \choose l } \eta^l (1-\eta)^{t-l} \P_n ( X_l^1 \in \hat{A} ) , \label{eqFirstMomPreWalks}
\end{align}
where the first equality uses the tower property of conditional expectation and the fact that $\hat{A}_l$ and $d(\i | j)^{-1}$ are fixed given the tree structure, the second uses the fact that $\hat{s}_{T_n-t}(\i) \sim \textrm{Bernoulli}(\theta)$, and the third holds by the tower property and the definition of $X_l^1$, i.e.\
\begin{equation} \label{eqWalkProbAppears}
\P_n ( X_l^1 \in \hat{A} ) = \E_n [ \P(  X_l^1 \in \hat{A} | \D, \T ) ] = \E_n \left[ \sum_{\i \in \hat{A}_l} \prod_{j=0}^{l-1} d_{in}(\i|j)^{-1}  \right] .
\end{equation}
Here we have also used the fact that $\{ X_l^1 \}_{l \in \N}$ is a random walk starting at the root of a directed tree; hence, for $\i \in \hat{A}_l$, $\P ( X_l^1 = \i | \D, \T )$ is the probability of the lone path from $\phi$ to $\i$, which is $\prod_{j=0}^{l-1} d_{in}(\i|j)^{-1} $, and $X_l^1 \in \hat{A} \Leftrightarrow X_l^1 = \i$ for some $\i \in \hat{A}_l$. Next, using \eqref{eqFirstMomPreWalks} and Lemma \ref{lemHitProbOneWalk}, we obtain
\begin{equation} \label{eqTreeBeliefMeanWalkSub}
\E_n [ \hat{\vartheta}_{T_n}(\phi) ] =  \frac{\theta}{T_n} \sum_{t=0}^{T_n-1} \left( \sum_{l=1}^t { t \choose l } \eta^l (1-\eta)^{t-l} \tilde{p}_n^* \tilde{p}_n^{l-1} + (1-\eta)^t \right) ,
\end{equation}
where by convention the summation over $l$ is zero when $t = 0$. Adding and subtracting $(1-\eta)^t \tilde{p}_n^* / \tilde{p}_n$, the previous equation can be rewritten as
\begin{align} \label{eqTreeBeliefMeanBinomGeo}
\E_n [ \hat{\vartheta}_{T_n}(\phi) ] & = \frac{\theta}{T_n}  \frac{\tilde{p}_n^*}{\tilde{p}_n} \sum_{t=0}^{T_n-1} \sum_{l=0}^t { t \choose l } (\eta \tilde{p}_n)^l (1-\eta)^{t-l} + \frac{\theta}{T_n} \left( 1 - \frac{\tilde{p}_n^*}{\tilde{p}_n} \right) \sum_{t=0}^{T_n-1} (1-\eta)^t  \\
& = \frac{\theta}{T_n}  \frac{\tilde{p}_n^*}{\tilde{p}_n} \sum_{t=0}^{T_n-1} ( 1 - \eta ( 1 - \tilde{p}_n ) )^t + \frac{\theta}{T_n} \left( 1 - \frac{\tilde{p}_n^*}{\tilde{p}_n} \right) \frac{1-(1-\eta)^{T_n}}{\eta}  \\
& = \frac{\theta}{T_n} \frac{\tilde{p}_n^*}{\tilde{p}_n} \frac{ 1 - (1- \eta ( 1 - \tilde{p}_n ) )^{T_n} }{ \eta (1-\tilde{p}_n) } + \frac{\theta}{T_n} \left( 1 - \frac{\tilde{p}_n^*}{\tilde{p}_n} \right) \frac{1-(1-\eta)^{T_n}}{\eta} ,
\end{align}
where we have simply used the binomial theorem and computed two geometric series. 

Next, we assume temporarily that $p_n \rightarrow 1$ as $n \rightarrow \infty$. By \ref{assBranchDegSeq}, we have for $\omega \in \Omega_{n,2}$
\begin{equation}
\tilde{p}_n(\omega) \in (p_n-\delta_n,p_n+\delta_n) .
\end{equation}
Hence, by $p_n \rightarrow 1$, and since $\delta_n \rightarrow 0$ by \ref{assBranchDegSeq}, we have for $\gamma_1 > 0$, $n$ sufficiently large, and such $\omega$
\begin{equation}
1 - \gamma_1 < \frac{\tilde{p}_n^*(\omega)}{\tilde{p}_n(\omega)} < 1 + \gamma_1 ,
\end{equation}
where we have also used the fact that $1 \geq \tilde{p}_n^* \geq \tilde{p}_n$ on $\Omega_{n,2}$ by \ref{assBranchDegSeq}. Also, by \ref{assBranchHorizon}, it is clear that $( 1 - (1-\eta) )^{T_n} / T_n  \rightarrow 0$, so for given $\gamma_2 > 0$ and $n$ sufficiently large,
\begin{equation}
0 < \frac{\theta}{T_n} \frac{1-(1-\eta)^{T_n}}{\eta} < \gamma_2 .
\end{equation}
Combining the previous four equations implies that for $n$ sufficiently large and $\omega \in \Omega_{n,2}$,
\begin{align} \label{eqMeanBeforeSpeciazling}
\E_n [ \hat{\vartheta}_{T_n}(\phi) ] (\omega) & < (1+\gamma_1) \frac{\theta}{T_n} \frac{ 1 - (1 - \eta ( 1 - p_n - \delta_n ) )^{T_n} }{ \eta (1-p_n - \delta_n) } + \gamma_1 \gamma_2 ,\\
\E_n [ \hat{\vartheta}_{T_n}(\phi) ] (\omega) & > (1-\gamma_1) \frac{\theta}{T_n} \frac{ 1 - (1- \eta ( 1 - p_n + \delta_n ))^{T_n} }{ \eta (1-p_n + \delta_n) } - \gamma_1 \gamma_2  .
\end{align}
We complete the proof for the case $T_n(1-p_n) \rightarrow 0$; the proof for the other two cases is similar. In this case, we can use Lemma \ref{lemTnPnDnasymp} from Appendix \ref{appAuxiliary} to obtain for any $\gamma_3> 0$ and for $n$ large enough
\begin{gather}
1 - \gamma_3 < \frac{ 1 - (1-\eta ( 1 - p_n - \delta_n ) )^{T_n} }{ T_n \eta (1-p_n - \delta_n) } < 1 + \gamma_3 , \\
1 - \gamma_3 < \frac{ 1 - (1-\eta ( 1 - p_n + \delta_n ))^{T_n} }{ T_n \eta (1-p_n + \delta_n) } < 1 + \gamma_3 .
\end{gather}
Combining the previous two equations gives for $n$ large and $\omega \in \Omega_{n,2}$
\begin{align}
\E_n [ \hat{\vartheta}_{T_n}(\phi) ] (\omega) & < \theta (1+\gamma_1)  (1+\gamma_3) + \gamma_1 \gamma_2 = \theta + \theta( \gamma_1  + \gamma_3 + \gamma_1 \gamma_3 ) + \gamma_1 \gamma_2 ,\\
\E_n [ \hat{\vartheta}_{T_n}(\phi) ] (\omega) & > \theta (1-\gamma_1) (1-\gamma_3) - \gamma_1 \gamma_2  = \theta - \theta ( \gamma_1 + \gamma_3 - \gamma_1 \gamma_3 ) - \gamma_1 \gamma_2 .
\end{align}
Hence, for given $\gamma > 0$, we can find $\gamma_1,\gamma_2,\gamma_3$ sufficiently small and $n$ sufficiently large such that, for $\omega \in \Omega_{n,2}$, $| \E_n [ \hat{\vartheta}_{T_n}(\phi) ] (\omega) - \theta | < \gamma$. This clearly also implies $| \E_n [ \hat{\vartheta}_{T_n}(\phi) ] (\omega) - \theta | 1(\Omega_{n,2})(\omega) < \gamma$ for such $\omega$. On the other hand, for $\omega \notin \Omega_{n,2}$, it is trivial that $| \E_n [ \hat{\vartheta}_{T_n}(\phi) ] (\omega) - \theta | 1 ( \Omega_{n,2} )(\omega) = 0 < \gamma$. This completes the proof for the case $T_n(1-p_n) \rightarrow 0$.

We now return to the case $p_n \rightarrow p \in [0,1)$. In this case, it follows from \ref{assBranchHorizon} that $T_n(1-p_n) \rightarrow [0,\infty)$ cannot occur, i.e.\ we need only consider the case $T_n(1-p_n) \rightarrow \infty$. First, note that since $p_n \rightarrow p < 1$ and $\delta_n \rightarrow 0$, we have $p_n + \delta_n < 1 - \gamma_1$ for some $\gamma_1 > 0$ and $n$ sufficiently large. For such $n$, and for $\omega \in \Omega_{n,2}$, we then obtain $\tilde{p}_n(\omega) < 1 - \gamma_1$; substituting into \eqref{eqTreeBeliefMeanWalkSub} (evaluated at $\omega$) gives
\begin{align} \label{eqMeanConvCase4}
& \E_n [ \hat{\vartheta}_{T_n}(\phi) ](\omega) <  \frac{\theta}{T_n} \sum_{t=0}^{T_n-1} \left( \sum_{l=1}^t { t \choose l } \eta^l (1-\eta)^{t-l} (1-\gamma_1)^{l-1} + (1-\eta)^t \right) \\
& \quad < \frac{\theta}{T_n} \frac{1}{1-\gamma_1} \sum_{t=0}^{T_n-1} \sum_{l=0}^t { t \choose l } \eta^l (1-\eta)^{t-l} (1-\gamma_1)^{l} = \frac{\theta}{T_n} \frac{1}{1-\gamma_1} \frac{1 - (1- \eta \gamma_1)^{T_n} }{ \eta \gamma_1 } < \frac{\theta}{T_n} \frac{1}{1-\gamma_1} \frac{1}{\eta \gamma_1} 
\end{align}
where in the first inequality we used $\tilde{p}_n(\omega) < 1 - \gamma_1$ and $\tilde{p}_n^*(\omega) \leq 1$, in the second we used $1-\gamma_1 \in (0,1)$ (so that $(1-\eta)^t < (1-\eta)^t / (1-\gamma_1)$), for the equality we used the binomial theorem and computed a geometric series, and the final inequality is immediate. Since $\theta, \eta, \gamma_1$ are independent of $n$, while $T_n \rightarrow \infty$ as $n \rightarrow \infty$ by \ref{assBranchHorizon}, it is clear from this final expression that, for given $\gamma > 0$, $n$ sufficiently large, and $\omega \in \Omega_{n,2}$, $0 \leq \E_n [ \hat{\vartheta}_{T_n}(\phi) ](\omega)  < \gamma$. It follows that $| \E_n [ \hat{\vartheta}_{T_n}(\phi) ] | 1(\Omega_{n,2}) \rightarrow 0\ a.s.$, completing the proof.

\subsubsection{Proof of Lemma \ref{lemSecondMomBelief}} \label{appProofSecondMomBelief} 

First, suppose $p_n \rightarrow p \in [0,1)$. Then, since $\hat{\vartheta}_{T_n}(\phi) \leq 1\ a.s.$ (see \eqref{eqBeliefIn01} and the following argument), $\Var_n ( \hat{\vartheta}_{T_n}(\phi) ) \leq \E_n \hat{\vartheta}_{T_n}(\phi)^2 \leq \E_n \hat{\vartheta}_{T_n}(\phi)$. Furthermore, since $T_n \rightarrow \infty$ by \ref{assBranchHorizon}, the fact that $p_n \rightarrow p \in [0,1)$ means only the case $T_n(1-p_n) \rightarrow \infty$ can occur. In this case, since $\E_n [ \hat{\vartheta}_{T_n}(\phi) ] 1 ( \Omega_{n,2} ) \rightarrow 0\ a.s.$ by Lemma \ref{lemFirstMomBelief}, we immediately obtain from $\Var_n ( \hat{\vartheta}_{T_n}(\phi) ) \leq \E_n [ \hat{\vartheta}_{T_n}(\phi) ]$ that $\Var_n ( \hat{\vartheta}_{T_n}(\phi) ) 1 ( \Omega_{n,2} ) \rightarrow 0\ a.s.$ as well. Hence, it only remains to prove the lemma in the case $p_n \rightarrow 1$, which we assume to hold for the remainder of the proof.

Towards this end, letting $\D$ denote the degree sequence and $\T$ denote the set of random variables defining the tree structure (as in Appendix \ref{appProofFirstMomBelief}), we have
\begin{align} \label{eqTotalVar}
\Var_n ( \hat{\vartheta}_{T_n}(\phi) ) %& \triangleq \E \left[ \left( \frac{1}{T} \sum_{t=0}^{T-1} \hat{s}_{T-t} \hat{Q}^t e_{\phi}^{\trans}  - \E \left[  \frac{1}{T} \sum_{t=0}^{T-1} \hat{s}_{T-t} \hat{Q}^t e_{\phi}^{\trans} \middle| \D  \right] \right)^2  \middle| \D \right] \\
%& = \E \left[ \left( \frac{1}{T} \sum_{t=0}^{T-1} \hat{s}_{T-t} \hat{Q}^t e_{\phi}^{\trans} \right)^2 \middle| \D \right] - \left( \E \left[ \frac{1}{T} \sum_{t=0}^{T-1} \hat{s}_{T-t} \hat{Q}^t e_{\phi}^{\trans} \middle| \D \right] \right)^2 \\
%& = \E \left[ \E \left[ \left( \frac{1}{T} \sum_{t=0}^{T-1} \hat{s}_{T-t} \hat{Q}^t e_{\phi}^{\trans} \right)^2 \middle| \D , \mathcal{T} \right] \middle| \D \right] \\
%& \quad\quad - \left( \E \left[ \E \left[ \frac{1}{T} \sum_{t=0}^{T-1} \hat{s}_{T-t} \hat{Q}^t e_{\phi}^{\trans} \middle| \D , \mathcal{T} \right] \middle| \D \right] \right)^2 \\
%& = \E \left[ \Var \left( \frac{1}{T} \sum_{t=0}^{T-1} \hat{s}_{T-t} \hat{Q}^t e_{\phi}^{\trans}  \middle| \D , \mathcal{T} \right) + \left( \E \left[ \frac{1}{T} \sum_{t=0}^{T-1} \hat{s}_{T-t} \hat{Q}^t e_{\phi}^{\trans} \middle| \D, \mathcal{T} \right] \right)^2 \middle| \D \right] \\
%& \quad\quad - \left( \E \left[ \E \left[ \frac{1}{T} \sum_{t=0}^{T-1} \hat{s}_{T-t} \hat{Q}^t e_{\phi}^{\trans} \middle| \D , \mathcal{T} \right] \middle| \D \right] \right)^2 \\
& = \E_n [ \Var (\hat{\vartheta}_{T_n}(\phi) | \D , \mathcal{T} ) ] + \Var_n ( \E [ \hat{\vartheta}_{T_n}(\phi) | \D , \mathcal{T} ] ) .
\end{align}
We next consider the two summands in \eqref{eqTotalVar} in turn. In particular, we aim to show that each summand multiplied by $1(\Omega_{n,2})$ tends to zero $a.s.$ as $n$ tends to infinity.

For the first summand in \eqref{eqTotalVar}, we use the fact that the signals are i.i.d.\ Bernoulli($\theta$) given the tree structure, as well as Lemma \ref{lemTreeBeliefExpand}, to write
\begin{align}
\Var ( \hat{\vartheta}_{T_n}(\phi) | \D , \mathcal{T} ) & =  \frac{1}{T_n^2} \sum_{t=0}^{T_n-1} \sum_{l=0}^t \sum_{\i \in \hat{A}_l} \Var( \hat{s}_{T_n-t}(\i) | \D,\mathcal{T} ) \left( { t \choose l } \eta^l (1-\eta)^{t-l} \prod_{j=0}^{l-1} d_{in}(\i|j)^{-1} \right)^{2} \\
& = \frac{1}{T_n^2} \sum_{t=0}^{T_n-1} \sum_{l=0}^t \sum_{\i \in \hat{A}_l} \theta(1-\theta) \left( { t \choose l } \eta^l (1-\eta)^{t-l} \prod_{j=0}^{l-1} d_{in}(\i|j)^{-1} \right)^{2} \\
& \leq \frac{1}{T_n^2} \sum_{t=0}^{T_n-1} \sum_{l=0}^t   { t \choose l } \eta^l (1-\eta)^{t-l} \sum_{\i \in \hat{A}_l} \prod_{j=0}^{l-1} d_{in}(\i|j)^{-1} \leq \frac{1}{T_n} ,
\end{align}
where in the final step we have used $\sum_{\i \in \hat{A}_l} \prod_{j=0}^{l-1} d_{in}(\i|j)^{-1} \leq 1$ and $\sum_{l=0}^t   { t \choose l } \eta^l (1-\eta)^{t-l} = 1$. It immediately follows that $0 \leq \E_n [ \Var (\hat{\vartheta}_{T_n}(\phi) | \D , \mathcal{T} ) ] 1(\Omega_{n,2}) \leq 1 / T_n\ a.s.$ Hence, because $T_n \rightarrow \infty$ as $n \rightarrow \infty$ by \ref{assBranchHorizon}, analysis of the first summand in \eqref{eqTotalVar} is complete.

For the second summand in \eqref{eqTotalVar}, we first use the argument of \eqref{eqFirstMomPreWalks} to write
\begin{align}
\E [ \hat{\vartheta}_{T_n}(\phi) | \D , \mathcal{T} ] % & = \frac{1}{T_n} \sum_{t=0}^{T_n-1} \sum_{l=0}^t \sum_{\i \in \hat{A}_l} \E [ \hat{s}_{T_n-t}(\i) | \D,\mathcal{T} ]  { t \choose l } \eta^l (1-\eta)^{t-l} \prod_{j=0}^{l-1} d_{in}(\i|j)^{-1} \\
& = \frac{\theta}{T_n} \sum_{t=0}^{T_n-1} \sum_{l=0}^t   { t \choose l } \eta^l (1-\eta)^{t-l} \sum_{\i \in \hat{A}_l} \prod_{j=0}^{l-1} d_{in}(\i|j)^{-1} \\
& = \frac{\theta}{T_n} \sum_{l=0}^{T_n-1} \sum_{\i \in \hat{A}_l} \prod_{j=0}^{l-1} d_{in}(\i|j)^{-1} \sum_{t=l}^{T_n-1}   { t \choose l } \eta^l (1-\eta)^{t-l}   \triangleq \frac{\theta}{T_n} \sum_{l=0}^{T_n-1} Y_l u_{T_n,l} ,
\end{align}
where we have defined $Y_l = \sum_{\i \in \hat{A}_l} \prod_{j=0}^{l-1} d_{in}(\i|j)^{-1}$ and $u_{T_n,l} = \sum_{t=l}^{T_n-1}   { t \choose l } \eta^l (1-\eta)^{t-l}$. Therefore,
\begin{align} \label{eqTotalVarTerm2}
\Var_n ( \E [ \hat{\vartheta}_{T_n}(\phi) | \D , \mathcal{T} ]  ) & = \frac{\theta^2}{T_n^2} \left( \sum_{l=0}^{T_n-1} u_{T_n,l}^2 \Var_n ( Y_l  ) + 2 \sum_{l=0}^{T_n-1} u_{T_n,l} \sum_{l' = l+1}^{T_n-1} u_{T_n,l'} \Cov_n (Y_l, Y_{l'} ) \right) .% \\
\end{align}
It remains to compute the variance and covariance terms in \eqref{eqTotalVarTerm2}. First, for any $l,l' \in \N$, we note
\begin{align} \label{eqYlToTwoWalks}
\E_n [ Y_l Y_{l'}  ] & = \E_n \left[ \P ( X_l^1 \in \hat{A}  | \D, \T ) \P ( X_{l'}^2 \in \hat{A}  | \D, \T ) \right] \\
&  = \E_n \left[ \P ( X_l^1 \in \hat{A} , X_{l'}^2 \in \hat{A}  | \D, \T ) \right]  = \P_n ( X_l^1 \in \hat{A} , X_{l'}^2 \in \hat{A}  ) ,
\end{align}
where we have used the argument of \eqref{eqWalkProbAppears} and the fact that $\{ X_i^1 \}_{i=1}^{\infty}$ and $\{ X_i^2 \}_{i=1}^{\infty}$ are independent random walks given the tree structure. By a similar argument, $\E_n [ Y_l ] = \P_n (X_l^1 \in \hat{A})$. Hence, using Lemmas \ref{lemHitProbOneWalk} and \ref{lemHitProbTwoWalks}, and assuming for the moment that $l> 1$, we have
\begin{align}
\Var_n(Y_l) & = \P_n ( X_l^1 \in \hat{A} , X_l^2 \in \hat{A}  ) - ( \P_n ( X_l^1 \in \hat{A} ) )^2 \\
& = \tilde{r}_n^* \tilde{p}_n^{2(l-1)}  +  \sum_{t=2}^{l} \tilde{q}_n^* \tilde{q}_n^{t-2} \tilde{r}_n \tilde{p}_n^{2(l-t)} + \tilde{q}_n^* \tilde{q}_n^{l-1} - ( \tilde{p}_n^* \tilde{p}_n^{l-1} )^2 \\
& = \frac{\tilde{r}_n^*}{\tilde{p}_n^2} \tilde{p}_n^{2l} + \frac{\tilde{q}_n^* \tilde{r}_n}{\tilde{p}_n^4} \tilde{p}_n^{2l} \sum_{l=0}^{l-2} \left( \frac{\tilde{q}_n}{\tilde{p}_n^2} \right)^t + \frac{\tilde{q}_n^*}{\tilde{q}_n} \tilde{q}_n^l - \frac{ (\tilde{p}_n^*)^2 }{\tilde{p}_n^2} \tilde{p}_n^{2l} \\
& = \frac{\tilde{r}_n^*}{\tilde{p}_n^2} \tilde{p}_n^{2l} + \frac{\tilde{q}_n^* \tilde{r}_n}{\tilde{p}_n^4} \tilde{p}_n^{2l} \frac{1-(\tilde{q}_n/\tilde{p}_n^2)^{l-1}}{1- (\tilde{q}_n/\tilde{p}_n^2) }+ \frac{\tilde{q}_n^*}{\tilde{q}_n} \tilde{q}_n^l - \frac{ (\tilde{p}_n^*)^2 }{\tilde{p}_n^2} \tilde{p}_n^{2l} \\
& = \frac{\tilde{r}_n^*}{\tilde{p}_n^2} \tilde{p}_n^{2l} + \frac{\tilde{q}_n^* \tilde{r}_n}{\tilde{p}_n^2} \tilde{p}_n^{2l} \frac{1-(\tilde{q}_n/\tilde{p}_n^2)^{l-1}}{\tilde{p}_n^2- \tilde{q}_n }+ \frac{\tilde{q}_n^*}{\tilde{q}_n} \tilde{q}_n^l - \frac{ (\tilde{p}_n^*)^2 }{\tilde{p}_n^2} \tilde{p}_n^{2l} \\
& = \frac{\tilde{r}_n^*}{\tilde{p}_n^2} \tilde{p}_n^{2l} + \frac{\tilde{q}_n^* \tilde{r}_n}{\tilde{p}_n^2 ( \tilde{p}_n^2 - \tilde{q}_n ) } \tilde{p}_n^{2l} - \frac{\tilde{q}_n^* \tilde{r}_n}{\tilde{q}_n ( \tilde{p}_n^2 - \tilde{q}_n ) } \tilde{q}_n^l   + \frac{\tilde{q}_n^*}{\tilde{q}_n} \tilde{q}_n^l - \frac{ (\tilde{p}_n^*)^2 }{\tilde{p}_n^2} \tilde{p}_n^{2l} \\
& = \frac{1}{\tilde{p}_n^2} \left( \tilde{r}_n^* + \frac{\tilde{q}_n^* \tilde{r}_n}{\tilde{p}_n^2 - \tilde{q}_n} - ( \tilde{p}_n^* )^2 \right) \tilde{p}_n^{2l} + \frac{\tilde{q}_n^*}{\tilde{q}_n} \left( 1 - \frac{\tilde{r}_n}{\tilde{p}_n^2 - \tilde{q}_n} \right) \tilde{q}_n^l . \label{eqVarYlTwoSummands}
\end{align}
Next, using \eqref{eqWalkRVs} and Jensen's inequality, we have
\begin{equation} \label{eqRnJensen}
\tilde{r}_n = \sum_{j \in \N, k \in \N_0} \left( \frac{j}{j+k} \right)^2 \sum_{i \in \N} f_n(i,j,k) - \tilde{q}_n \geq \left( \sum_{j \in \N, k \in \N_0}  \frac{j}{j+k} \sum_{i \in \N} f_n(i,j,k) \right)^2 - \tilde{q}_n = \tilde{p}_n^2 - \tilde{q}_n ,
\end{equation}
and so $1 - \tilde{r}_n / (\tilde{p}_n^2 - \tilde{q}_n) \leq 0$, i.e.\ the second term in \eqref{eqVarYlTwoSummands} is non-positive, so $\forall\ l > 1$,
\begin{equation} \label{eqVarYlRv}
\Var_n(Y_l) \leq \frac{1}{\tilde{p}_n^2} \left( \tilde{r}_n^* + \frac{\tilde{q}_n^* \tilde{r}_n}{\tilde{p}_n^2 - \tilde{q}_n} - ( \tilde{p}_n^* )^2 \right) \tilde{p}_n^{2l} .
\end{equation}
In the case $l = 1$, we have (again by Lemmas \ref{lemHitProbOneWalk} and \ref{lemHitProbTwoWalks})
\begin{align}
\Var_n(Y_l) & =  ( \tilde{r}_n^* + \tilde{q}_n^* ) - \tilde{p}_n^* \leq \tilde{r}_n^* + \frac{\tilde{q}_n^* \tilde{r}_n}{\tilde{p}_n^2 - \tilde{q}_n} - ( \tilde{p}_n^* )^2 = \frac{1}{\tilde{p}_n^2} \left( \tilde{r}_n^* + \frac{\tilde{q}_n^* \tilde{r}_n}{\tilde{p}_n^2 - \tilde{q}_n} - ( \tilde{p}_n^* )^2 \right) \tilde{p}_n^{2l} ,
\end{align}
where the inequality is \eqref{eqRnJensen} and $\tilde{p}_n^* \leq 1$; hence, \eqref{eqVarYlRv} holds for $l = 1$ as well. Finally, since $Y_0 = 1\ a.s.$, it is immediate that \eqref{eqVarYlRv} also holds for $l =0$. We next analyze the covariance terms in \eqref{eqTotalVarTerm2}. First, if $l' > l > 0$, we can use \eqref{eqYlToTwoWalks} and Lemmas \ref{lemHitProbOneWalk} and \ref{lemHitProbTwoWalks} to obtain
\begin{gather} 
\E_n [ Y_l Y_{l'}  ] = \P_n ( X_l^1 \in \hat{A} , X_{l'}^2 \in \hat{A}  ) = \tilde{p}_n^{l'-l} \P_n ( X_l^1 \in \hat{A} , X_l^2 \in \hat{A} ) = \tilde{p}_n^{l'-l} \E_n [ Y_l^2 ] , \\
\E_n [ Y_{l'} ] = \P ( X_{l'}^2 \in \hat{A} ) = \tilde{p}_n^* \tilde{p}_n^{l'-1} = \tilde{p}_n^* \tilde{p}_n^{l-1} \tilde{p}_n^{l'-l} = \P ( X_{l}^1 \in \hat{A} ) \tilde{p}_n^{l'-l} = \E_n [ Y_{l} ] \tilde{p}_n^{l'-l} , \\
\Rightarrow \Cov_n(Y_l,Y_{l'}) = \tilde{p}_n^{l'-l} \left( \E_n [ Y_l ^2 ] - ( \E_n [ Y_l ] )^2 \right) = \tilde{p}_n^{l'-l} \Var_n(Y_l) .
\end{gather}
On the other hand, if $l' > l = 0$, we have $Y_l = 1\ a.s.$, so $\Cov_n(Y_l,Y_{l'}) = 0 = \tilde{p}_n^{l'} \Var_n(Y_0)$. Hence, combined with \eqref{eqVarYlRv}, we have argued
\begin{equation} \label{eqYlCovariances}
\Cov_n(Y_l,Y_{l'}) = \tilde{p}_n^{l'-l} \Var_n(Y_l) \leq \frac{1}{\tilde{p}_n^2} \left( \tilde{r}_n^* + \frac{\tilde{q}_n^* \tilde{r}_n}{\tilde{p}_n^2 - \tilde{q}_n} - ( \tilde{p}_n^* )^2 \right) \tilde{p}_n^{l+l'}\ \forall\ l \in \N_0, l' > l .
\end{equation}
Hence, combining \eqref{eqTotalVarTerm2}, \eqref{eqVarYlRv}, and \eqref{eqYlCovariances}, we obtain
\begin{align} 
& \Var_n ( \E [ \hat{\vartheta}_{T_n}(\phi) | \D , \mathcal{T} ]  )  \\
& \quad\quad \leq \frac{1}{\tilde{p}_n^2} \left( \tilde{r}_n^* + \frac{\tilde{q}_n^* \tilde{r}_n}{\tilde{p}_n^2 - \tilde{q}_n} - ( \tilde{p}_n^* )^2 \right) \frac{\theta^2}{T_n^2} \left( \sum_{l=0}^{T_n-1} u_{T_n,l}^2 \tilde{p}_n^{2l} + 2 \sum_{l=0}^{T_n-1} u_{T_n,l} \sum_{l' = l+1}^{T_n-1} u_{T_n,l'} \tilde{p}_n^{l+l'} \right) \\
& \quad\quad \leq \frac{1}{\tilde{p}_n^2} \left( \tilde{r}_n^* + \frac{\tilde{q}_n^* \tilde{r}_n}{\tilde{p}_n^2 - \tilde{q}_n} - ( \tilde{p}_n^* )^2 \right) \frac{1}{T_n^2} \left( \sum_{l=0}^{T_n-1} u_{T_n,l}^2  + 2 \sum_{l=0}^{T_n-1} u_{T_n,l} \sum_{l' = l+1}^{T_n-1} u_{T_n,l'}  \right) \\
& \quad\quad = \frac{1}{\tilde{p}_n^2} \left( \tilde{r}_n^* + \frac{\tilde{q}_n^* \tilde{r}_n}{\tilde{p}_n^2 - \tilde{q}_n} - ( \tilde{p}_n^* )^2 \right) \left( \frac{1}{T_n}  \sum_{l=0}^{T_n-1} u_{T_n,l}  \right)^2 = \frac{1}{\tilde{p}_n^2} \left( \tilde{r}_n^* + \frac{\tilde{q}_n^* \tilde{r}_n}{\tilde{p}_n^2 - \tilde{q}_n} - ( \tilde{p}_n^* )^2 \right) ,
\end{align}
where the second inequality is simply $\theta, \tilde{p}_n \leq 1$, the first equality is immediate, and the second equality holds by definition of $u_{T_n,l}$. It clearly follows that
\begin{equation} \label{eqFinalVarRv}
\Var_n ( \E [ \hat{\vartheta}_{T_n}(\phi) | \D , \mathcal{T} ]  )  1 ( \Omega_{n,2} ) \leq \frac{1}{\tilde{p}_n^2} \left( \tilde{r}_n^* + \frac{\tilde{q}_n^* \tilde{r}_n}{\tilde{p}_n^2 - \tilde{q}_n} - ( \tilde{p}_n^* )^2 \right) 1 ( \Omega_{n,2} ) ,
\end{equation}
and so we can complete the proof by showing the right side of \eqref{eqFinalVarRv} tends to zero $a.s.$ Clearly, the right side is zero if $\omega \notin \Omega_{n,2}$; we aim to also show that, given $\gamma > 0$, $\exists\ N$ s.t.\ for $n > N$ and $\omega \in \Omega_{n,2}$, 
\begin{equation} \label{eqFinalStepVar}
\frac{1}{\tilde{p}_n(\omega)^2} \left( \tilde{r}_n^*(\omega) + \frac{\tilde{q}_n^*(\omega) \tilde{r}_n(\omega)}{\tilde{p}_n(\omega)^2 - \tilde{q}_n(\omega)} -  \tilde{p}_n^*(\omega)^2 \right) < \gamma .
\end{equation}
To prove \eqref{eqFinalStepVar}, we first recall that by \ref{assBranchDegSeq}, we have for $\omega \in \Omega_{n,2}$, $\tilde{p}_n^*(\omega) \geq \tilde{p}_n(\omega) > p_n - \delta_n$. Hence, since we are assuming $p_n \rightarrow 1$, and since $\delta_n \rightarrow 0$ by \ref{assBranchDegSeq}, we have for $\gamma' > 0$, $n$ sufficiently large, and such $\omega$, $\tilde{p}_n(\omega)^2 , \tilde{p}_n^*(\omega)^2 > 1 - \gamma'$. We thus obtain for $n$ large and $\omega \in \Omega_{n,2}$,
\begin{equation} \label{eqFinalStepVarPterms}
\frac{1}{\tilde{p}_n(\omega)^2} \left( \tilde{r}_n^*(\omega) + \frac{\tilde{q}_n^*(\omega) \tilde{r}_n(\omega)}{\tilde{p}_n(\omega)^2 - \tilde{q}_n(\omega)} -  \tilde{p}_n^*(\omega)^2 \right) < \frac{1}{1-\gamma'} \left( \tilde{r}_n^*(\omega) + \frac{\tilde{q}_n^*(\omega) \tilde{r}_n(\omega)}{1-\gamma' - \tilde{q}_n(\omega)} - (1-\gamma') \right) .
\end{equation}
To further upper bound the right side of \eqref{eqFinalStepVarPterms}, we note $\tilde{r}_n \leq 1 - \tilde{q}_n\ a.s.$ by the first equality in \eqref{eqRnJensen}. The same argument gives $\tilde{r}_n^* \leq 1 - \tilde{q}_n^*\ a.s.$ Note, however, that to use the second bound, we must ensure $1-\gamma' - \tilde{q}_n(\omega) > 0$. To this end, recall that $\tilde{q}_n(\omega) < 1 - \xi$ for $\omega \in \Omega_{n,2}$ by \ref{assBranchDegSeq}. Hence, assuming we choose $\gamma' < \xi$, we obtain $1-\gamma' - \tilde{q}_n(\omega) > 0$ for such $\omega$. Thus,
\begin{align} 
& \frac{1}{\tilde{p}_n(\omega)^2} \left( \tilde{r}_n^*(\omega) + \frac{\tilde{q}_n^*(\omega) \tilde{r}_n(\omega)}{\tilde{p}_n(\omega)^2 - \tilde{q}_n(\omega)} -  \tilde{p}_n^*(\omega)^2 \right) \\
& \quad\quad  < \frac{1}{1-\gamma'} \left( (1-\tilde{q}_n^*(\omega))+ \frac{\tilde{q}_n^*(\omega) (1-\tilde{q}_n(\omega))}{1-\gamma' - \tilde{q}_n(\omega)} - (1-\gamma') \right) \\
& \quad\quad = \frac{1}{1-\gamma'} \left( \tilde{q}_n^*(\omega) \left( \frac{ 1-\tilde{q}_n(\omega)}{1-\gamma' - \tilde{q}_n(\omega)} - 1 \right) + \gamma' \right) = \frac{1}{1-\gamma'} \left( \tilde{q}_n^*(\omega) \left( \frac{\gamma'}{1-\gamma' - \tilde{q}_n(\omega)} \right) + \gamma' \right)  \\
& \quad\quad = \frac{\gamma'}{1-\gamma'} \left( \frac{\tilde{q}_n^*(\omega)}{1-\gamma'-\tilde{q}_n(\omega)} + 1 \right) < \frac{\gamma'}{1-\gamma'} \left( \frac{\tilde{q}_n^*(\omega)}{\xi-\gamma'} + 1 \right) \leq \frac{\gamma'}{1-\gamma'} \left( \frac{1}{\xi-\gamma'} + 1 \right) \label{eqFinalStepVarDone} ,
\end{align}
where the first inequality uses \eqref{eqFinalStepVarPterms} and the bounds from the previous paragraph, the equalities are straightforward, the second inequality uses $\tilde{q}_n(\omega) < 1 - \xi$ for $\omega \in \Omega_{n,2}$ by \ref{assBranchDegSeq}, and the third uses $\tilde{q}_n^*(\omega) \leq 1$ (recall we have chosen $\gamma' < \xi$). Finally, it is straightforward to see the final bound in \eqref{eqFinalStepVarDone} tends to zero with $\gamma'$. Hence, for sufficiently small $\gamma'$, \eqref{eqFinalStepVar} follows, completing the proof.

\subsubsection{Notation for proofs of Lemmas \ref{lemHitProbOneWalk} and \ref{lemHitProbTwoWalks}}

In the next two subsections, we prove Lemmas \ref{lemHitProbOneWalk} and \ref{lemHitProbTwoWalks}. For these proofs, we let $\D$ denote the degree sequence $\{ d_{out}(i) , d_{in}^A(i) , d_{in}^B(i) \}_{i \in [n]}$, and we let $D$ denote a realization of this set. Note that the random variables defined in \eqref{eqWalkRVs} are all functions of $\D$; for a realization $D$ of $\D$, we let e.g.\ $\tilde{p}_{n,D}$ denote the realization of $\tilde{p}_n$. We similarly define $f_{n,D}, f_{n,D}^*$ for realizations of $f_n, f_n^*$, defined in \eqref{eqEmpDist}. Finally, letting $g(D) = \P(\cdot | \D = D)$, we have $\P_n(\cdot) = g(\D)$ by definition of $\P_n$. Hence, to prove Lemma \ref{lemHitProbOneWalk}, it suffices to show
\begin{equation}
\P ( X_l \in \hat{A} | \D = D ) = \begin{cases} \tilde{p}_{n,D}^* \tilde{p}_{n,D}^{l-1}  , & l \in \N \\ 1 , & l = 0 \end{cases} .
\end{equation}
while to prove Lemma \ref{lemHitProbTwoWalks}, it suffices to show
\begin{gather}
\P ( X_l^1 \in \hat{A} , X_{l'}^2 \in \hat{A} | \D = D ) = \begin{cases} \P ( X_l^1 \in \hat{A} , X_l^2 \in \hat{A} | \D = D ) \tilde{p}_{n,D}^{l'-l}  , & l \in \N \\ \tilde{p}_{n,D}^* \tilde{p}_{n,D}^{l'-1} , & l = 0 \end{cases} , \label{eqLemHitProbTwoWalksState1} \\ 
\P ( X_l^1 \in \hat{A} , X_l^2 \in \hat{A} | \D = D ) = \begin{cases} \tilde{r}_{n,D}^* \tilde{p}_{n,D}^{2(l-1)}  +  \sum_{t=2}^{l} \tilde{q}_{n,D}^* \tilde{q}_{n,D}^{t-2} \tilde{r}_{n,D} \tilde{p}_{n,D}^{2(l-t)} + \tilde{q}_{n,D}^* \tilde{q}_{n,D}^{l-1}  , & l \in \{2,3,\ldots\} \\ \tilde{r}_{n,D}^* + \tilde{q}_{n,D}^* , & l = 1  \\ 1 , & l = 0 \end{cases} . \label{eqLemHitProbTwoWalksState2}
\end{gather}

\subsubsection{Proof of Lemma \ref{lemHitProbOneWalk}} \label{appProofHitProbOneWalk}

The $l = 0$ case is trivial, since $X_0^1 = \phi \in \hat{A}$, so we assume $l \in \N$ moving forward. First, since $\hat{A}^C = \hat{B}$ is an absorbing set, we have $X_l^1 \in \hat{A} \Rightarrow X_{l-1}^1 \in \hat{A}$, so
\begin{align} \label{eq1stMomTotalProb}
\P ( X_l^1 \in \hat{A} | \D = D ) & = \P ( X_l^1 \in \hat{A} | X_{l-1}^1 \in \hat{A} , \D = D ) \P ( X_{l-1}^1 \in \hat{A} | \D = D ) .
\end{align}
For the first term in \eqref{eq1stMomTotalProb}, we have
\begin{align} \label{eq1stMomCondOnDeg}
& \P ( X_l^1 \in \hat{A} | X_{l-1}^1 \in \hat{A} , \D = D ) \\
& \quad = \sum_{j \in \N, k \in \N_0} \P ( X_l^1 \in \hat{A} | d_{in}^A(X_{l-1}^1) = j , d_{in}^B(X_{l-1}^1) = k , X_{l-1}^1 \in \hat{A} , \D = D ) \\
& \quad\quad\quad \quad\quad\quad \times \P ( d_{in}^A(X_{l-1}^1) = j , d_{in}^B(X_{l-1}^1) = k | X_{l-1}^1 \in \hat{A} , \D = D ) \\
& \quad = \begin{cases} \sum_{j \in \N, k \in \N_0} \frac{j}{j+k} \sum_{i \in \N} f_{n,D}(i,j,k) = \tilde{p}_{n,D} , & l \in \{2,3,\ldots\} \\ \sum_{j \in \N, k \in \N_0} \frac{j}{j+k} \sum_{i \in \N} f_{n,D}^*(i,j,k) = \tilde{p}_{n,D}^* , & l = 1 \end{cases} ,
\end{align}
where the second equality holds by Algorithm \ref{algTree}. More specifically, for $l > 1$, the degrees of $X_{l-1}^1$ are sampled from $f_{n,D}$ (Line \ref{algTreeDegSample} in Algorithm \ref{algTree}) after realizing $X_{l-1}^1$ (Line \ref{algTreeNextStep}), yielding the $\sum_{i \in \N} f_{n,D}(i,j,k)$ term; further, $X_l^1$ is chosen uniformly from the incoming neighbors of $X_{l-1}^1$ (Line \ref{algTreeNextStep}) after realizing the degrees of $X_{l-1}^1$, yielding the $j / (j+k)$ term (the $l=1$ case is similarly justified). Combining \eqref{eq1stMomTotalProb} and \eqref{eq1stMomCondOnDeg}, and using the fact that $X_0^1 = \phi \in \hat{A}$ by definition, completes the proof in the case $l = 1$. For $l > 1$, we again use \eqref{eq1stMomTotalProb} and \eqref{eq1stMomCondOnDeg} to obtain
\begin{equation}
\P ( X_l^1 \in \hat{A} | \D = D ) = \tilde{p}_{n,D} \P ( X_{l-1}^1 \in \hat{A} | \D = D ) = \cdots = \tilde{p}_{n,D}^{l-1} \P ( X_1^1 \in \hat{A} | \D = D ) = \tilde{p}_{n,D}^{l-1} \tilde{p}_{n,D}^* ,
\end{equation}
which completes the proof.

\subsubsection{Proof of Lemma \ref{lemHitProbTwoWalks}} \label{appProofHitProbTwoWalks}

We begin by proving the first statement in the lemma, i.e.\ \eqref{eqLemHitProbTwoWalksState1}. First, we note that for the $l = 0$ case, $X_0 = \phi \in \hat{A}$ by definition, so $\P ( X_l^1 \in \hat{A} , X_{l'}^2 \in \hat{A} | \D = D ) = \P ( X_{l'}^2 \in \hat{A} | \D = D )$, and the statement holds by Lemma \ref{lemHitProbOneWalk}. For the $l \in \N$ case, we first write
\begin{align}
\P ( X_l^1 \in \hat{A} , X_{l'}^2 \in \hat{A} | \D = D ) & = \P ( X_l^1 \in \hat{A} , X_{l'-1}^2 \in \hat{A} , X_{l'}^2 \in \hat{A} | \D = D ) \\
& = \P ( X_{l'}^2 \in \hat{A} | X_l^1 \in \hat{A} , X_{l'-1}^2 \in \hat{A} , \D = D )  \P ( X_l^1 \in \hat{A} , X_{l'-1}^2 \in \hat{A} | \D = D ) ,
\end{align}
where the first equality holds since $\hat{A}^C = \hat{B}$ is an absorbing set (i.e.\ $X_{l'}^2 \in \hat{A} \Rightarrow X_{l'-1}^2 \in \hat{A}$) and the second simply rewrites a conditional probability. Next, by the same argument as \eqref{eq1stMomCondOnDeg},
\begin{align}
\P ( X_{l'}^2 \in \hat{A} | X_l^1 \in \hat{A} , X_{l'-1}^2 \in \hat{A} , \D = D ) = \tilde{p}_{n,D} ,
\end{align}
where we have used the $l' > 1$ case of \eqref{eq1stMomCondOnDeg}, since $l' > l \geq 1$. Hence, the previous two equations give
\begin{align}
\P ( X_l^1 \in \hat{A} , X_{l'}^2 \in \hat{A} | \D = D ) & =  \tilde{p}_{n,D} \P ( X_l^1 \in \hat{A} , X_{l'-1}^2 \in \hat{A} | \D = D )\\
& = \cdots = \tilde{p}_{n,D}^{l'-l}  \P ( X_l^1 \in \hat{A} , X_l^2 \in \hat{A} | \D = D ) .
\end{align}
This completes the proof of \eqref{eqLemHitProbTwoWalksState1}. For the second statement, i.e.\ \eqref{eqLemHitProbTwoWalksState2}, the $l = 0$ case is trivial, since $X_0^1 = X_0^2 = \phi \in \hat{A}$ by definition, so we assume $l \in \N$ for the remainder of the proof. First, let $\tau = \inf \{ t \in \N_0 : X_t^1 \neq X_t^2 \}$ denote the first step at which the two walks diverge. Note that $X_0^1 = X_0^2 = \phi$ by definition, so $\tau \in \N\ a.s.$; also, due to the tree structure, the walks remain apart forever after diverging, i.e.\ $X_{\tau+1}^1 \neq X_{\tau+1}^2, X_{\tau+2}^1 \neq X_{\tau+1}^2, \ldots\ a.s.$ Next, for $l \in \N$, we write
\begin{align} \label{eq2ndMomTotalProb}
& \P ( X_l^1 \in \hat{A} , X_l^2 \in \hat{A} | \D = D ) \\
& \quad\quad = \sum_{t=1}^l \P ( X_l^1 \in \hat{A} , X_l^2 \in \hat{A} , \tau = t | \D = D ) + \P ( X_l^1 \in \hat{A} , X_l^2 \in \hat{A} , \tau > l | \D = D )
\end{align}
We begin by computing the second term in \eqref{eq2ndMomTotalProb}. Here we have
\begin{align} \label{eq2ndMomTauGreaterInit}
& \P ( X_l^1 \in \hat{A} , X_l^2 \in \hat{A} , \tau > l | \D = D )  \\
& \quad = \P ( X_l^1 \in \hat{A} , X_l^2 \in \hat{A} , X_l^1 = X_l^2 , X_{l-1}^1 = X_{l-1}^2 , \ldots , X_1^1 = X_1^2 | \D = D ) \\
& \quad = \P ( X_l^1 \in \hat{A} , X_l^2 \in \hat{A} , X_l^1 = X_l^2 , X_{l-1}^1 \in \hat{A}, X_{l-1}^2 \in \hat{A} , X_{l-1}^1 = X_{l-1}^2 , \ldots , X_1^1 = X_1^2 | \D = D )  \\
& \quad = \P ( X_l^1 \in \hat{A} , X_l^2 \in \hat{A} , X_l^1 = X_l^2 | X_{l-1}^1 \in \hat{A}, X_{l-1}^2 \in \hat{A} , X_{l-1}^1 = X_{l-1}^2 , \ldots , X_1^1 = X_1^2 , \D = D ) \\
& \quad\quad\quad \times \P ( X_{l-1}^1 \in \hat{A}, X_{l-1}^2 \in \hat{A} , X_{l-1}^1 = X_{l-1}^2 , \ldots , X_1^1 = X_1^2 | \D = D ) \\
& \quad = \P ( X_l^1 \in \hat{A} , X_l^2 \in \hat{A} , X_l^1 = X_l^2 | X_{l-1}^1 \in \hat{A}, X_{l-1}^2 \in \hat{A} , X_{l-1}^1 = X_{l-1}^2 , \ldots , X_1^1 = X_1^2 , \D = D ) \\
& \quad\quad\quad \times \P ( X_{l-1}^1 \in \hat{A}, X_{l-1}^2 \in \hat{A} , \tau > l-1 | \D = D ) ,
\end{align}
where the first and last equalities hold by definition of $\tau$ and the second holds since $\hat{A}^C = \hat{B}$ is an absorbing set. Now for $l > 1$, we obtain
\begin{align} \label{eq2ndMomTauGreaterToQhat}
& \P ( X_l^1 \in \hat{A} , X_l^2 \in \hat{A} , X_l^1 = X_l^2 | X_{l-1}^1 \in \hat{A}, X_{l-1}^2 \in \hat{A} , X_{l-1}^1 = X_{l-1}^2 , \ldots , X_1^1 = X_1^2 , \D = D ) \\
& \quad = \P ( X_l^1 \in \hat{A} , X_l^1 = X_l^2 | X_{l-1}^1 \in \hat{A},  X_{l-1}^1 = X_{l-1}^2 , \D = D ) \\
& \quad = \sum_{j \in \N, k \in \N_0} \P ( X_l^1 \in \hat{A} , X_l^1 = X_l^2 | d_{in}^A(X_{l-1}^1) = j, d_{in}^B(X_{l-1}) = k , X_{l-1}^1 \in \hat{A},  X_{l-1}^1 = X_{l-1}^2 , \D = D ) \\
& \quad\quad\quad\quad\quad\quad\quad\quad\quad \times \P ( d_{in}^A(X_{l-1}^1) = j, d_{in}^B(X_{l-1}) = k | X_{l-1}^1 \in \hat{A},  X_{l-1}^1 = X_{l-1}^2 , \D = D ) \\
& \quad = \sum_{j \in \N, k \in \N_0} \frac{j}{j+k} \frac{1}{j+k} \sum_{i \in \N} f_{n,D}(i,j,k) = \tilde{q}_{n,D} ,
\end{align}
where the first equality uses independence and eliminates repetitive events, and the third follows an argument similar to that following \eqref{eq1stMomCondOnDeg}. Combining \eqref{eq2ndMomTauGreaterInit} and \eqref{eq2ndMomTauGreaterToQhat},
\begin{align} \label{eq2ndMomTauGreaterRecursion}
\P ( X_l^1 \in \hat{A} , X_l^2 \in \hat{A} , \tau > l | \D = D ) & = \tilde{q}_{n,D} \P ( X_{l-1}^1 \in \hat{A}, X_{l-1}^2 \in \hat{A} , \tau > l-1 | \D = D ) \\
& = \cdots = \tilde{q}_{n,D}^{l-1} \P ( X_1^1 \in \hat{A}, X_1^2 \in \hat{A} , \tau > 1 | \D = D ) .
\end{align}
Finally, by an argument similar to \eqref{eq2ndMomTauGreaterToQhat}, we have
\begin{align} \label{eq2ndMomTauGreaterToQhatStar}
& \P ( X_1^1 \in \hat{A}, X_1^2 \in \hat{A} , \tau > 1 | \D = D )  = \P ( X_1^1 \in \hat{A} , X_1^1 = X_1^2 | \D = D ) \\
& \quad = \sum_{j \in \N, k \in \N_0} \P ( X_1^1 \in \hat{A} , X_1^1 = X_1^2 | d_{in}^A(\phi) = j , d_{in}^B(\phi) = k , \D = D ) \P ( d_{in}^A(\phi) = j , d_{in}^B(\phi) = k | \D = D ) \\
& \quad = \sum_{j \in \N, k \in \N_0} \frac{j}{j+k} \frac{1}{j+k} \sum_{i \in \N} f_{n,D}^*(i,j,k) = \tilde{q}_{n,D}^* .
\end{align}
Hence, combining \eqref{eq2ndMomTauGreaterRecursion} and \eqref{eq2ndMomTauGreaterToQhatStar} gives
\begin{equation} \label{eq2ndMomTauGreaterFinal} 
\P ( X_l^1 \in \hat{A} , X_l^2 \in \hat{A} , \tau > l | \D = D ) = \tilde{q}_{n,D}^* \tilde{q}_{n,D}^{l-1}\ \forall\ l \in \N .
\end{equation}
For the first term in \eqref{eq2ndMomTotalProb}, we first consider the $t = l$ summand. For $l > 1$, similar to \eqref{eq2ndMomTauGreaterToQhat},
\begin{align}
& \P ( X_l^1 \in \hat{A} , X_l^2 \in \hat{A} , \tau = l | \D = D ) \\
& \quad = \P ( X_l^1 \in \hat{A} , X_l^2 \in \hat{A} , X_l^1 \neq X_l^2 , X_{l-1}^1 = X_{l-1}^2 , \ldots , X_1^1 = X_1^2 | \D = D ) \\
& \quad = \P ( X_l^1 \in \hat{A} , X_l^2 \in \hat{A} , X_l^1 \neq X_l^2 , X_{l-1}^1 \in \hat{A} , X_{l-1}^1 = X_{l-1}^2 , \ldots , X_1^1 = X_1^2 | \D = D ) \\
& \quad = \P ( X_l^1 \in \hat{A} , X_l^2 \in \hat{A} , X_l^1 \neq X_l^2 | X_{l-1}^1 \in \hat{A} , X_{l-1}^1 = X_{l-1}^2 , \ldots , X_1^1 = X_1^2 , \D = D ) \\
& \quad\quad\quad \times \P ( X_{l-1}^1 \in \hat{A} , X_{l-1}^1 = X_{l-1}^2 , \ldots , X_1^1 = X_1^2 | \D = D ) \\
& \quad = \sum_{j \in \N, k \in \N_0} \frac{j}{j+k} \frac{j-1}{j+k} \sum_{i \in \N} f_{n,D}(i,j,k) \P ( X_{l-1}^1 \in \hat{A}, X_{l-1}^2 \in \hat{A} , \tau > l-1 | \D = D )  = \tilde{r}_{n,D} \tilde{q}_{n,D}^{l-2} \tilde{q}_{n,D}^* ,
\end{align}
where in the final step we have also used \eqref{eq2ndMomTauGreaterFinal}. Similarly, for $l=1$,
\begin{align}
\P ( X_1^1 \in \hat{A} , X_1^2 \in \hat{A} , \tau = 1 | \D = D ) & = \P ( X_1^1 \in \hat{A} , X_1^2 \in \hat{A} , X_1^1 \neq X_1^2 | \D = D ) \\
&  = \sum_{j \in \N, k \in \N_0} \frac{j}{j+k} \frac{j-1}{j+k} \sum_{i \in \N} f_{n,D}^* (i,j,k) = \tilde{r}_{n,D}^* .
\end{align}
To summarize, we have shown
\begin{equation} \label{eq2ndMomSummandL}
\P ( X_l^1 \in \hat{A} , X_l^2 \in \hat{A} , \tau = l | \D = D ) = \begin{cases} \tilde{q}_{n,D}^*  \tilde{q}_{n,D}^{l-2} \tilde{r}_{n,D}  , &  l \in \{2,3,\ldots\} \\ \tilde{r}_{n,D}^* , &  l = 1 \end{cases} .
\end{equation}
Next, we consider the $t < l$ summands in \eqref{eq2ndMomTotalProb} (such summands are present only for $l > 1$). We have
\begin{align}
& \P ( X_l^1 \in \hat{A} , X_l^2 \in \hat{A} , \tau = t | \D = D ) \\
& \quad = \P ( X_l^1 \in \hat{A} , X_l^2 \in \hat{A} , X_{l-1}^1 \in \hat{A} , X_{l-1}^2 \in \hat{A} , X_{l-1}^1 \neq X_{l-1}^2 , \tau = t | \D = D ) \\
& \quad = \P ( X_l^1 \in \hat{A} , X_l^2 \in \hat{A} | X_{l-1}^1 \in \hat{A} , X_{l-1}^2 \in \hat{A} , X_{l-1}^1 \neq X_{l-1}^2 , \tau = t , \D = D ) \\
& \quad\quad\quad \times \P( X_{l-1}^1 \in \hat{A} , X_{l-1}^2 \in \hat{A} , X_{l-1}^1 \neq X_{l-1}^2 , \tau = t | \D = D ) \\
& \quad = \prod_{h=1}^2 \P ( X_l^h \in \hat{A} | X_{l-1}^1 \in \hat{A} , X_{l-1}^2 \in \hat{A} , X_{l-1}^1 \neq X_{l-1}^2 , \tau = t , \D = D ) \\
& \quad\quad\quad \times \P( X_{l-1}^1 \in \hat{A} , X_{l-1}^2 \in \hat{A} , \tau = t | \D = D ) ,
\end{align}
where in the first equality we used the fact that $\hat{A}^C = \hat{B}$ is an absorbing set and the fact that once the walks diverge they remain apart; in the second equality we used the fact that $X_l^1$ and $X_l^2$ are conditionally independent given the event $X_{l-1}^1 \neq X_{l-1}^2$. Further, for $h \in \{1,2\}$, 
\begin{align}
& \P ( X_l^h \in \hat{A}  | X_{l-1}^1 \in \hat{A} , X_{l-1}^2 \in \hat{A} , X_{l-1}^1 \neq X_{l-1}^2 , \tau = t , \D = D )  = \sum_{j \in \N, k \in \N_0} \frac{j}{j+k} \sum_{i \in \N} f_{n,D}(i,j,k) = \tilde{p}_{n,D} ,
\end{align}
and so, combining the previous two equations and applying recursively yields
\begin{align} \label{eq2ndMomSummandT}
\P ( X_l^1 \in \hat{A} , X_l^2 \in \hat{A} , \tau = t | \D = D ) & = \tilde{p}_{n,D}^2 \P( X_{l-1}^1 \in \hat{A} , X_{l-1}^2 \in \hat{A} , \tau = t | \D = D ) \\
& = \cdots = \tilde{p}_{n,D}^{2(l-t)} \P( X_{t}^1 \in \hat{A} , X_{t}^2 \in \hat{A} , \tau = t | \D = D ) \\
& = \begin{cases}  \tilde{q}_{n,D}^* \tilde{q}_{n,D}^{t-2} \tilde{r}_{n,D}  \tilde{p}_{n,D}^{2(l-t)}  , & t \in \{2,3,\ldots,l-1\} \\ \tilde{r}_{n,D}^* \tilde{p}_{n,D}^{2(l-1)}  , & t = 1 \end{cases}\ \forall\ l \in \{2,3,\ldots\} .
\end{align}
where the final equality uses \eqref{eq2ndMomSummandL}. Finally, combining \eqref{eq2ndMomTotalProb}, \eqref{eq2ndMomTauGreaterFinal}, \eqref{eq2ndMomSummandL}, and \eqref{eq2ndMomSummandT} yields
\begin{equation}
\P ( X_l^1 \in \hat{A} , X_l^2 \in \hat{A} | \D = D ) = \begin{cases} \tilde{r}_{n,D}^* \tilde{p}_{n,D}^{2(l-1)}  +  \sum_{t=2}^{l} \tilde{q}_{n,D}^* \tilde{q}_{n,D}^{t-2} \tilde{r}_{n,D} \tilde{p}_{n,D}^{2(l-t)} + \tilde{q}_{n,D}^* \tilde{q}_{n,D}^{l-1}  , & l \in \{2,3,\ldots\} \\ \tilde{r}_{n,D}^* + \tilde{q}_{n,D}^* , & l = 1  \end{cases} ,
\end{equation}
which is what we set out to prove.

\subsection{Step 2 for proof of Theorem \ref{thmSec}} \label{appRootBelief2ProofDetails}

\subsubsection{Proof of Lemma \ref{lemSecThmSignals}} \label{appProofSecThmSignals}

We first write
\begin{align} 
& \P \left( \left| \hat{\vartheta}_{T_n}(\phi) - \E [ \hat{\vartheta}_{T_n}(\phi) | \T ] \right| > \epsilon \right) = \E \left[ \P \left( \left| \hat{\vartheta}_{T_n}(\phi) - \E [ \hat{\vartheta}_{T_n}(\phi) | \T ] \right| > \epsilon \middle| \T \right) \right] \\
& \quad = \E \left[ \P \left( \hat{\vartheta}_{T_n}(\phi) - \E [ \hat{\vartheta}_{T_n}(\phi) | \T ] > \epsilon \middle| \T \right) + \P \left( \E [ \hat{\vartheta}_{T_n}(\phi) | \T ] - \hat{\vartheta}_{T_n}(\phi) > \epsilon \middle| \T \right) \right] \label{eqSpOnlySignalsChernoff}
\end{align}
where the first equality uses the law of total expectation and the second is immediate. For the first summand in the expectation in \eqref{eqSpOnlySignalsChernoff}, we fix $\lambda > 0$ and write
\begin{align}
& \P \left( \hat{\vartheta}_{T_n}(\phi) - \E [ \hat{\vartheta}_{T_n}(\phi) | \T ] > \epsilon \middle| \T \right) = \P \left( \exp( \lambda( \hat{\vartheta}_{T_n}(\phi) - \E [ \hat{\vartheta}_{T_n}(\phi) | \T ] )) > e^{- \lambda \epsilon} \middle| \T \right) \\
& \quad \leq e^{-\lambda \epsilon} \E \left[ \exp ( \lambda ( \hat{\vartheta}_{T_n}(\phi) - \E [ \hat{\vartheta}_{T_n}(\phi) | \T ] ) \middle| \T \right] \\
& \quad = e^{-\lambda \epsilon} \prod_{t=0}^{T_n-1} \prod_{l=0}^t \prod_{\i \in \hat{A}_l} \E \left[ \exp \left(  \frac{\lambda}{T_n} {t \choose l} \eta^l (1-\eta)^{t-l} \prod_{j=0}^{l-1} d_{in}(\i | j )^{-1} ( \hat{s}_{T_n-t}(\i) - \theta ) \right) \middle| \T \right] \\
& \quad \leq e^{-\lambda \epsilon} \prod_{t=0}^{T_n-1} \prod_{l=0}^t \prod_{\i \in \hat{A}_l} \exp \left( \frac{1}{8} \left( \frac{\lambda}{T_n} {t \choose l} \eta^l (1-\eta)^{t-l} \prod_{j=0}^{l-1} d_{in}(\i | j )^{-1} \right)^2 \right) \\
& \quad \leq e^{-\lambda \epsilon} \prod_{t=0}^{T_n-1} \prod_{l=0}^t \prod_{\i \in \hat{A}_l} \exp \left( \frac{\lambda^2}{8 T_n^2} {t \choose l} \eta^l (1-\eta)^{t-l} \prod_{j=0}^{l-1} d_{in}(\i | j )^{-1}  \right) \\
& \quad = \exp \left( -\lambda \epsilon + \frac{\lambda^2}{8 T_n} \frac{1}{T_n} \sum_{t=0}^{T_n-1} \sum_{l=0}^t \sum_{\i \in \hat{A}_l} {t \choose l} \eta^l (1-\eta)^{t-l} \prod_{j=0}^{l-1} d_{in}(\i | j )^{-1}  \right) \\
& \quad = \exp \left( -\lambda \epsilon +  \frac{\lambda^2}{8 T_n \theta} \E [ \hat{\vartheta}_{T_n}(\phi) | \T ]  \right) \leq \exp \left( - \lambda \epsilon + \frac{\lambda^2}{8 T_n} \right) . \label{eqSpOnlySignalsChernoffBeforeMin}
\end{align}
Here the first equality holds by monotonicity of $x \mapsto e^{\lambda x}$, the first inequality is Markov's, the second equality holds by \eqref{eqMeanTreeBeliefGivenTree}, the second inequality uses Lemma \ref{lemHoeffding} from Appendix \ref{appAuxiliary}, the third inequality uses ${t \choose l} \eta^l (1-\eta)^{t-l}, \prod_{j=0}^{l-1} d_{in}(\i | j )^{-1} \leq 1$, the third equality is immediate, the fourth equality again uses \eqref{eqMeanTreeBeliefGivenTree}, and the fourth inequality uses \eqref{eqBeliefGivenTreeIn0theta}. Since the preceding argument holds $\forall\ \lambda > 0$, we choose $\lambda = 4 \epsilon T_n$ to minimize the bound. Upon substituting into \eqref{eqSpOnlySignalsChernoffBeforeMin}, we obtain $e^{-2 \epsilon^2 T_n}$. The same argument holds for the second summand in the expectation of \eqref{eqSpOnlySignalsChernoff}. We also note that the bound $e^{-2 \epsilon^2 T_n}$ is non-random, so we may discard the expectation. In summary, we have shown
\begin{equation}
\P \left( \left| \hat{\vartheta}_{T_n}(\phi) - \E [ \hat{\vartheta}_{T_n}(\phi) | \T ] \right| > \epsilon \right) \leq 2 e^{-2 \epsilon^2 T_n} .
\end{equation}
Hence, for $n$ sufficiently large, we have by assumption on $T_n$
\begin{equation}
\P \left( \left| \hat{\vartheta}_{T_n}(\phi) - \E [ \hat{\vartheta}_{T_n}(\phi) | \T ] \right| > \epsilon \right) \leq 2 e^{-2 \epsilon^2 \mu \log n} = 2 n^{-2 \epsilon^2 \mu} = O \left( n^{-2 \epsilon^2 \mu} \right) ,
\end{equation}
which is what we set out to prove.

\subsubsection{Proof of Lemma \ref{lemSecThmStructure}} \label{appProofSecThmStructure}

We begin by deriving a bound conditioned on the degree sequence. First, we fix $\tilde{\lambda} > 0$ and use monotonicity of $x \mapsto e^{\tilde{\lambda} x}$ and Markov's inequality to write
\begin{equation} \label{eqStructureInitChernoff}
\P_n ( \E [ \hat{\vartheta}_{T_n}(\phi) | \T ] > \epsilon  ) \leq e^{-\tilde{\lambda} \epsilon} \E_n \exp ( \tilde{\lambda} \E [ \hat{\vartheta}_{T_n}(\phi) | \T ] ) .
\end{equation}
The bulk of the proof will involve bounding the expectation term. For this, we first note
\begin{align}
& \E_n \exp ( \tilde{\lambda} \E [ \hat{\vartheta}_{T_n}(\phi) | \T ] ) = \E_n \exp \left( \tilde{\lambda} \frac{\theta}{T_n} \sum_{t=0}^{T_n-1} \sum_{l=0}^t { t \choose l } \eta^l (1-\eta)^{t-l}  \sum_{\i \in \hat{A}_l} \prod_{j=0}^{l-1} d_{in}(\i|j)^{-1} \right)  \\
& \quad = \E_n \exp \left(  \frac{\tilde{\lambda} \theta}{T_n} \sum_{l=0}^{T_n-1} \left( \sum_{t=l}^{T_n-1} { t \choose l } \eta^l (1-\eta)^{t-l}  \right) \left( \sum_{\i \in \hat{A}_l} \prod_{j=0}^{l-1} d_{in}(\i|j)^{-1} \right) \right)  = \E_n \prod_{l=0}^{T_n-1} \exp ( \lambda u_{T_n,l} Y_l ) ,
\end{align}
where the first equality holds by \eqref{eqMeanTreeBeliefGivenTree}, the second rearranges summations, and in the third we have defined $\lambda = \tilde{\lambda} \theta / T_n$, $u_{T_n,l} = \sum_{t=l}^{T_n-1} { t \choose l } \eta^l (1-\eta)^{t-l}$, and $Y_l = \sum_{\i \in \hat{A}_l} \prod_{j=0}^{l-1} d_{in}(\i|j)^{-1}$. For the remainder of the proof, we use $\E_{n,l}$ to denote conditional expectation with respect to the degree sequence \textit{and} the set of random variables realized during the first $l$ iterations of Algorithm \ref{algTree} (i.e.\ the random variables defining the first $l$ generations of the tree). Using this notation, we have
\begin{align} \label{eqHoeffFirstRecursion}
& \E_n \Big[ \prod_{l=0}^{T_n-1} \exp ( \lambda u_{T_n,l} Y_l ) \Big] = \E_n \Big[ \E_{n,T_n-2} \Big[ \prod_{l=0}^{T_n-1} \exp ( \lambda u_{T_n,l} Y_l ) \Big] \Big] \\
& \quad = \E_n \Big[ \prod_{l=0}^{T_n-2} \exp ( \lambda u_{T_n,l} Y_l ) \E_{n,T_n-2} \Big[ \exp ( \lambda u_{T_n,T_n-1} Y_{T_n-1} ) \Big] \Big] \\
& \quad = \E_n \Big[ \prod_{l=0}^{T_n-3} \exp ( \lambda u_{T_n,l} Y_l ) \exp ( \lambda ( u_{T_n,T_n-2} + u_{T_n,T_n-1} \tilde{p}_n ) Y_{T_n-2} ) \\
& \quad\quad \times \E_{n,T_n-2} \Big[  \exp ( \lambda u_{T_n,T_n-1} ( Y_{T_n-1} - \tilde{p}_n Y_{T_n-2} ) ) \Big] \Big] ,
\end{align}
where in the third equality we have multiplied and divided $\exp ( \lambda u_{T_n,T_n-1} \tilde{p}_n Y_{T_n-2} )$. Next, we note
\begin{align} \label{eqYlToYlmin1final}
& Y_{T_n-1} = \sum_{\i' \in \hat{A}_{T_n-2}} \sum_{\i \in \hat{A}_{T_n-1} : \i | (T_n-2) = \i' } \prod_{j=0}^{T_n-2} d_{in}(\i|j)^{-1}  = \sum_{\i' \in \hat{A}_{T_n-2}} \sum_{\i \in \hat{A}_{T_n-1} : \i | (T_n-2) = \i' } \prod_{j=0}^{T_n-2} d_{in}(\i'|j)^{-1} \\
& \quad = \sum_{\i' \in \hat{A}_{T_n-2}} \prod_{j=0}^{T_n-2} d_{in}(\i'|j)^{-1} | \{ \i \in \hat{A}_{T_n-1} : \i | (T_n-2) = \i' \} | = \sum_{\i' \in \hat{A}_{T_n-2}} \prod_{j=0}^{T_n-3} d_{in}(\i'|j)^{-1} d_{in}(\i')^{-1} d_{in}^A(\i') , 
\end{align}
where in the first equality we rewrote the sum based on the construction of $\hat{A}_{T_n-1}$ in Algorithm \ref{algTree}, in the second we have used the fact that $\i | j = \i' | j$ for $j \in \{0,\ldots,T_n-2\}$ by Algorithm \ref{algTree} (in words, $\i$ and $\i'$ share the same ancestry in the tree), in the third we have recognized that the $\i$-th summand does not depend on $\i$, and in the fourth we have used $\i' | (T_n-2) = \i'$ (since $\i' \in \hat{A}_{T_n-2}$) and the construction of the agent offspring of $\i'$ in Algorithm \ref{algTree}. It follows that
\begin{equation}
\E_{n,T_n-2} Y_{T_n-1} = \sum_{\i' \in \hat{A}_{T_n-2}} \prod_{j=0}^{T_n-3} d_{in}(\i'|j)^{-1} \E_{n,T_n-2} ( d_{in}^A(\i') / d_{in}(\i') ) = \prod_{j=0}^{T_n-3} d_{in}(\i'|j)^{-1} \tilde{p}_n = Y_{T_n-2} \tilde{p}_n ,
\end{equation}
where $\E_{n,T_n-2} ( d_{in}^A(\i') / d_{in}(\i') ) = \tilde{p}_n$ holds by definition of $d_{in}^A(\i'), d_{in}(\i')$ in Algorithm \ref{algTree} and of $\tilde{p}_n$ from \eqref{eqWalkRVs}. In summary, we have argued $\E_{n,T_n-2} ( Y_{T_n-1} -  Y_{T_n-2} \tilde{p}_n ) = 0$. On the other hand, we note $0 \leq Y_{T_n-1} \leq Y_{T_n-2} \leq \cdots \leq Y_0 = 1$, where the first inequality holds since $Y_{T_n-1}$ is a sum of nonnegative terms and the second holds by \eqref{eqYlToYlmin1final} (using $d_{in}(\i') = d_{in}^A(\i') + d_{in}^B(\i') \geq d_{in}^A(\i')$), and where $Y_0 = 1$ by definition. Hence, we can use Lemma \ref{lemHoeffding} from Appendix \ref{appAuxiliary} to obtain
\begin{equation} \label{eqSpStructureFirstHoeffdingApp}
\E_{n,T_n-2} \exp ( \lambda u_{T_n,T_n-1} ( Y_{T_n-1} - \tilde{p}_n Y_{T_n-2} ) ) \leq e^{\lambda^2 u_{T_n,T_n-1}^2/8} .
\end{equation}
Substituting into \eqref{eqHoeffFirstRecursion} then yields
\begin{align} \label{eqSubFirstRecursion}
& \E_n \Big[ \prod_{l=0}^{T_n-1} \exp ( \lambda u_{T_n,l} Y_l ) \Big] \\
& \quad \leq \E_n \Big[ \prod_{l=0}^{T_n-3} \exp ( \lambda u_{T_n,l} Y_l ) \exp ( \lambda ( u_{T_n,T_n-2} + u_{T_n,T_n-1} \tilde{p}_n ) Y_{T_n-2} ) \Big] \exp \left( \frac{\lambda^2}{8} u_{T_n,T_n-1}^2 \right) .
\end{align}
We can then iteratively apply the preceding argument. Namely, we have
\begin{align}
& \E_n \Big[ \prod_{l=0}^{T_n-3} \exp ( \lambda u_{T_n,l} Y_l ) \exp ( \lambda ( u_{T_n,T_n-2} + u_{T_n,T_n-1} \tilde{p}_n ) Y_{T_n-2} ) \Big] \exp \left( \frac{\lambda^2}{8} u_{T_n,T_n-1}^2 \right)  \\
& \quad = \E_n \Big[ \prod_{l=0}^{T_n-4} \exp ( \lambda u_{T_n,l} Y_l ) \exp ( \lambda ( u_{T_n,T_n-3} + u_{T_n,T_n-2} \tilde{p}_n + u_{T_n,T_n-1} \tilde{p}_n^2 ) Y_{T_n-3} ) \\
& \quad\quad \times \E_{n,T_n-3} \Big[ \exp ( \lambda ( u_{T_n,T_n-2} + u_{T_n,T_n-1} \tilde{p}_n ) ( Y_{T_n-2} - \tilde{p}_n Y_{T_n-3} ) ) \Big] \Big] \exp \left( \frac{\lambda^2}{8} u_{T_n,T_n-1}^2 \right) \\
& \quad \leq \E_n \Big[ \prod_{l=0}^{T_n-4} \exp ( \lambda u_{T_n,l} Y_l ) \exp ( \lambda ( u_{T_n,T_n-3} + u_{T_n,T_n-2} \tilde{p}_n + u_{T_n,T_n-1} \tilde{p}_n^2 ) Y_{T_n-3} ) \Big] \label{eqVerify1} \\
& \quad\quad \times \exp \left( \frac{\lambda^2}{8} \left( ( u_{T_n,T_n-2} + u_{T_n,T_n-1} \tilde{p}_n )^2 + u_{T_n,T_n-1}^2 \right) \right) \label{eqVerify2} \\
& \quad \leq \cdots \leq \E_n \Big[ \exp ( \lambda u_{T_n,0} Y_0 ) \exp \left( \lambda \sum_{l=1}^{T_n-1} u_{T_n,l} \tilde{p}_n^{l-1} Y_1 \right) \Big] \exp \left( \frac{\lambda^2}{8} \sum_{l=2}^{T_n-1} \left( \sum_{l'=l}^{T_n-1} u_{T_n,l'} \tilde{p}_n^{l'-l} \right)^2 \right) \label{eqVerify3} .
\end{align}
(The precise form of the summations in \eqref{eqVerify3} can be verified by considering the case $T_n = 4$ in \eqref{eqVerify1} and \eqref{eqVerify2}.) Note that the final step of the iteration is slightly different; this is because the root node has degrees sampled from $f_n^*$ (the uniform distribution) instead of $f_n$ (the size-biased distribution) in Algorithm \ref{algTree}. Nevertheless, a similar argument holds: here we have $\E_{n,0} Y_1 = \tilde{p}_n^* Y_0$ and $Y_1 \in [0,1]\ a.s.$, so by an argument similar to that leading to \eqref{eqSpStructureFirstHoeffdingApp},
\begin{align}
& \E_n \Big[ \exp ( \lambda u_{T_n,0} Y_0 ) \exp \left( \lambda \sum_{l=1}^{T_n-1} u_{T_n,l} \tilde{p}_n^{l-1} Y_1 \right) \Big] \\
& \quad = \E_n \Big[ \exp \left( \lambda \left( u_{T_n,0} + \tilde{p}_n^* \sum_{l=1}^{T_n-1} u_{T_n,l} \tilde{p}_n^{l-1} \right) Y_0 \right) \E_{n,0} \Big[ \exp \left( \lambda \sum_{l=1}^{T_n-1} u_{T_n,l} \tilde{p}_n^{l-1} ( Y_1 - \tilde{p}_n^* Y_0 ) \right) \Big] \Big] \\
& \quad \leq \E_n \Big[ \exp \left( \lambda \left( u_{T_n,0} + \tilde{p}_n^* \sum_{l=1}^{T_n-1} u_{T_n,l} \tilde{p}_n^{l-1} \right) Y_0 \right) \Big] \exp \left( \frac{\lambda^2}{8} \left(  \sum_{l=1}^{T_n-1} u_{T_n,l} \tilde{p}_n^{l-1} \right)^2 \right) .
\end{align}
Combining the previous inequality with \eqref{eqSubFirstRecursion} and \eqref{eqVerify3} then yields
\begin{align}
& \E_n \Big[ \prod_{l=0}^{T_n-1} \exp ( \lambda u_{T_n,l} Y_l ) \Big] \\
& \quad \leq \E_n \Big[ \exp \left( \lambda \left( u_{T_n,0} + \tilde{p}_n^* \sum_{l=1}^{T_n-1} u_{T_n,l} \tilde{p}_n^{l-1} \right) Y_0 \right) \Big] \exp \left( \frac{\lambda^2}{8} \sum_{l=1}^{T_n-1} \left( \sum_{l'=l}^{T_n-1} u_{T_n,l'} \tilde{p}_n^{l'-l} \right)^2 \right)
\end{align}
Next, we recall $Y_0 = 1$ by definition. Additionally, we have
\begin{align}
u_{T_n,0} + \tilde{p}_n^* \sum_{l=1}^{T_n-1} u_{T_n,l} \tilde{p}_n^{l-1} & = \sum_{t=0}^{T_n-1} (1-\eta)^t + \tilde{p}_n^* \sum_{l=1}^{T_n-1} \sum_{t=l}^{T_n-1} {t \choose l} \eta^l (1-\eta)^{t-l} \tilde{p}_n^{l-1} \\
& = \sum_{t=0}^{T_n-1} \left( \sum_{t=1}^l {t \choose l} \eta^l (1-\eta)^{t-l} \tilde{p}_n^* \tilde{p}_n^{l-1} + (1-\eta)^t \right) = \frac{T_n}{\theta} \E_n [ \hat{\vartheta}_{T_n}(\phi) ] ,
\end{align}
where the first equality uses the definition of $u_{T_n,l}$, the second rearranges summations, and the third uses \eqref{eqTreeBeliefMeanWalkSub}. Combining the previous two equations therefore yields
\begin{equation}
\E_n \Big[ \prod_{l=0}^{T_n-1} \exp ( \lambda u_{T_n,l} Y_l ) \Big] \leq \exp \left( \lambda \frac{T_n}{\theta} \E_n [ \hat{\vartheta}_{T_n}(\phi) ] + \frac{\lambda^2}{8} \sum_{l=1}^{T_n-1} \left( \sum_{l'=l}^{T_n-1} u_{T_n,l'} \tilde{p}_n^{l'-l} \right)^2 \right) .
\end{equation}
Hence, recalling that $\lambda = \tilde{\lambda} \theta / T_n$, and substituting into \eqref{eqStructureInitChernoff}, we have shown
\begin{equation} \label{eqWhereItFails}
\P_n ( \E [ \hat{\vartheta}_{T_n}(\phi) | \T ] > \epsilon  ) \leq \exp \left( - \tilde{\lambda} \epsilon + \tilde{\lambda} \E_n [ \hat{\vartheta}_{T_n}(\phi) ] + \frac{\tilde{\lambda}^2 \theta^2}{8 T_n^2} \sum_{l=1}^{T_n-1} \left( \sum_{l'=l}^{T_n-1} u_{T_n,l'} \tilde{p}_n^{l'-l} \right)^2 \right) .
\end{equation}
Clearly, this inequality still holds if we multiply both sides by $1(\Omega_{n,2})$. Additionally, by \ref{assBranchDegSeq}, $\tilde{p}_n(\omega) < p_n + \delta_n$ for $\omega \in \Omega_{n,2}$, where $p_n \rightarrow p$ and $\delta_n \rightarrow 0$; since we additionally assume $p < 1$ in the statement of the lemma, we conclude $\tilde{p}_n(\omega) < p_n + \delta_n < 1$ for $\omega \in \Omega_{n,2}$ and $n$ sufficiently large. For such $n$, we can therefore write
\begin{align}
& \P_n ( \E [ \hat{\vartheta}_{T_n}(\phi) | \T ] > \epsilon  ) 1(\Omega_{n,2}) \\
& \quad \leq \exp \left( - \tilde{\lambda} \epsilon + \tilde{\lambda} \E_n [ \hat{\vartheta}_{T_n}(\phi) ] + \frac{\tilde{\lambda}^2 \theta^2}{8 T_n^2} \sum_{l=1}^{T_n-1} \left( \sum_{l'=l}^{T_n-1} u_{T_n,l'} (p_n+\delta_n)^{l'-l} \right)^2 \right) 1(\Omega_{n,2}) \\
& \quad \leq \exp \left( - \tilde{\lambda} \epsilon + \tilde{\lambda} \E_n [ \hat{\vartheta}_{T_n}(\phi) ] + \frac{\tilde{\lambda}^2 \theta^2}{8 T_n \eta^2 ( 1 -(p_n+\delta_n) )^2} \right) 1(\Omega_{n,2}) ,
\end{align}
where the second inequality uses Lemma \ref{lemUglySumInHoeffding} from Appendix \ref{appAuxiliary}. Additionally, since $p_n \rightarrow p < 1$, we can use the argument leading to \eqref{eqMeanConvCase4} to obtain $\E_n [ \hat{\vartheta}_{T_n}(\phi)](\omega)  < c / T_n$ (for some $c$ independent of $n$) whenever $\omega \in \Omega_{n,2}$ and $n$ is sufficiently large. For such $n$, we obtain
\begin{equation} \label{eqAfterUglySumGone}
\P_n ( \E [ \hat{\vartheta}_{T_n}(\phi) | \T ] > \epsilon  ) 1(\Omega_{n,2}) \leq \exp \left( - \tilde{\lambda} \epsilon + \frac{\tilde{\lambda} c}{T_n} + \frac{\tilde{\lambda}^2 \theta^2}{8 T_n \eta^2 ( 1 -(p_n+\delta_n) )^2} \right) 1(\Omega_{n,2}) ,
\end{equation}
Now since $\tilde{\lambda} > 0$ was arbitrary, we can choose $\tilde{\lambda} = 4 T_n \epsilon  \eta^2 ( 1 -(p_n+\delta_n) )^2 / \theta^2$. Upon substituting into the exponent in the previous equation, this exponent becomes
\begin{align}
& - \tilde{\lambda} \epsilon + \frac{\tilde{\lambda}^2}{8 T_n \eta^2 ( 1 -(p_n+\delta_n) )^2} + \frac{\tilde{\lambda} c}{T_n} = -2 T_n \epsilon^2 \eta^2 ( 1 -(p_n+\delta_n) )^2 / \theta^2 + 4  c \epsilon \eta^2 ( 1 -(p_n+\delta_n) )^2 / \theta^2 \\
& \quad = - 2 T_n \epsilon^2  \eta^2 \left( ( 1- p_n )^2 - 2(1-p_n) \delta_n + \delta_n^2 \right) / \theta^2 + 4 c \epsilon \eta^2 ( 1 -(p_n+\delta_n) )^2 / \theta^2 \\
& \quad = - 2 T_n \epsilon^2  \eta^2 ( 1- p_n )^2 / \theta^2 + 2 T_n \epsilon^2  \eta^2 \delta_n ( 2(1-p_n) - \delta_n ) / \theta^2 + 4 c \epsilon \eta^2 ( 1 -(p_n+\delta_n) )^2 / \theta^2 \\
& \quad \leq - 2 T_n \epsilon^2  \eta^2 ( 1- p_n )^2 / \theta^2 + 4 T_n \epsilon^2  \eta^2 \delta_n / \theta^2 + 4 c \epsilon \eta^2 / \theta^2  \label{eqBoundingExponent} ,
\end{align}
where the inequality simply uses $p_n , \delta_n > 0$ and $p_n+\delta_n \in (0,1)$ (for large $n$). Now note that since $p_n \rightarrow p$, we have (for example) $(1-p_n)^2 > (1-p)^2/2$ for $n$ sufficiently large. Additionally, since $\delta_n = o ( 1 / T_n )$, we have (for example) $T_n \delta_n < c / \epsilon$ for $n$ sufficiently large. Combining these observations, we can upper bound \eqref{eqBoundingExponent} as
\begin{equation}
- 2 \epsilon^2 T_n \eta^2 ( 1- p_n )^2 / \theta^2 + 4 \eta^2 \epsilon^2 T_n \delta_n / \theta^2 + 4 \eta^2 \epsilon c / \theta^2 \leq - (  \epsilon \eta (1-p) )^2 T_n / \theta^2 + 8 c \epsilon \eta^2 / \theta^2 .
\end{equation}
Hence, substituting into \eqref{eqAfterUglySumGone} gives
\begin{equation} \label{eqFinalStructuralOnDegSeq}
\P_n ( \E [ \hat{\vartheta}_{T_n}(\phi) | \T ] > \epsilon  ) 1(\Omega_{n,2}) \leq \exp ( 8 c \epsilon \eta^2 / \theta^2 ) \exp( - ( \epsilon \eta (1-p) / \theta )^2 T_n ) 1 ( \Omega_{n,2}) .
\end{equation}
Finally, we write
\begin{align}
\P ( \E [ \hat{\vartheta}_{T_n}(\phi) | \T ] > \epsilon  ) & = \E [ \P_n ( \E [ \hat{\vartheta}_{T_n}(\phi) | \T ] > \epsilon ) 1(\Omega_{n,2}) + \P_n ( \E [ \hat{\vartheta}_{T_n}(\phi) | \T ] > \epsilon ) 1(\Omega_{n,2}^C)  ] \\
& \leq O \left( e^{ - ( \epsilon \eta (1-p) / \theta )^2 T_n }  \right) + \P ( \Omega_{n,2}^C ) = O \left( e^{ - ( \epsilon \eta (1-p) / \theta )^2 \mu \log n } + n^{-\kappa} \right) ,
\end{align}
where the first equality is the law of total expectation, the inequality uses \eqref{eqFinalStructuralOnDegSeq} and upper bounds a probability by 1, and the second equality uses the assumptions in the statement of the lemma.

\subsubsection{Where the proof fails in the general case} \label{appWhereFailWhenPnTo1}

As shown in Appendix \ref{appExtendSecThmToOthers}, extending Theorem \ref{thmSec} to the case $p_n \rightarrow 1$ amounts to showing that for some $\gamma' > 0$,
\begin{equation} \label{eqSuffCond3}
\P( | \E [ \hat{\vartheta}_{T_n}(\phi) | \T ] - L(p_n) |  > \epsilon) = O \left( n^{-\gamma'} \right) ,
\end{equation}
where $L(p_n)$ is the appropriate limit from \eqref{eqLofPn}. Here we show (roughly) why the approach from the preceding proof fails to establish \eqref{eqSuffCond3} in the case $p_n \rightarrow 1$. To begin, we note we first used the assumption $p_n \rightarrow p < 1$ following \eqref{eqWhereItFails}. Hence, in the case $p_n \rightarrow 1$, we can still follow the approach leading to \eqref{eqWhereItFails} to obtain the (one-sided) bound
\begin{align}
& \P( \E [ \hat{\vartheta}_{T_n}(\phi) | \T ] - L(p_n)  > \epsilon) 1(\Omega_{n,2}) \leq \exp ( - \tilde{\lambda} ( \epsilon + L(p_n) ) \E \exp ( \tilde{\lambda} \E [ \hat{\vartheta}_{T_n}(\phi) | \T ] ) 1(\Omega_{n,2})  \\
& \quad \leq \exp \left( - \tilde{\lambda} \epsilon + \tilde{\lambda} \left( - L(p_n)+ \E_n [ \hat{\vartheta}_{T_n}(\phi) ] \right) + \frac{\tilde{\lambda}^2 \theta^2}{8 T_n^2} \sum_{l=1}^{T_n-1} \left( \sum_{l'=l}^{T_n-1} u_{T_n,l'} \tilde{p}_n^{l'-l} \right)^2 \right) 1(\Omega_{n,2}) \\
& \quad \approx \exp \left( - \tilde{\lambda} \epsilon + \frac{\tilde{\lambda}^2 \theta^2}{8 T_n^2} \sum_{l=1}^{T_n-1} \left( \sum_{l'=l}^{T_n-1} u_{T_n,l'} \tilde{p}_n^{l'-l} \right)^2 \right) 1(\Omega_{n,2}) \label{eqFirstBoundForFail} ,
\end{align}
where the approximate equality uses $\E_n [ \hat{\vartheta}_{T_n}(\phi) ] \approx L(p_n)$ on $\Omega_{n,2}$ by Lemma \ref{lemFirstMomBelief}. We next note
\begin{align}
& \sum_{l=1}^{T_n-1} \left( \sum_{l'=l}^{T_n-1} u_{T_n,l'} \tilde{p}_n^{l'-l} \right)^2 \geq \left( \sum_{l'=1}^{T_n-1} u_{T_n,l'} \tilde{p}_n^{l'-1} \right)^2  = \left( \sum_{l'=1}^{T_n-1} \left( \sum_{t=l'}^{T_n-1} {t \choose l'} \eta^{l'} (1-\eta)^{t-l'} \right) \tilde{p}_n^{l'-1} \right)^2 \\
& \quad = ( \tilde{p}_n^* )^{-2} \left( \sum_{t=1}^{T_n-1} \sum_{l'=1}^t  {t \choose l'} \eta^{l'} (1-\eta)^{t-l'} \tilde{p}_n^*  \tilde{p}_n^{l'-1} \right)^2 = ( \tilde{p}_n^* )^{-2} \left( \frac{T_n}{\theta} \E_n [ \hat{\vartheta}_{T_n}(\phi) ] - \frac{1-(1-\eta)^{T_n}}{\eta} \right)^2 ,
\end{align}
where the inequality discards nonnegative terms, the first equality is by definition of $u_{T_n,l'}$, the second rearranges summations and multiplies/divides by $(\tilde{p}_n^* )^2$,  and the third uses \eqref{eqTreeBeliefMeanWalkSub}. Hence, we have shown \eqref{eqFirstBoundForFail} is (roughly) lower bounded by
\begin{equation}
\exp \left( - \tilde{\lambda} \epsilon + \frac{\tilde{\lambda}^2}{8} \left( \E_n [ \hat{\vartheta}_{T_n}(\phi) ] - \frac{\theta (  1-(1-\eta)^{T_n} )}{T_n \eta} \right)^2 \right) 1(\Omega_{n,2}) ,
\end{equation}
where we have also used $\tilde{p}_n^* \approx 1$ for large $n$ on $\Omega_{n,2}$ when $p_n \rightarrow 1$ by \ref{assBranchDegSeq}. Now we consider three cases for the exponent in the previous expression:
\begin{itemize}
\item $T_n(1-p_n) \rightarrow 0$: Here Lemma \ref{lemFirstMomBelief} states $\E_n [ \hat{\vartheta}_{T_n}(\phi) ]  \approx \theta$ for large $n$ on $\Omega_{n,2}$; for such $n$, the exponent is roughly
\begin{equation}
- \tilde{\lambda} \epsilon + \frac{\tilde{\lambda}^2 \theta^2}{8} \left( 1 - \frac{\theta (  1-(1-\eta)^{T_n} )}{T_n \eta} \right)^2 \geq - \tilde{\lambda} \epsilon + \frac{\tilde{\lambda}^2 \theta^2}{16} = - \frac{4 \epsilon^2}{\theta^2} ,
\end{equation}
where the inequality holds for large $n$ (so that $\theta (  1-(1-\eta)^{T_n} ) / ( T_n \eta ) < 1 - 1 / \sqrt{2}$, which holds since $T_n \rightarrow \infty$) and the equality holds by choosing the minimizing $\tilde{\lambda}$ (namely, $\tilde{\lambda} = 8 \epsilon / \theta^2$). Since this lower bound is constant in $n$, \eqref{eqFirstBoundForFail} does not decay as $n$ grows.
\item $T_n(1-p_n) \rightarrow c \in (0,\infty)$: Here Lemma \ref{lemFirstMomBelief} states $\E_n [ \hat{\vartheta}_{T_n}(\phi) ]  \approx \theta (1-e^{-c\eta})/(c \eta)$ for large $n$ on $\Omega_{n,2}$. An argument similar to the previous case shows \eqref{eqFirstBoundForFail} does not decay as $n$ grows.
\item $T_n(1-p_n) \rightarrow \infty$ with $p_n \rightarrow 1$: Here we consider an example to show \eqref{eqFirstBoundForFail} does not decay sufficiently quickly for the general case. In particular, we assume $T_n = \bar{c} \log n$ for some constant $\bar{c}$ that satisfies the theorem assumptions and we set $p_n = 1 - ( \log n )^{-0.9}$. Then since $\delta_n =o ( (\log n)^{-1} )$ per \ref{assBranchDegSeq}, we have e.g.\ $1-p_n + \delta_n < (1-p_n)/2$ for large $n$. Hence,
\begin{align}
\E_n [ \hat{\vartheta}_{T_n}(\phi) ] & \gtrsim  \frac{\theta(1-(1-\eta(1-p_n+\delta_n))^{T_n})}{\eta T_n (1-p_n+\delta_n)} \\
& > \frac{\theta ( 1 - ( 1 - (\eta/2) ( \log n )^{-0.9} )^{\bar{c} \log n} )}{ (\bar{c} \eta/2) ( \log n )^{0.1} } > \frac{\tilde{c}}{ ( \log n )^{0.1} } ,
\end{align}
where the first inequality holds by \eqref{eqMeanBeforeSpeciazling} in Appendix \ref{appProofFirstMomBelief} (where $\gamma_1,\gamma_2$ are arbitrarily small, hence the approximate inequality), the second holds for our chosen $T_n, p_n, \delta_n$, and the third holds for some constant $\tilde{c}$ and for large $n$. Hence, the exponent is (roughly) lower bounded by
\begin{equation}
- \tilde{\lambda} \epsilon + \frac{\tilde{\lambda}^2}{8} \frac{\tilde{c}^2}{ ( \log n )^{0.2} } = - \frac{2 \epsilon^2}{\tilde{c}^2} ( \log n )^{0.2} , 
\end{equation}
where the equality holds for the minimizer $\tilde{\lambda} = ( 4 \epsilon / \tilde{c}^2 ) ( \log n )^{0.2}$. From here it follows that \eqref{eqFirstBoundForFail} cannot be $O(n^{-\gamma'})$: if it is, we have for all large $n$ and for some constant $\tilde{C}$,
\begin{equation}
\exp \left( - \frac{2 \epsilon^2}{\tilde{c}^2} ( \log n )^{0.2} \right) < \tilde{C} n^{-\gamma'} \Rightarrow \exp \left( - \frac{2 \epsilon^2}{\tilde{c}^2} ( \log n )^{0.2} + \gamma' \log n  \right) < \tilde{C} .
\end{equation}
The final inequality is a contradiction, since $- (2 \epsilon^2 / \tilde{c}^2) ( \log n )^{0.2} + \gamma' \log n \rightarrow \infty$ as $n \rightarrow \infty$.
\end{itemize}

\subsection{Auxiliary results} \label{appAuxiliary}

In this appendix, we collect several auxiliary results used in other proofs. (These results are either cited from other sources, or their proofs are computationally heavy but elementary, so we collect them here to avoid cluttering other parts of our analysis.) 

\begin{lemma} \label{lemTnPnDnasymp}
For $T_n \rightarrow \infty$, $p_n \rightarrow 1$, and $\delta_n \rightarrow 0$ s.t.\ $\delta_n = o(1/T_n)$, we have
\begin{align}
& \frac{ 1 - (1-\eta ( 1 - p_n - \delta_n ) )^{T_n} }{ \eta T_n  (1-p_n - \delta_n) }  \xrightarrow[n \rightarrow \infty]{} \begin{cases} 1 , & T_n(1-p_n) \xrightarrow[n \rightarrow \infty]{} 0 \\ ( 1 - e^{-c \eta} ) / (c \eta) , & T_n(1-p_n) \xrightarrow[n \rightarrow \infty]{}  c \in (0,\infty) \\ 0 , & T_n(1-p_n) \xrightarrow[n \rightarrow \infty]{}  \infty \end{cases} , \label{eqTnPnDnasympMin} \\
& \frac{ 1 - (1-\eta ( 1 - p_n + \delta_n ))^{T_n} }{ \eta T_n  (1-p_n + \delta_n) } \xrightarrow[n \rightarrow \infty]{} \begin{cases} 1 , & T_n(1-p_n) \xrightarrow[n \rightarrow \infty]{} 0 \\ ( 1 - e^{-c \eta} ) / (c \eta) , & T_n(1-p_n) \xrightarrow[n \rightarrow \infty]{}  c \in (0,\infty) \\ 0 , & T_n(1-p_n) \xrightarrow[n \rightarrow \infty]{}  \infty \end{cases}  \label{eqTnPnDnasympPlus} .
\end{align}
\end{lemma}
\begin{proof}
We consider the three cases of \eqref{eqTnPnDnasympMin} in turn; the proof of \eqref{eqTnPnDnasympPlus} follows the same approach.

First, suppose $\lim_{n \rightarrow \infty} T_n(1-p_n) = \infty$. Then since $T_n \delta_n \rightarrow 0$ and $T_n(1-p_n) \rightarrow \infty$, we have $T_n \delta_n < 1 < T_n(1-p_n)$ for sufficiently large $n$, which implies $(1-p_n-\delta_n) > 0$ for such $n$. Clearly, we also have $(1-p_n-\delta_n) < 1$ for all $n$. Taken together, it follows that $1 - (1-\eta ( 1 - p_n - \delta_n ) )^{T_n} \in (0,1)$ for $n$ large. For such $n$, we can then write
\begin{equation} \label{eqGeomTailLimitAdoptCase}
0 < \frac{ 1 - (1-\eta ( 1 - p_n - \delta_n ) )^{T_n} }{ \eta T_n  (1-p_n - \delta_n) } < \frac{ 1}{ \eta T_n  (1-p_n - \delta_n) } ,
\end{equation}
where we used $(1-p_n-\delta_n) > 0$ in the denominator. Now since $T_n(1-p_n) \rightarrow \infty$ and $T_n \delta_n \rightarrow 0$, $T_n  (1-p_n - \delta_n) \rightarrow \infty$, so taking $n \rightarrow \infty$ in the above inequality gives the result. 

Next, suppose $\lim_{n \rightarrow \infty} T_n(1-p_n) = c \in (0,\infty)$. Since $\eta T_n  (1-p_n - \delta_n) \rightarrow \eta c$ by $T_n(1-p_n) \rightarrow c$ and $T_n \delta_n \rightarrow 0$, it suffices to show $(1-\eta ( 1 - p_n - \delta_n ) )^{T_n} \rightarrow e^{-\eta c}$ as $n \rightarrow \infty$. First, since $T_n(1-p_n) \rightarrow c$, $\forall\ \epsilon_1 > 0\ \exists\ N_1$ s.t.\ $c-\epsilon_1 < T_n(1-p_n) <c+\epsilon_1\ \forall\ n \geq N_1$. Further, since $T_n \delta_n \rightarrow 0$, $\forall\ \epsilon_2 > 0\ \exists\ N_2$ s.t.\ $-\epsilon_2 < T_n \delta_n <\epsilon_2\ \forall\ n \geq N_2$. Hence, $\forall\ n \geq \max \{N_1,N_2\}$,
\begin{equation}
\left( 1 - \frac{\eta (c+\epsilon_1+\epsilon_2)}{T_n} \right)^{T_n} < (1-\eta ( 1 - p_n - \delta_n ) )^{T_n} < \left( 1 - \frac{\eta (c-\epsilon_1-\epsilon_2)}{T_n} \right)^{T_n} .
\end{equation}
Next, we note
\begin{equation}
\lim_{n \rightarrow \infty} \left( 1 - \frac{\eta (c+\epsilon_1+\epsilon_2)}{T_n} \right)^{T_n} = e^{-\eta(c+\epsilon_1+\epsilon_2)} ,  \lim_{n \rightarrow \infty} \left( 1 - \frac{\eta (c-\epsilon_1-\epsilon_2)}{T_n} \right)^{T_n} = e^{-\eta(c-\epsilon_1-\epsilon_2)} .
\end{equation}
Hence, $\forall\ \epsilon_3 > 0\ \exists\ N_3$ s.t.\ $\forall\ n \geq N_3$, 
\begin{equation}
e^{-\eta(c+\epsilon_1+\epsilon_2)} - \epsilon_3 < \left( 1 - \frac{\eta (c+\epsilon_1+\epsilon_2)}{T_n} \right)^{T_n}   ,  \left( 1 - \frac{\eta (c-\epsilon_1)}{T_n} \right)^{T_n} < e^{-\eta(c-\epsilon_1-\epsilon_2)} + \epsilon_3 .
\end{equation}
Combining these arguments, we obtain $\forall\ n \geq \max \{N_1,N_2,N_3\}$
\begin{equation}
e^{-\eta(c+\epsilon_1+\epsilon_2)} - \epsilon_3 < (1-\eta ( 1 - p_n - \delta_n ) )^{T_n} < e^{-\eta(c-\epsilon_1-\epsilon_2)} + \epsilon_3 .
\end{equation}
Since both bounds converge to $e^{-\eta c}$ as $\epsilon_1,\epsilon_2, \epsilon_3 \rightarrow 0$, $(1-\eta ( 1 - p_n - \delta_n  ) )^{T_n} \rightarrow e^{-\eta c}$ follows. 

Finally, suppose $\lim_{n \rightarrow \infty} T_n(1-p_n) = 0$. First, we observe
\begin{equation} \label{eqTnPnDnCase3_1}
\frac{ 1 - (1-\eta ( 1 - p_n - \delta_n ) )^{T_n} }{ \eta T_n  (1-p_n - \delta_n) } = \frac{1}{T_n} \sum_{t=0}^{T_n-1} (1-\eta ( 1 - p_n - \delta_n ) )^t \leq 1,
\end{equation}
where the inequality holds for $n$ s.t.\ $( 1 - p_n - \delta_n ) > 0$ (which indeed occurs for large $n$; see proof of $T_n(1-p_n) \rightarrow \infty$ case), since then the sum is over $T_n$ terms, each upper bounded by 1. On the other hand, we can use the binomial theorem to write
\begin{align} \label{eqTnPnDnCase3_2}
& \frac{ 1 - (1-\eta ( 1 - p_n - \delta_n ) )^{T_n} }{ \eta T_n  (1-p_n - \delta_n) } = \frac{ 1 - \sum_{t=0}^{T_n} {T_n \choose t} (-\eta(1-p_n-\delta_n))^t }{ \eta T_n  (1-p_n - \delta_n) } \\
& \quad = 1 -  \sum_{t=2}^{T_n} \frac{(T_n-1) \cdots (T_n-t+1) (-1)^t (\eta(1-p_n-\delta_n))^{t-1} }{t!}   .
\end{align}
Next, we observe (assuming $(1-p_n-\delta_n) > 0)$ as above)
\begin{align} \label{eqTnPnDnCase3_3}
& \sum_{t=2}^{T_n} \frac{(T_n-1) \cdots (T_n-t+1) (-1)^t (\eta(1-p_n-\delta_n))^{t-1} }{t!} \\
& \quad < \sum_{t=2}^{T_n} \frac{(T_n-1) \cdots (T_n-t+1)  (\eta(1-p_n-\delta_n))^{t-1} }{t!}  < \sum_{t=2}^{T_n}  \frac{(T_n (1-p_n-\delta_n))^{t-1} }{(t-2)!} \\
& \quad = T_n (1-p_n-\delta_n) \sum_{t=0}^{T_n-2} \frac{(T_n (1-p_n-\delta_n))^{t} }{t!}  <  T_n (1-p_n-\delta_n) e^{T_n (1-p_n-\delta_n)} , 
\end{align}
where the first inequality replaces negative terms with positive ones; the second inequality uses $\eta < 1$, $(t-2)! < t!$, and $(T_n-j) < T_n$ for $j > 0$; and the third inequality upper bounds the summation by replacing its upper limit with infinity. Hence, \eqref{eqTnPnDnCase3_1}, \eqref{eqTnPnDnCase3_2}, and \eqref{eqTnPnDnCase3_3} yield
\begin{align}
& 1 \geq \frac{ 1 - (1-\eta ( 1 - p_n - \delta_n ) )^{T_n} }{ \eta T_n  (1-p_n - \delta_n) } > 1 - T_n (1-p_n-\delta_n) e^{T_n (1-p_n-\delta_n)} \\
&  \Rightarrow 1 \geq \lim_{n \rightarrow \infty} \frac{ 1 - (1-\eta ( 1 - p_n - \delta_n ) )^{T_n} }{ \eta T_n  (1-p_n - \delta_n) } \geq 1 - \lim_{n \rightarrow \infty} T_n (1-p_n-\delta_n) e^{T_n (1-p_n-\delta_n)} = 1 ,
\end{align}
where the final equality holds since $T_n(1-p_n), T_n \delta_n \rightarrow 0$ by assumption.
\end{proof}

\begin{lemma} \label{lemUglySumInHoeffding}
Let $u_{T_n,l} = \sum_{t=l}^{T_n-1} { t \choose l } \eta^l (1-\eta)^{t-l}$. Then for any $x \in (0,1)$,
\begin{equation}
\sum_{l=1}^{T_n-1} \left( \sum_{l'=l}^{T_n-1} u_{T_n,l'} x^{l'-l} \right)^2 \leq  \frac{T_n}{\eta^2(1-x)^2}  .
\end{equation}
\end{lemma} 
\begin{proof}
For $l \in \N_0$, define $w_l = \sum_{l'=l}^{T_n-1} u_{T_n,l'} x^{l'-l}$. Then
\begin{equation}
w_l = u_{T_n,l} + x  \sum_{l'=l+1}^{T_n-1} u_{T_n,l'} x^{l'-(l+1)} = u_{T_n,l} + x w_{l+1} .
\end{equation}
Assuming temporarily that $u_{T_n,l'} \geq u_{T_n,l''}$ whenever $l' \leq l''$ (which we will return to prove),
\begin{equation}
w_{l+1} \leq u_{T_n,l} \sum_{l'=l+1}^{T_n-1} x^{l'-(l+1)} = u_{T_n,l} \sum_{l'=0}^{T_n-l-2}  x^{l'} \leq u_{T_n,l} \sum_{l'=0}^{\infty}  x^{l'} = \frac{u_{T_n,l}}{1-x} .
\end{equation}
Hence, using the previous two equations, we obtain $w_{l+1} - w_l = (1-x) w_{l+1} - u_{T_n-l} \leq 0$, i.e.\ the sequence $\{ w_l \}$ decreases in $l$. It is also clearly nonnegative. Therefore,
\begin{equation}
\sum_{l=1}^{T_n-1} \left( \sum_{l'=l}^{T_n-1} u_{T_n,l'} x^{l'-l} \right)^2 = \sum_{l=1}^{T_n-1} w_l^2 \leq T_n w_0^2 .
\end{equation}
To further bound the right hand side, we note
\begin{align}
w_0 & =  \sum_{l'=0}^{T_n-1} \left( \sum_{t=l'}^{T_n-1} {t \choose l'} \eta^{l'} (1-\eta)^{t-l'} \right) x^{l'} = \sum_{t=0}^{T_n-1} \sum_{l'=0}^t {t \choose l'} ( \eta x )^{l'} (1-\eta)^{t-l'} \\
& = \sum_{t=0}^{T_n-1} ( \eta x + (1-\eta) )^t = \sum_{t=0}^{T_n-1} ( 1 - \eta (1-x) )^t \leq \sum_{t=0}^{\infty} ( 1 - \eta (1-x) )^t = \frac{1}{\eta(1-x)} ,
\end{align}
where the first equality uses the definition of $u_{T_n,l'}$, the second rearranges summations, the third uses the binomial theorem, the fourth is immediate, the inequality is immediate, and the final equality computes a geometric series. Combining the previous two inequalities proves the lemma.

We return to prove $u_{T_n,l'} \geq u_{T_n,l''}$ whenever $l' \leq l''$. For this, we first claim
\begin{equation} \label{eqStrongerUmono}
\sum_{t=l}^{t^*} {t \choose l} \eta^l (1-\eta)^{t-l} - \sum_{t=l+1}^{t^*+1} {t \choose l+1} \eta^{l+1} (1-\eta)^{t-(l+1)} = {t^*+1 \choose l+1} \eta^l (1-\eta)^{t^*+1-l}\ \forall\ t^* \in \N, l \in \{1,\ldots,t^*\} .
\end{equation}
We prove \eqref{eqStrongerUmono} by induction on $t^*$. First, when $t^* = 1$, the only case to prove is $l=1$; when $t^* = l = 1$, it is immediate that both sides of \eqref{eqStrongerUmono} equal $\eta(1-\eta)$. Next, assume \eqref{eqStrongerUmono} holds for $t^*-1$. If $l = t^*$, both sides of \eqref{eqStrongerUmono} equal $\eta^{t^*}(1-\eta)$. If $l \in \{1,\ldots,t^*-1\}$, we write
\begin{align}
& \sum_{t=l}^{t^*} {t \choose l} \eta^l (1-\eta)^{t-l} - \sum_{t=l+1}^{t^*+1} {t \choose l+1} \eta^{l+1} (1-\eta)^{t-(l+1)}  \\
& \quad = \left( \sum_{t=l}^{t^*-1} {t \choose l} \eta^l (1-\eta)^{t-l} - \sum_{t=l+1}^{t^*} {t \choose l+1} \eta^{l+1} (1-\eta)^{t-(l+1)}  \right) \\
& \quad\quad +  {t^* \choose l} \eta^l (1-\eta)^{t^*-l} - {t^*+1 \choose l+1} \eta^{l+1} (1-\eta)^{t^*+1-(l+1)}  \\
& \quad = \left( {t^* \choose l+1} \eta^l (1-\eta)^{t^*-l} \right) + {t^* \choose l} \eta^l (1-\eta)^{t^*-l} - {t^*+1 \choose l+1} \eta^{l+1} (1-\eta)^{t^*-l} \\
& \quad = \eta^l (1-\eta)^{t^*-l} \left( {t^* \choose l+1} + {t^* \choose l} - \eta {t^*+1 \choose l+1} \right) \\
& \quad = \eta^l (1-\eta)^{t^*-l} \left( {t^*+1 \choose l+1} - \eta {t^*+1 \choose l+1} \right) = {t^*+1 \choose l+1} \eta^l ( 1- \eta )^{t^*+1-l} ,
\end{align}
where the first equality simply writes the final summands separately, the second uses the inductive hypothesis on the term in parentheses, the third is immediate, the fourth uses Pascal's rule ($[t^*+1]$ has ${t^*+1 \choose l+1}$ subsets of cardinality $l+1$; ${t^* \choose l}$ that contain 1 and ${t^* \choose l+1}$ that do not contain 1), and the fifth is immediate. This establishes \eqref{eqStrongerUmono}. We then write
\begin{align}
& u_{T_n,l'} - u_{T_n,l'+1} = \sum_{t=l'}^{T_n-1} {t \choose l'} \eta^{l'} (1-\eta)^{t-l'} - \sum_{t=l'+1}^{T_n-1} {t \choose l'+1} \eta^{l'+1} (1-\eta)^{t-(l'+1)} \\
& \quad = \sum_{t=l'}^{T_n-1} {t \choose l'} \eta^{l'} (1-\eta)^{t-l'} - \sum_{t=l'+1}^{T_n} {t \choose l'+1} \eta^{l'+1} (1-\eta)^{t-(l'+1)}+ {T_n \choose l'+1} \eta^{l'+1} (1-\eta)^{T_n-(l'+1)} \\
& \quad = { T_n \choose l'+1 } \eta^{l'} (1-\eta)^{T_n-l'} + {T_n \choose l'+1} \eta^{l'+1} (1-\eta)^{T_n-(l'+1)}  = { T_n \choose l'+1 } \eta^{l'} (1-\eta)^{T_n-(l'+1)} \geq 0 ,
\end{align}
where the first equality holds by definition of $u_{T_n,l'}$, the second adds and subtracts a term, and the third uses \eqref{eqStrongerUmono}. This shows $u_{T_n,l'} \geq u_{T_n,l'+1}$, iterating gives $u_{T_n,l'} \geq u_{T_n,l''}$ whenever $l' \leq l''$.
\end{proof}

\begin{lemma} \label{lemHoeffding}
Let $Z$ be a random variable satisfying $\E Z = 0$ and $Z \in [a,b]\ a.s.$, and let $\lambda > 0$. Then
\begin{equation}
\E e^{\lambda Z} \leq e^{\lambda^2 (b-a)^2 / 8} .
\end{equation}
\end{lemma}
\begin{proof}
See e.g.\ \cite[Lemma 5.1]{dubhashi2009concentration}.
\end{proof}

%% file: proofBeliefs.tex
\section{Belief convergence results} \label{appProofBeliefs}

In this appendix, we prove two belief convergence results. We first discuss some basic tools used in both proofs. To begin, fix $p \geq 1$ and $y \in [0,1]$, and let $\mu$ be a probability measure over $[0,1]$. Then clearly,
\begin{equation} \label{eqWassDirac}
W_p(\mu,\delta_y)^p = \inf_{ (X,Y) : X \sim \mu , Y \sim \delta_y } \E | X - Y |^p = \E | X - y |^p .
\end{equation}
where $X$ is distributed as $\mu$ in the final expression. Also, for any $\upsilon > 0$, since $|X-y| \leq 1$, we clearly have
\begin{align} \label{eqExpXminyP}
\E | X - y |^p = \E [ | X - y |^p 1 ( | X - y |^p \leq \upsilon ) ] + \E [ | X - y |^p 1 ( | X - y |^p > \upsilon ) ]  \leq \upsilon + \P (  | X - y |^p > \upsilon ) . 
\end{align}
Furthermore, for any $z \in [0,1]$ such that $|y-z|^p \leq 2^{-p} \upsilon$, by convexity,
\begin{equation}
| X - y |^p = 2^p \left| \frac{X-z}{2} + \frac{z-y}{2} \right|^p \leq 2^p \left( \frac{|X-z|^p}{2} + \frac{|z-y|^p}{2} \right) = 2^{p-1} ( |X-z|^p + |z-y|^p ) \leq 2^{p-1} |X-z|^p + \frac{\upsilon}{2} ,
\end{equation}
which implies
\begin{equation}
\P (  | X - y |^p > \upsilon ) \leq \P ( 2^{p-1} |X-z|^p + \upsilon / 2 > \upsilon ) = \P ( |X-z|^p > 2^{-p} \upsilon ) .
\end{equation}
Combined with \eqref{eqWassDirac} and \eqref{eqExpXminyP}, we obtain
\begin{equation} \label{eqWpBound}
W_p(\mu,\delta_y)^p \leq \upsilon + \P ( | X-z |^p > 2^{-p} \upsilon )\ \forall\ y,z \in [0,1]\ s.t.\ |y-z|^p \leq 2^{-p} \upsilon .
\end{equation}

Next, we recall some notation and basic results from Appendix \ref{appBranchApproxProofOutline} and \ref{appProofGraphBeliefApprox}. 
Denote by $\alpha_t$ and $\beta_t$ the parameters $\{ \alpha_t(i) \}_{i \in A \cup B}$ and $\{ \beta_t(i) \}_{i \in A \cup B}$ in vector form. Set $Q = (1-\eta ) I + \eta P$, where $P$ is the graph's column-normalized adjacency matrix (normalized so each column sums to $1$). Let $\mathbf{1}$ be the all ones vector. Then (see \eqref{eqGraphBeliefMatrixForm} in Appendix \ref{appProofGraphBeliefApprox})
\begin{gather} \label{eqRecallGraphBeliefMatrixForm}
\alpha_t = (1-\eta) \sum_{\tau=1}^{t} s_{\tau} Q^{t-\tau}  + \alpha_0 Q^t  , \quad \beta_t = (1-\eta) \sum_{\tau=1}^{t} (\mathbf{1}-s_{\tau}) Q^{t-\tau}  + \beta_0 Q^t \quad \forall\ t \in \N .
\end{gather}
Hence, by column stochasticity of $Q$, we obtain the following componentwise inequality:
\begin{equation} \label{eqAlphaPlusBeta}
\alpha_t + \beta_t = (1-\eta) \sum_{\tau=1}^{t} \mathbf{1} Q^{t-\tau} + ( \alpha_0 + \beta_0 ) Q^t = (1-\eta) t \mathbf{1}  + ( \alpha_0 + \beta_0 ) Q^t \geq (1-\eta) t \mathbf{1} .
\end{equation}

Finally, for any $\alpha,\beta \in (0,\infty)$, we let $X(\alpha,\beta)$ denote a $\text{Beta}(\alpha,\beta)$ random variable. Thus, recalling the expressions for the mean and variance of the beta distribution, for any $\upsilon > 0$, Chebyshev's inequality implies
\begin{equation}
\P \left( \left| X(\alpha,\beta) - \frac{ \alpha  }{ \alpha + \beta } \right| > 2^{-p} \upsilon \right) = \P ( | X(\alpha,\beta) - \E  X(\alpha,\beta) | > 2^{-p} \upsilon ) \leq \frac{ \text{Var}(X(\alpha,\beta)) }{ (2^{-p} \upsilon)^2 } = \frac{ 2^{2p} \alpha \beta }{ \upsilon^2 ( \alpha + \beta )^2 ( \alpha + \beta + 1 )^2 } .
\end{equation}
In particular, for any $i \in [n]$ and $t \in \N$, since $\theta_t(i) = \frac{\alpha_t(i)}{  \alpha_t(i) + \beta_t(i) }$ by definition, the previous two inequalities imply
\begin{equation} \label{eqChebshevCond}
\P ( | X(\alpha_t(i),\beta_t(i)) - \theta_t(i)| > 2^{-p} \upsilon | \alpha_t(i),\beta_t(i) )  \leq \frac{ 2^{2p} \alpha_t(i) \beta_t(i) }{ \upsilon^2 ( \alpha_t(i) + \beta_t(i)  )^2 ( \alpha_t(i) + \beta_t(i) +1 ) } \leq \frac{2^{2p}}{ \upsilon^2 ( 1- \eta )  t } .
\end{equation}
Hence, because $\mu_t(i)$ is the $\text{Beta}(\alpha_t(i),\beta_t(i))$ distribution, we can use \eqref{eqWpBound} to obtain the following:
\begin{equation} \label{eqApplyWpBound}
W_p ( \mu_t(i) , \delta_y )^p \leq \upsilon + \frac{2^{2p}}{ \upsilon^2 (1-\eta) t }\ \forall\ y \in [0,1]\ s.t.\ | y - \theta_t(i) | \leq 2^{-p} \upsilon .
\end{equation}

\subsection{Proof of Proposition \ref{propBeliefToZero}}

Fix $i \in A$. We first show $\theta_t(i) \rightarrow 0$ Let $e_A = \sum_{i' \in A} e_{i'}$. Recall $s_\tau(i') \in \{0,1\}$ for $i' \in A$ and $s_\tau(i') = 0$ for $i' \in B$, so $s_\tau \leq e_A$ componentwise. Combined with \eqref{eqRecallGraphBeliefMatrixForm} and \eqref{eqAlphaPlusBeta}, we obtain
\begin{equation}
0 \leq \theta_t(i) = \frac{ \alpha_t(i) }{ \alpha_t(i) + \beta_t(i) } \leq \frac{ 1}{t} \sum_{\tau=1}^t e_A Q^{t-\tau} e_i^\trans + \frac{ \alpha_0 Q^t e_i^\trans  }{ (1-\eta) t } .
\end{equation}
Hence, because $\alpha_0 Q^t e_i^\trans$ is bounded independent of $t$, it suffices to show that for any $\epsilon > 0$ and all $t$ large,
\begin{equation}\label{eqAbsorbMeanGoal}
\epsilon > \frac{ 1}{t} \sum_{\tau=1}^t e_A Q^{t-\tau} e_i^\trans = \frac{ 1}{t} \sum_{\tau=0}^{t-1} e_A Q^\tau e_i^\trans .
\end{equation}
Toward this end, we first observe that $Q(i',i') = 1$ for any $i' \in B$, i.e., $B$ is a set of absorbing states in the Markov chain with transition matrix $Q$. By the assumption of the proposition, all agents can reach this set. Taken together, we have an absorbing Markov chain with absorbing states $B$ and non-absorbing states $A$. It follows that $e_A Q^\tau e_i^\trans \rightarrow 0$ as $\tau \rightarrow \infty$. Hence, we can find $T_\epsilon$ such that $e_A Q^\tau e_i^\trans  < \frac{\epsilon}{2}$ whenever $\tau \geq T_\epsilon$. Thus, for any $t \geq 2 T_\epsilon / \epsilon$, we obtain the desired inequality \eqref{eqAbsorbMeanGoal}:
\begin{equation} 
\frac{ 1}{t} \sum_{\tau=0}^{t-1} e_A Q^\tau e_i^\trans = \frac{ 1}{t} \sum_{\tau=0}^{  T_\epsilon -1} e_A Q^\tau e_i^\trans + \frac{ 1}{t} \sum_{\tau= T_\epsilon }^{t-1} e_A Q^\tau e_i^\trans < \frac{ T_\epsilon }{ t } + \frac{ \epsilon }{ 2 } \leq \epsilon .
\end{equation}

Next, we show $W_p ( \mu_t(i) , \delta_0 ) \rightarrow 0$ Fix $\epsilon > 0$. Since $\theta_t(i) \rightarrow 0$, we can find $T(\epsilon)$ such that
\begin{equation} \label{eqTepsPrime}
|\theta_t(i)| < 2^{-(p+1)} \epsilon^p , \quad \frac{ 2^{2p} }{ ( \epsilon^p / 2 )^2 (1-\eta ) t } < \frac{\epsilon^p}{2}\ \forall\ t \geq T(\epsilon) .
\end{equation}
Hence, if we let $y = 0$ and $\upsilon = \epsilon^p/2$, then $|y - \theta_t(i)| = |\theta_t(i)| < 2^{-(p+1)} \epsilon^p = 2^{-p} \upsilon$, so by \eqref{eqApplyWpBound} and \eqref{eqTepsPrime}, for any $t \geq T(\epsilon)$,
\begin{equation}
W_p ( \mu_t(i) , \delta_0 )^p \leq \frac{ \epsilon^p }{ 2 } + \frac{2^{2p}}{ ( \epsilon^p / 2 )^2 (1-\eta) t } < \epsilon^p \quad \Rightarrow \quad W_p ( \mu_t(i) , \delta_0 ) < \epsilon .
\end{equation}

\subsection{Proof of Corollary \ref{corBeliefs}} 

Fix $p \geq 1$ and $\epsilon > 0$. Similar to the proof of Proposition \ref{propBeliefToZero}, we set $y = L(p_n)$ and $\upsilon = \epsilon^p / 2$. Then by \eqref{eqApplyWpBound},
\begin{equation}
| L(p_n) - \theta_{T_n}(i^*) | \leq 2^{-(p+1)} \epsilon^p = 2^{-p} \upsilon  \quad \Rightarrow \quad W_p ( \mu_{T_n}(i^*) , \delta_{L(p_n)} )^p \leq \frac{ \epsilon^p }{ 2 } + \frac{2^{2p}}{ ( \epsilon^p / 2 )^2 (1-\eta) T_n } .
\end{equation}
Hence, because $T_n \rightarrow \infty$ as $n \rightarrow \infty$ by \ref{assBranchHorizon}, we conclude that for all $n$ sufficiently large,
\begin{equation}
\left\{ W_p ( \mu_{T_n}(i^*) , \delta_{L(p_n)} )^p > \epsilon^p , | L(p_n) - \theta_{T_n}(i^*) | \leq 2^{-(p+1)} \epsilon^p \right\} \subset \left\{ \frac{ \epsilon^p }{ 2 } + \frac{2^{2p}}{ ( \epsilon^p / 2 )^2 (1-\eta) T_n } > \epsilon^p \right\} = \emptyset .
\end{equation}
Combined with Theorem \ref{thmMain}, we thus obtain
\begin{align}
\lim_{n \rightarrow \infty} \P ( W_p ( \mu_{T_n}(i^*) , \delta_{L(p_n)} ) > \epsilon ) & = \lim_{n \rightarrow \infty} \P ( W_p ( \mu_{T_n}(i^*) , \delta_{L(p_n)} )^p > \epsilon^p , | L(p_n) - \theta_{T_n}(i^*) | > 2^{-(p+1)} \epsilon^p ) \\
& \quad\quad + \lim_{n \rightarrow \infty} \P ( W_p ( \mu_{T_n}(i^*) , \delta_{L(p_n)} )^p > \epsilon^p , | L(p_n) - \theta_{T_n}(i^*) | \leq 2^{-(p+1)} \epsilon^p ) \\
& \leq \lim_{n \rightarrow \infty} \P \left( | L(p_n) - \theta_{T_n}(i^*)  | > 2^{-(p+1)} \epsilon^p \right)  = 0 .
\end{align}